\documentclass[11pt]{article}
\usepackage{setspace}
\usepackage{amsmath,amssymb}
\usepackage{amsthm}
\usepackage{algorithm, algpseudocode}
\usepackage{url}
\usepackage{fullpage}
\usepackage{color}
\usepackage{graphicx}
\usepackage{bbm}
\usepackage{enumerate}
\usepackage{hyperref}
\allowdisplaybreaks

\newtheorem{theorem}{Theorem}[section]
\newtheorem{lemma}[theorem]{Lemma}

\newtheorem{corollary}[theorem]{Corollary}
\theoremstyle{definition}
\newtheorem{definition}[theorem]{Definition}
\newtheorem{remark}[theorem]{Remark}
\newtheorem{example}[theorem]{Example}

\def\<{\langle}
\def\>{\rangle}
\def\Graph{{\sf G}}
\def\Ball{{\sf B}}
\def\sT{{\sf T}}
\def\dist{{\sf dist}}
\def\tnu{\tilde{\nu}}
\def\de{{\rm d}}

\def\ed{\stackrel{{\rm d}}{=}}
\def\lsdp{\lambda^{\mbox{\tiny\rm SDP}}}
\def\bsigma{{\boldsymbol \sigma}}
\def\btau{{\boldsymbol \tau}}
\def\bxi{{\boldsymbol \xi}}

\def\tF{\widetilde{F}}
\def\td{\tilde{d}}

\def\bnu{{\boldsymbol \nu}}
\def\bv{{\boldsymbol v}}
\def\bu{{\boldsymbol u}}

\def\bA{{\boldsymbol A}}
\def\bD{{\boldsymbol D}}
\def\bX{{\boldsymbol X}}
\def\bB{{\boldsymbol B}}
\def\bR{{\boldsymbol R}}
\def\bT{{\boldsymbol T}}
\def\bS{{\boldsymbol S}}
\def\bH{{\boldsymbol H}}
\def\btau{\pmb{\tau}}
\def\bz{{\boldsymbol z}}

\def\rd{{\rm d}}
\def\bLambda{{\boldsymbol \Lambda}}
\def\e{{\boldsymbol 1}}
\def\sG{{\sf G}}
\def\us{\underline{{\sf s}}}
\def\SDP{{\sf SDP}}
\def\PSD{{\sf PSD}}
\def\E{\mathbb{E}}
\def\P{\mathbb{P}}
\def\R{\mathbb{R}}
\def\reals{\mathbb{R}}
\def\cF{\mathcal{F}}
\def\cE{\mathcal{E}}
\def\cA{\mathcal{A}}
\def\cW{\mathcal{W}}
\def\cI{\mathcal{I}}
\def\cT{\mathcal{T}}
\def\cG{\mathcal{G}}
\def\cM{\mathcal{M}}
\def\root{{\o}}
\def\1{\mathbbm{1}}
\def\eps{\varepsilon}
\def\oE{E^{o}}
\def\ER{Erd\H{o}s-R\'enyi } 
\def\numcycle{\#_c}
\def\toloc{\Rightarrow}

\def\nue{\nu_e}
\DeclareMathOperator{\Var}{Var}
\DeclareMathOperator{\Tr}{Tr}
\DeclareMathOperator{\diag}{diag}
\DeclareMathOperator{\Pois}{Poisson}
\DeclareMathOperator{\Normal}{Normal}
\DeclareMathOperator{\Binom}{Binom}
\DeclareMathOperator{\sign}{sign}

\DeclareMathOperator{\OP}{\mathcal{O}}
\DeclareMathOperator{\Id}{Id}
\newcommand{\Ker}{\operatorname{Ker}}
\renewcommand{\Im}{\operatorname{Im}}
\def\cond{\mathrm{c}}
\def\bc{\overline{\mathrm{c}}}
\def\tcond{\tilde{\mathrm{c}}}
\def\true{\mathrm{true}}

\title{How Well Do Local Algorithms Solve Semidefinite Programs?}

\author{Zhou~Fan\footnote{Department of Statistics, Stanford
    University}
\and
Andrea~Montanari\footnote{Department of Electrical
    Engineering and Department of Statistics, Stanford University}
}

\begin{document}

\maketitle

\begin{abstract}
Several probabilistic models from high-dimensional statistics and machine learning reveal
an intriguing --and yet poorly understood-- dichotomy. Either simple local algorithms succeed in estimating
the object of interest, or even sophisticated semi-definite programming (SDP) relaxations fail.
 In order to explore this phenomenon, we study a classical SDP  relaxation of the 
minimum graph bisection problem, when applied to \ER random graphs with bounded average degree 
$d>1$, and obtain several types of results. First, we use a dual witness construction (using the so-called 
non-backtracking matrix of the graph) to upper bound the SDP value. Second, we prove that a 
simple local algorithm approximately solves the SDP to within a factor
$2d^2/(2d^2+d-1)$ of the upper bound. In particular,
the local algorithm is at most $8/9$ suboptimal, and $1+O(1/d)$ suboptimal for large degree.

We then analyze a more sophisticated local algorithm, which aggregates information according to the 
harmonic measure on the limiting Galton-Watson (GW) tree. The resulting lower bound is expressed in terms of the
conductance of the GW tree and matches surprisingly well the empirically
determined SDP values on large-scale \ER graphs.

We finally consider  the planted partition model. In this case, purely local algorithms are known to fail, but they do
succeed if a small amount of side information is available. Our results imply quantitative bounds on the threshold for partial recovery using SDP in this model.
\end{abstract}

\section{Introduction}

Semi-definite programming (SDP) relaxations are among the most powerful tools available to 
the algorithm designer. However, while efficient specialized solvers exist for several important applications
\cite{burer2003nonlinear,weinberger2008fast,nesterov2013first}, generic SDP
algorithms are not well suited for large-scale problems. 
At the other end of the spectrum, local algorithms attempt to solve
graph-structured problems by taking, at each vertex of the graph, a decision that is only based on a bounded-radius neighborhood 
of that vertex \cite{suomela2013survey}. As such, they can be implemented in linear time, or constant time on a distributed platform. 
On the flip side, their power is obviously limited.

Given these fundamental differences, it is surprising that these two classes of algorithms 
behave similarly on a  number of probabilistic models arising from statistics and machine learning.
Let us briefly review two well-studied  examples of this  phenomenon. 

\vspace{0.2cm}

In the (generalized) \emph{hidden clique problem}, a random graph $G$ over $n$
vertices is generated as follows:
A subset $S$ of $k$ vertices is chosen uniformly at random among all $\binom{n}{k}$ sets of that size.
Conditional on $S$, any two vertices $i$, $j$ are connected by an edge independently with probability 
$p$ if $\{i,j\}\subseteq S$ and probability $q<p$ otherwise. 
Given a single realization of this random graph $G$, we are requested to find the set $S$.
(The original formulation \cite{jerrum1992large} of the problem uses $p=1$, $q=1/2$
but it is useful to consider the case of general $p$, $q$.)

SDP relaxations for the hidden clique problem were studied in a number of papers, beginning with the 
seminal work of Feige and Krauthgamer \cite{feige2000finding,ames2011nuclear,meka2015sum,deshpande2015improved,barak2016nearly}. 
Remarkably, even the most powerful among these relaxations --which are constructed through the sum-of-squares (SOS) hierarchy--
fail unless $k\gtrsim \sqrt{n}$ \cite{barak2016nearly}, while exhaustive search succeeds with high probability
as soon as $k\ge C_0\log n$ for $C_0=C_0(p,q)$ a constant.
Local algorithms can be formally defined only for a sparse version of this model, whereby $p,q=\Theta(1/n)$
and hence each node has bounded average degree \cite{montanari2015finding}.
In this regime, there exists an optimal local algorithm for this problem that is related to the `belief propagation'
heuristic in graphical models. The very same algorithm can be applied to dense graphs (i.e. 
$p,q=\Theta(1)$),  and was proven to succeed if and only if $k\ge C_1\sqrt{n}$ \cite{deshpande2015finding}. 
Summarizing, the full power of the SOS hierarchy, despite having a much larger
computational burden, does not qualitatively improve upon the
performance of simple local heuristics.

\vspace{0.2cm}

As a second example, we consider  the \emph{two-groups symmetric stochastic block model} (also known as the planted partition problem)
that has attracted considerable attention in recent years as a toy model for community detection in networks 
\cite{decelle2011asymptotic,krzakala2013spectral,mossel2013proof,massoulie2014community,bordenaveetal,guedon2015community}.
A random graph $G$ over $n$ vertices is generated by partitioning 
the vertex set into two subsets\footnote{To avoid notational nuisances, we assume $n$ even.} 
$S_{+}\cup S_{-}$ of size $n/2$ uniformly at random. Conditional on this
partition, any two vertices $i$, $j$ 
are connected by an edge independently with probability 
$a/n$ if $\{i,j\}\subseteq S_{+}$ or $\{i,j\}\subseteq S_{-}$ (the two vertices are on the same side of the partition),
and with probability $b/n$ otherwise (the two vertices are on different sides).
Given a single realization of the random graph $G$, we are requested to identify the partition.

While several `success' criteria have been studied for this model, for the sake of simplicity we will focus 
on weak recovery (also referred to as `detection' or `partial recovery'). Namely, we want to attribute 
$\{+,-\}$ labels to the vertices so that --with high probability-- at least $(1/2+\eps)n$ vertices are labeled correctly 
(up to a global sign flip that cannot be identified). It was conjectured in \cite{decelle2011asymptotic} that this is possible if
and only if $\lambda >1$, where $\lambda\equiv (a-b)/\sqrt{2(a+b)}$  is an
effective `signal-to-noise ratio' parameter.
This conjecture followed from the heuristic analysis of a local algorithm based --once again-- on belief-propagation.
The conjecture was subsequently proven in \cite{mossel2012stochastic,mossel2013proof,massoulie2014community} through the analysis of 
carefully constructed spectral algorithms.
While these algorithms are, strictly speaking,  not local, they are related to the linearization of belief propagation around a 
`non-informative fixed point'.

Convex optimization approaches for this problem are based on the classical SDP relaxation of the 
minimum-bisection problem. Denoting by $\bA = \bA_G$ the adjacency matrix of $G$, the minimum bisection 
problem is written as
\begin{align}
\mbox{\rm maximize} &\;\;\; \<\bsigma,\bA \bsigma\>\, ,\\
\mbox{\rm subject to} &\;\;\; \bsigma\in\{+1,-1\}^n\, , \;\; \<\bsigma,\e\>=0\, .
\end{align}
The following SDP relaxes the above problem, where $d=(a+b)/2$ is the
average degree:
\begin{align}
\mbox{\rm maximize} &\;\;\; \<\bA-\frac{d}{n}\e\e^{\sT},\bX \>\, ,\label{eq:SDP}\\
\mbox{\rm subject to}&\;\;\; \bX \succeq 0\, , \;\; \bX_{ii}=1 \;\;\forall i\nonumber\, .
\end{align}
(Here, the term $-(d/n)\e\e^{\sT}$ can be thought of as a relaxation of the hard  constraint $\<\bsigma,\e\>=0$.)
This SDP relaxation  has a weak recovery threshold $\lsdp$ that appears to be very close to the ideal one 
$\lambda=1$.
Namely, Gu\'edon and Vershynin \cite{guedon2015community} proved $\lsdp\le C$ for $C$ a universal constant, while \cite{montanari2016semidefinite}
established $\lsdp = 1+o_d(1)$ for large average degree $d$.

Summarizing, also for the planted partition problem local algorithms (belief propagation) and SDP relaxations
behave in strikingly similar ways\footnote{An important remark is that strictly
local algorithms are ineffective in the 
planted partition problem. This peculiarity is related to the symmetry of the model, and can be resolved in several ways,
 for instance by an oracle that reveals an arbitrarily small fraction of the true vertex labels, or running belief propagation a 
logarithmic (rather than constant) number of iterations. We refer to Section \ref{sec:SBM} for further discussion of this point.}. In addition to the above rigorous results, numerical evidence suggests
that the two thresholds are  very close for all degrees $d$, and that the reconstruction accuracy above these
thresholds is also  very similar \cite{javanmard2016phase}. 

\vspace{0.2cm}

The conjectural picture emerging from these and similar examples can be described as follows. For statistical inference problems
on sparse random graphs, SDP relaxations are no more powerful than local algorithms (eventually supplemented 
with a small amount of side information to break symmetries). On the other hand, any  information that is genuinely non-local
is not exploited even by sophisticated SDP hierarchies.  
Of course, formalizing this picture is  of utmost practical interest, since it would entail a dramatic 
simplification of  algorithmic options.

\begin{figure}[t!]
\begin{center}
\includegraphics[width=.75\textwidth]{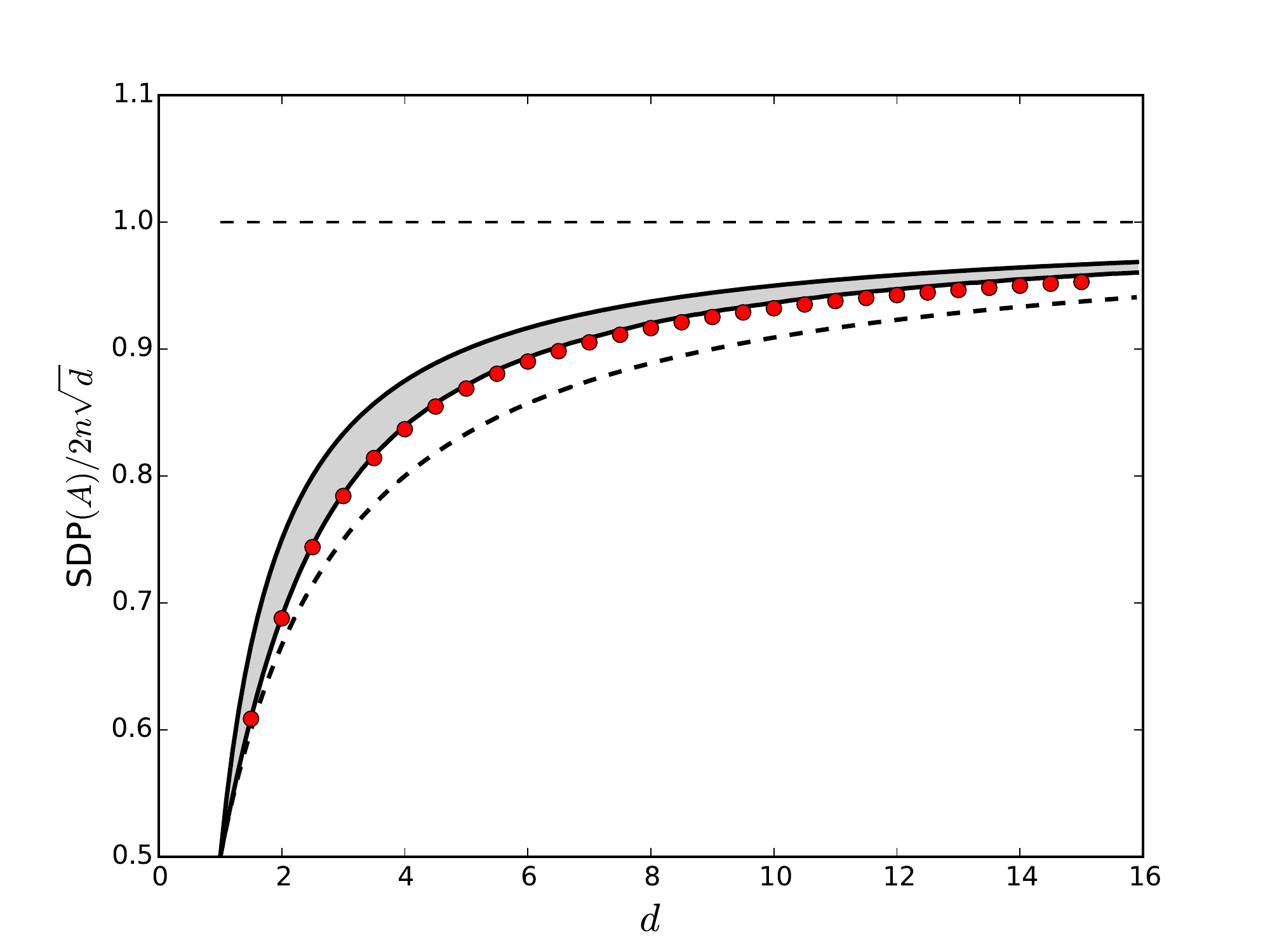}
\end{center}
\caption{Typical value $\SDP(\bA_G)$ of the min-bisection SDP for large \ER random graphs with average degree $d$,
normalized by the large degree formula $2n\sqrt{d}$. Circles: numerical simulations with graphs of size $n=10^6$. Solid lines:
Upper bound from Theorem \ref{thm:ER} and local algorithm lower bound (evaluated
numerically) from Theorem \ref{thm:Harmonic}. Lower dashed line: 
Explicit local-algorithm lower bound from Theorem \ref{thm:ER}. 
(The small inconsistency between numerical SDP values and the lower bound at large $d$ is due to non-asymptotic effects
that appear to vanish as $n\to\infty$.)}\label{fig:SDP_Value}
\end{figure}
With this general picture in mind, it is natural to ask:
\emph{Can semidefinite programs be (approximately) solved by local algorithms for a large class of random graph models}?
A positive answer to this question would clarify the equivalence between local algorithms and SDP relaxations.

Here, we address this problem by considering the  semidefinite program (\ref{eq:SDP}),
for two simple graph models, the \ER random graph with average degree $d$, $G\sim \Graph(n,d/n)$,
and the two-groups symmetric block model, $G\sim \Graph(n,a/n,b/n)$.
We establish the following results (denoting by $\SDP(\bA_G)$  the value of (\ref{eq:SDP})).
\begin{description}
\item[Approximation ratio of local algorithms.] We prove that there exists a simple local algorithm that 
approximates $\SDP(\bA_G)$  (when $G\sim\Graph(n,d/n)$) within a factor $2d^2/(2d^2+d-1)$, 
asymptotically for large $n$. In particular, the local algorithm is at most a
factor $8/9$ suboptimal, and $1+O(1/d)$ 
suboptimal for large degree.

Note that $\SDP(\bA_G)$ concentrates tightly around its expected value.  When we write that an algorithm \emph{approximates
$\SDP(\bA_G)$}, we mean that it returns a feasible solution whose value satisfies the claimed approximation guarantee.
\item[Typical SDP value.] Our proof provides upper and lower bounds on $\SDP(\bA_G)$ for
$G\sim\Graph(n,d/n)$, implying in particular $\SDP(\bA_G)/n = 2\sqrt{d}(1-\Theta(1/d)) +o_n(1)$ where
the term $\Theta(1/d)$ has explicit upper and lower bounds. While the lower bound is based on the 
analysis of a local algorithm, the upper bound follows from a dual witness construction which is of independent interest.

Our upper and lower bounds are plotted in Fig.\ \ref{fig:SDP_Value} together with the results of numerical simulations.
\item[A local algorithm based on harmonic measures.] The simple local algorithm
above uses randomness available at each vertex of 
$G$ and aggregates it uniformly within a neighborhood of each vertex. We analyze a different local algorithm that aggregates information in proportion to 
the harmonic measure of each vertex. We characterize the  value achieved by this
algorithm in the large $n$ limit in terms of the conductance of a random
Galton-Watson tree.  Numerical data (obtained by evaluating this value and also
solving the SDP (\ref{eq:SDP}) on large random graphs), as well as a large-$d$
asymptotic expansion, suggest that this lower
bound is very accurate, cf.\ Fig.\ \ref{fig:SDP_Value}.
\item[SDP detection threshold for the stochastic block model.] We then turn to the weak recovery problem in the two-group
symmetric stochastic block model $G\sim \Graph(n,a/n,b/n)$. As above, it is more convenient to parametrize this 
model by the average degree $d=(a+b)/2$ and the signal-to-noise ratio $\lambda = (a-b)/\sqrt{2(a+b)}$.
It was known from \cite{guedon2015community} that the threshold for SDP to achieve weak recovery is $\lsdp(d) \le 10^4$, and in
\cite{montanari2016semidefinite} that $\lsdp(d)\le 1+o_d(1)$ for large degree. Our results provide more precise information, implying in particular 
$\lsdp(d) \le \min(2-d^{-1}, 1+C\, d^{-1/4})$ for $C$ a universal constant.
\end{description}

\section{Main results}\label{sec:results}

In this section we recall the notion of local algorithms, as specialized to
solving the problem (\ref{eq:SDP}). We then state formally our main results. 
For general background on local algorithms, we refer to \cite{hatami2014limits,gamarnik2014limits,lyons2014factors}: this line of work is briefly 
discussed in Section \ref{sec:Related}.

Note that the application of local algorithms to solve SDPs is not entirely obvious,
since local algorithms are normally defined to return a quantity for each vertex in $G$,
instead of a matrix $\bX$ whose rows and columns are indexed by those vertices.
Our point of view will be that a local algorithm can solve the SDP (\ref{eq:SDP}) by returning, for each vertex $i$,
a random variable $\xi_i$, and the SDP solution $\bX$ associated to this local
algorithm is the covariance matrix of $\bxi=(\xi_1,\ldots,\xi_n)$ with respect
to the randomness of the algorithm execution.
An arbitrarily good approximation of this solution
$\bX$ can be obtained by repeatedly sampling the vector $\bxi\in \R^n$ (i.e.\ by
repeatedly running the algorithm with independent randomness).

Formally, let $\cG$ be the space of (finite or) locally finite rooted graphs,
i.e.\ of pairs $(G,\root)$ where $G=(V,E)$ is a locally finite graph and
$\root\in V$ is a distinguished root vertex. 
We denote by $\cG^*$ the space of tuples $(G,\root,\bz)$ where $(G,\root) \in \cG$ and
$\bz:V \to \R$ associates a real-valued mark to each vertex of $G$. 
Given a graph $G=(V,E)$ and a vertex $i\in V$, we denote by $\Ball_{\ell}(i;G)$ the subgraph induced by vertices
$j$ whose graph distance from $i$ is at most $\ell$, rooted at $i$. If $G$ carries marks $\bz:V\to\R$, it
is understood that $\Ball_{\ell}(i;G)$ inherits the `same' marks. We will write in this case 
$(\Ball_{\ell}(i;G),\bz)$ instead of the cumbersome (but more explicit) notation $(\Ball_{\ell}(i;G),i,\bz_{\Ball_{\ell}(i;G)})$.
\begin{definition}\label{def:Local}
A {\bf radius-$\pmb{\ell}$ local algorithm} for the semidefinite program (\ref{eq:SDP}) 
is any measurable function $F:\cG^* \to \R$ such that 
\begin{enumerate}
\item $F(G_1,\root_1,\bz_1)=
F(G_2,\root_2,\bz_2)$ if $(\Ball_\ell(\root_1;G_1),\bz_1) \simeq (\Ball_\ell(\root_2;G_2),\bz_2)$, where $\simeq$ denotes
graph isomorphism that preserves the root vertex and
 vertex marks. 
\item Letting $\bz = (z(i))_{i\in V}$ be i.i.d.\ with $z(i)\sim\Normal(0,1)$, we have 
$\E_{\bz}\big\{F(G,\root,\bz)^2\big\}=1$. (Here and below $\E_{\bz}$ denotes expectation with respect to the
random variables $\bz$).
\end{enumerate}
We denote the set of such functions $F$ by $\cF_*(\ell)$. 

A {\bf local algorithm} is a radius-$\ell$ local algorithm for some fixed $\ell$ (independent of the graph).
The set of such functions is denoted by $\cF_*\equiv \cup_{\ell\ge 1}\cF_*(\ell)$. 
\end{definition}
We apply a local algorithm to a fixed graph $G$ by generating
i.i.d.\ marks $\bz=(z(i))_{i \in V}$ as $z(i) \sim \Normal(0,1)$,
and producing the random variable $\xi_i=F(G,i,\bz)$ for each vertex $i \in V$. 
In other words, we use the radius-$\ell$ local algorithm to compute, for each
vertex of $G$, a function
of the ball of radius $\ell$ around that vertex that depends on the additional
randomness provided by the $z(i)$'s in this ball. The covariance matrix
$\bX=\E_{\bz}\{\bxi\bxi^{\sT}\}$ is a feasible point for the SDP
(\ref{eq:SDP}), achieving the value $n\,\cE(F;G)$ where
\begin{align}
\cE(F;G)\equiv \frac{1}{n}\E_{\bz}\left\{\sum_{i,j\in V}\Big((\bA_G)_{ij}-\frac{d}{n}\Big) F(G,i,\bz) F(G,j,\bz)\right\}\, .\label{eq:ValueDef}
\end{align}
We are now in position to state our main results.

\subsection{\ER random graphs $G\sim\Graph(n,d/n)$}
\label{sec:MainER}
We first prove an optimality guarantee, in the large $n$ limit, for the value
achieved by a simple local algorithm (or more precisely, a sequence of simple
local algorithms) solving (\ref{eq:SDP}) on the \ER graph.
\begin{theorem}\label{thm:ER}
Fix $d\ge 0$ and let $\bA=\bA_{G_n}$ be the adjacency matrix of the \ER random graph
$G_n \sim \sG(n,d/n)$. Then for $d< 1$, almost surely,
$\lim_{n\to\infty} \SDP(\bA)/n=d$. For $d\ge 1$, almost surely,
\begin{align}
2\sqrt{d}\left(1-\frac{1}{d+1}\right)
\leq \liminf_{n \to \infty} \frac{1}{n}\SDP(\bA)
\leq \limsup_{n \to \infty} \frac{1}{n}\SDP(\bA)
\leq 2\sqrt{d}\left(1-\frac{1}{2d} \right). \label{eq:UpperLowerMain}
\end{align}
Furthermore, there exists a sequence of local algorithms that achieves the lower
bound. Namely, for each $\eps>0$,
there exist $\ell(\eps)>0$ and $F\in\cF_*(\ell(\eps))$  such that, almost
surely,
\begin{align} 
\lim_{n\to\infty}\cE(F;G_n) \ge 2\sqrt{d}\left(1-\frac{1}{d+1}\right)-\eps\,
.\label{eq:locallowermain}
\end{align}
\end{theorem}
As anticipated in the introduction, the  upper and lower bounds of (\ref{eq:UpperLowerMain}) approach each other for large
$d$, implying in particular $\SDP(\bA)/n = 2\sqrt{d} \big(1-\Theta(1/d)\big) + o_n(1)$. This should be compared with 
the result of \cite{montanari2016semidefinite} yielding  $\SDP(\bA)/n = 2\sqrt{d} \big(1+o_d(1)\big) + o_n(1)$. 
Also, by simple calculus, the upper and lower bounds stay within a ratio bounded
by $8/9$ for all $d$, 
with the worst-case ratio $8/9$ being achieved at $d=2$. 
Finally, they again converge as $d\to 1$, implying in particular $\lim_{n\to\infty}\SDP(\bA)/n=1$ for $d=1$.
\begin{remark}
The result $\lim_{n\to\infty} \SDP(\bA)/n=d$ for $d<1$ is elementary and only stated for completeness. Indeed, for $d<1$,
the graph $G\sim \Graph(n,d/n)$ decomposes with high probability into
disconnected components of size $O(\log n)$, which are all trees or
unicyclic \cite{janson2011random}. As a consequence, the vertex set can be partitioned into two subsets of size $n/2$ so that
at most one connected component of $G$ has vertices on both sides of the
partition, and hence at most two edges cross the partition. 
By using the feasible point $\bX = \bsigma\bsigma^{\sT}$ with $\bsigma\in\{+1,-1\}^n$ the indicator vector of the partition,
we get $2|E|-8\le \SDP(\bA)\le 2|E|$ whence the claim follows immediately.

In the proof of Theorem \ref{thm:ER}, we will assume $d>1$. Note that the case $d=1$ follows as well,
since $\SDP(\bA)$ is a Lipschitz continuous function of $\bA$, with Lipschitz constant equal to one. 
This implies that $\limsup_{n\to\infty} \SDP(\bA)/n$, $\liminf_{n\to\infty} \SDP(\bA)/n$ are continuous functions of $d$.
\end{remark}

The local algorithm achieving the lower bound of Theorem \ref{thm:ER} is extremely simple. At each vertex $i\in V$,
it outputs a weighted sum of the random variables $z(j)$ with $j\in\Ball_{\ell}(i;G)$, with weight proportional to $d^{-\dist(i,j)/2}$
(here $\dist(i,j)$ is the graph distance between vertices $i$ and $j$). When applied to random $d$-regular graphs, 
this approach is related to the Gaussian wave function of \cite{csokaetal2014independent} and is known to achieve the 
SDP value $\SDP(\bA)$ in the large $n$ limit \cite{montanari2016semidefinite}.

\subsection{A local algorithm based on harmonic measures}\label{sec:harmonic}

A natural question arising from the previous section is whether a better local algorithm can be constructed 
by summing the random variables $z(j)$ with different weights, to account for the graph geometry. It turns 
out that indeed this is possible by using a certain harmonic weighting scheme that we next describe, deferring some technical 
details to Section \ref{sec:Lower}. Throughout we assume $d>1$. 

Recall that the random graph $G\sim\Graph(n,d/n)$ converges locally to a Galton-Watson tree (see Section \ref{sec:Lower} for
background on local weak convergence). This can be shown to imply that it is sufficient to define the function 
$F\in\cF_*$  for trees. Let $(T,\root)$ be an infinite rooted tree and consider the simple random walk on $T$ started at $\root$, 
which we assume to be transient. The harmonic measure assigns to vertex $v\in
V(T)$, with $\dist(\root,v)=k$,
 a weight $h^{(\root)}(v)$ which is the probability\footnote{For each distance
$k$, the weights $h^{(\root)}(v)$ 
form a probability distribution over vertices at distance $k$
from the root. These distributions can be derived from a unique probability measure over the boundary of $T$ at infinity, 
as is done in \cite{lyons1995ergodic}, but this is not necessary here.}  
that  the walk exits for the last time $\Ball_k(\root;T)$ at $v$ \cite{lyons1995ergodic}. 
We then define 
\begin{align}
\tF(T,\root,\bz)\equiv \frac{1}{\sqrt{\ell+1}}\sum_{v\in\Ball_{\ell}(\root,T)}
\sqrt{h^{(\root)}(v)} z(v)\, .\label{eq:Averaging}
\end{align}
Technically speaking, this is not a local function because the weights
$h^{(\root)}(v)$ depend on the whole tree $T$. 
However a local approximation to these weights can be constructed by truncating
$T$ at a depth $L \geq \ell$: details are provided in Section \ref{sec:Lower}.

Given the well-understood relationship between random walks and electrical networks,
it is not surprising that the value achieved by this local algorithm can be expressed in terms of conductance.
The conductance $\cond(T,\root)$ of a rooted tree $(T,\root)$ is the intensity of current flowing out of the root
when a unit potential difference is imposed between the root and the boundary (`at infinity'). It is understood that
$\cond(T,\root)=0$ if $T$ is finite. 
\begin{theorem}\label{thm:Harmonic}
For $(T,\root)$ a Galton-Watson tree with offspring distribution $\Pois(d)$, let $\cond_1,\cond_2\ed \cond(T,\root)$
be two independent and identically distributed copies of the conductance of $T$. Let
$\bA=\bA_{G_n}$ be the adjacency matrix of the \ER random graph $G_n \sim
\sG(n,d/n)$. Then for $d>1$, almost surely,
\begin{align}
\liminf_{n\to\infty}\frac{1}{n}\SDP(\bA)&\ge d\, \E\, \Psi(\cond_1,\cond_2)\,
,\label{eq:lowerharmonic}\\
\Psi(\cond_1,\cond_2) & \equiv \begin{cases}
\frac{\cond_1\sqrt{1+\cond_2}+\cond_2\sqrt{1+\cond_1}}{\cond_1+\cond_2+\cond_1\cond_2} &
\;\;\; \mbox{if $\cond_1>0$ or $\cond_2>0$,}\\
1  & \;\;\; \mbox{otherwise.}
\end{cases}\label{eq:Psi}
\end{align} 
Furthermore, for each $\eps>0$, there exist $\ell(\eps)>0$ and $F \in
\cF_*(\ell(\eps))$ such that, almost surely,
\begin{align}
\lim_{n \to \infty} \cE(F;\,G_n) \geq d\,\E\,\Psi(\cond_1,\cond_2)-\eps.
\label{eq:locallowerharmonic}
\end{align}
Finally, for large $d$, this lower bound behaves as
\begin{align}
d\, \E\, \Psi(\cond_1,\cond_2) = \sqrt{2d}\left(1-\frac{5}{8d} +O\Big(\Big(\frac{\log d}{d}\Big)^{3/2}\Big)\right)\,
.\label{eq:ExpansionCond}
\end{align}
\end{theorem}
The lower bound $d\,\E\, \Psi(\cond_1,\cond_2)$ is not explicit but can be efficiently evaluated numerically, by sampling the 
distributional recursion satisfied by $\cond$. This numerical technique was used in \cite{javanmard2016phase}, to which we refer for further details.
The result of such a numerical evaluation is plotted as the lower solid line in Figure \ref{fig:SDP_Value}.
This harmonic lower bound seems to capture extremely well our numerical data for
$\SDP(\bA)$ (red circles).

Note that Theorem \ref{thm:Harmonic} implies that the lower bound in Theorem \ref{thm:ER}
is not tight (see in particular Eq.~(\ref{eq:ExpansionCond})) and it provides a tighter lower bound (at least for large $d$).

\subsection{Stochastic block model $G\sim\Graph(n,a/n,b/n)$}
\label{sec:SBM}
As discussed in the previous sections, local algorithms can approximately solve the SDP (\ref{eq:SDP}) for $\bA=\bA_G$
the adjacency matrix of $G\sim\Graph(n,d/n)$. The stochastic block model $G\sim\Graph(n,a/n,b/n)$ provides a simple
example in which they are bound to fail, although they can succeed with a small
amount of additional side information.

As stated in the introduction, a random graph $G\sim\Graph(n,a/n,b/n)$ over $n$ vertices is generated as follows.
Let $\bsigma\in\{+1,-1\}^n$ be distributed uniformly at random, conditional on
$\<\bsigma,\e\>=0$. Conditional on  $\bsigma$, any two vertices $i$, $j$ 
are connected by an edge independently with probability  $a/n$ if
$\sigma(i)=\sigma(j)$ and with probability $b/n$ otherwise.  
We will assume $a>b$: in the social sciences parlance, the graph is assortative.

The average vertex degree of such a graph is $d=(a+b)/2$. We assume $d>1$ to ensure that $G$ has a giant component with high probability.
The signal-to-noise ratio parameter $\lambda = (a-b)/\sqrt{2(a+b)}$ plays
a special role in the model's behavior. If $\lambda<1$, then the total variation distance between 
$\Graph(n,a/n,b/n)$ and the \ER graph $\Graph(n,d/n)$ is bounded away from $1$. 
On the other hand, if $\lambda\ge 1$, then we can test whether $G\sim \Graph(n,a/n,b/n)$ or $G\sim \Graph(n,d/n)$  
with probability of error converging to $0$ as $n\to\infty$
\cite{mossel2012stochastic}.

The next theorem lower-bounds the SDP value for the stochastic block model.
\begin{theorem}\label{thm:SBM}
Let $\bA=\bA_{G_n}$ be the adjacency matrix of the random graph
$G_n \sim \sG(n,a/n,b/n)$. If $d=(a+b)/2>1$ and
$\lambda = (a-b)/\sqrt{2(a+b)}>1$, then
 for a universal constant $C>0$ (independent of $a$ and $b$), almost surely,
\begin{equation}\label{eq:lowerSBM}
\liminf_{n \to \infty} \frac{1}{n}\SDP(\bA) \geq \sqrt{d}
\, \max\left(\lambda,\;2+\frac{(\lambda-1)^2}{\lambda\sqrt{d}}-
\frac{C}{d}\right).
\end{equation}
\end{theorem}
\noindent (The first bound in (\ref{eq:lowerSBM}) dominates for large
$\lambda$, whereas the second dominates near the information-theoretic
threshold $\lambda=1$ for large $d$.)

On one hand, this theorem implies that local algorithms fail to approximately
solve the SDP (\ref{eq:SDP}) for the stochastic block model, for the following
reason: The local structures of $G\sim \sG(n,a/n,b/n)$ and $G\sim \sG(n,d/n)$
are the same asymptotically, in the sense that
they both converge locally to the Galton-Watson tree with $\Pois(d)$ offspring
distribution. This and the upper bound of Theorem \ref{thm:ER} immediately imply
that for any $F \in \cF_*$,
\begin{equation}\label{eq:LocalUpper}
\limsup_{n \to \infty} \cE(F;G_n)\le 2\sqrt{d}\Big(1-\frac{1}{2d}\Big)\,.
\end{equation}
In particular, the gap between this upper bound and the lower bound
(\ref{eq:lowerSBM}) for the SDP value is unbounded for large $\lambda$.

This problem is related to the symmetry between $+1$ and
$-1$ labels in this model. It can be resolved if we allow the local algorithm
to depend on $\Ball_{\ell_n}(\root;G)$ where $\ell_n$ grows logarithmically
in $n$, or alternatively if we provide a small amount of side information about
the hidden partition. Here we explore the latter scenario (see also \cite{mossel2016local} for related work).

Suppose that for each vertex $i \in V$, the label
$\sigma(i) \in \{+1,-1\}$ is revealed independently with probability $\delta$
for some fixed $\delta>0$, and that the radius-$\ell$ local algorithm has access
to the revealed labels in $\Ball_\ell(\root;G)$. More formally, let
$\cM=\{+1,-1,u\}$ be the set of possible vertex labels, where $u$ codes for
`unrevealed', let $\bsigma:V \to \cM$ be any assignment of labels to vertices,
and let $\cG^*(\cM)$ be the space of tuples $(G,\root,\bsigma,\bz)$ (where
$(G,\root,\bz) \in \cG^*$ as before).
\begin{definition}\label{def:LocalMarked}
A {\bf radius-$\pmb{\ell}$ local algorithm using partially revealed labels}
for the semidefinite program (\ref{eq:SDP}) is any measurable function
$F:\cG^*(\cM) \to \reals$ such that
\begin{enumerate}
\item $F(G_1,\root_1,\bz_1,\bsigma_1)=F(G_2,\root_2,\bz_2,\bsigma_2)$ if
$(\Ball_\ell(\root_1;G_1),\bz_1,\bsigma_1) \simeq (\Ball_\ell(\root_2;G_2),
\bz_2,\bsigma_2)$, where $\simeq$ denotes isomorphism that preserves the root
vertex, vertex marks, and vertex labels in $\cM$.
\item Letting $\bz=(z(i))_{i \in V}$ be i.i.d.\ with $z(i) \sim \Normal(0,1)$,
we have $\E_\bz\{F(G,\root,\bsigma,\bz)^2\}=1$, where $\E_\bz$ denotes
expectation only over $\bz$.
\end{enumerate}
\end{definition}
We denote the set of such functions $F$ by $\cF_*^\cM(\ell)$, and we denote
$\cF_*^\cM \equiv \cup_{\ell \geq 1} \cF_*^\cM(\ell)$. For any $F \in
\cF_*^\cM$, we denote
\begin{equation}
\cE(F;G,\bsigma) \equiv \frac{1}{n}\E_{\bz}\left\{\sum_{i,j \in V}
\left((\bA_G)_{ij}-\frac{d}{n}\right)F(G,i,\bsigma,\bz)F(G,j,\bsigma,\bz)
\right\},\label{eq:ValueDefMarked}
\end{equation}
so that $F$ yields a solution to the SDP (\ref{eq:SDP}) achieving value
$n\,\cE(F;G;\bsigma)$. Then we have the following result:
\begin{theorem}\label{thm:SBMlocal}
Let $\bA=\bA_{G_n}$ be the adjacency matrix of the random graph $G_n \sim
\sG(n,a/n,b/n)$. For any fixed $\delta>0$, let
$\bsigma_n=(\sigma_n(i))_{i \in V(G_n)}$ be random and such that,
independently for
each $i \in V(G_n)$, with probability $1-\delta$ we have $\sigma_n(i)=u$,
and with probability $\delta$ we have that $\sigma_n(i) \in \{+1,-1\}$
identifies the component of the hidden partition containing $i$.
If $d=(a+b)/2 \geq 2$ and $\lambda=(a-b)/\sqrt{2(a+b)}>1$, then
for any $\eps>0$, there exist $\ell(\eps)>0$ and $F \in \cF_*^\cM(\ell(\eps))$
for which, almost surely,
\begin{align}
\lim_{n \to \infty} \cE(F;G_n,\bsigma_n)\geq \sqrt{d}
\left(2+\frac{(\lambda-1)^2}{\lambda\sqrt{d}}-
\frac{C}{d}\right)-\eps.\label{eq:locallowerSBM}
\end{align}
\end{theorem}
The restriction to $d \geq 2$ above is arbitrary; our proof is valid if this
constraint is relaxed to $d \geq
1+\eta$ for any $\eta>0$, at the expense of a larger constant $C:=C_\eta$.
This theorem implies the second lower bound of (\ref{eq:lowerSBM}) when $d \geq
2$; the first
lower bound of (\ref{eq:lowerSBM}) is trivial and proven in Section
\ref{sec:Lower}.
%

%*****************************************************************
%
\subsection{Testing in the stochastic block model}

Semidefinite programming can be used as follows to test whether $G\sim \Graph(n,a/n,b/n)$ or $G\sim \Graph(n,d/n)$:
\begin{enumerate}
\item Given the graph $G$, compute the value $\SDP(\bA_G)$ of the program (\ref{eq:SDP}).
\item If $\SDP(\bA_G)/n\ge 2\sqrt{d}(1-(2d)^{-1})+\eps$, reject the null hypothesis  $G\sim\Graph(n,d/n)$.
\end{enumerate}
(Here $\eps$ is a small constant independent of $n$.)
The rationale for this procedure is provided by Theorem \ref{thm:ER}, implying that, if $G\sim\Graph(n,d/n)$,
then the probability of false discovery (i.e.\ rejecting the null when $G\sim \Graph(n,d/n)$) converges to $0$ as $n\to\infty$.

We have the following immediate consequence of Theorem \ref{thm:ER} and Theorem \ref{thm:SBM}
(here error probability refers to the probability of false discovery plus the
probability of miss-detection, i.e.\ not rejecting the null when
$G\sim \Graph(n,a/n,b/n)$):
\begin{corollary}
The SDP-based test has error probability converging to 0 provided $\lambda> \lsdp(d)$, where
\begin{align}
\lsdp(d) \leq \min\Big(2-\frac{1}{d},\; 1+\frac{C}{d^{1/4}}\Big)\, .
\end{align} 
\end{corollary}
\noindent For comparison\footnote{Note \cite{montanari2016semidefinite} also proves guarantees on the estimation error. We believe it should be 
possible to improve those results using the methods in the present paper, but we defer it to future work.}, 
\cite{montanari2016semidefinite} proved $\lsdp(d) = 1+o_d(1)$ for large $d$,
while the last result gives a quantitative bound for all $d$. 
%
%*****************************************************************
%
\section{Further related work}
\label{sec:Related}
The SDP relaxation (\ref{eq:SDP}) has attracted a significant amount of work
since Goemans-Williamson's seminal work on the MAXCUT problem
\cite{goemans1995improved}. In the last few years, several authors used this approach for clustering or community detection and derived optimality
or near-optimality guarantees. An incomplete list includes \cite{bandeira2014multireference,abbe2016exact,hajek2016achieving,hajek2015achieving,awasthi2015relax}.
Under the assumption that $G$ is generated according to the stochastic block model (whose two-groups version was introduced in Section \ref{sec:SBM}),
these papers provide conditions under which the SDP approach recovers \emph{exactly} the vertex labels.
This can be regarded as a `high signal-to-noise ratio' regime, in which (with high probability)
the SDP solution has rank one and is deterministic (i.e. independent of the graph realization).
In contrast, we focus on the `pure noise' scenario in which $G\sim
\Graph(n,d/n)$ is an \ER random graph,
or on the two-groups stochastic block-model $G\sim\Graph(n,a/n,b/n)$ close to the detection threshold. In this regime, 
the SDP optimum  has rank larger than one and is non-deterministic.
The only papers that have addressed this regime using SDP are
\cite{guedon2015community,montanari2016semidefinite}, discussed previously.

Several papers applied sophisticated spectral methods to the stochastic block model near the detection threshold \cite{massoulie2014community,mossel2013proof,bordenaveetal}. 
Our upper bound in Theorem \ref{thm:ER} is based on a duality argument, where
we establish feasibility of a certain dual witness 
construction using an argument similar to \cite{bordenaveetal}.

Several papers studied the use of local algorithms to solve combinatorial optimization problems on graphs.
Hatami, Lov\'asz and Szegedy \cite{hatami2014limits} investigated several notions of graph convergence, and put forward a conjecture implying --in particular--
that local algorithms are able to find (nearly)  optimal solutions of a broad class of combinatorial problems on random $d$-regular graphs. 
This conjecture was disproved by Gamarnik and Sudan \cite{gamarnik2014limits} by
considering maximum size independent sets on random $d$-regular graphs.
In particular, they proved that the size of an independent set produced by a local algorithm is at most 
$(1/2)+(1/\sqrt{8})+\eps$ times the maximum independent set, for large enough $d$.  Rahman and Virag \cite{rahman2014local} improved this result 
by establishing that  no local algorithm can produce independent sets of size larger than $(1/2)+\eps$ times the maximum 
independent set, for large enough $d$. This approximation ratio is essentially optimal, since known local algorithms can achieve $(1/2)-\eps$ of
the maximum independent set. 
It is unknown whether  a similar gap is present for small degree $d$. In
particular, Cs\'oka et al. \cite{csokaetal2014independent} establish a lower-bound on
the max-size  independent set on random $3$-regular graphs. A similar technique is used by Lyons \cite{lyons2014factors} to lower bound the max-cut on
random $3$-regular graphs.
In summary, the question  of which graph-structured optimization problems can be approximated by local algorithms is broadly open.

By construction, local algorithms can be applied to infinite random graphs, and have a well defined value provided the
graph distribution is unimodular (see below). Asymptotic results for graph
sequences can be `read-off' these infinite-graph settings (our proofs will use this
device multiple times). In this context, the (random) solutions generated by local algorithms, together with their limits in the weak topology,
 are referred to as `factors of i.i.d.\ processes' \cite{lyons2014factors}.
%

%
%*****************************************************************
%
\section{Notation}

We use upper case boldface for matrices (e.g. $\bA$, $\bB$, \dots), lower case boldface for vectors (e.g. $\bu$, $\bv$, etc.)
and lower case plain for scalars (e.g. $x,y,\dots$). The scalar product of vectors $\bu,\bv\in\reals^m$ is denoted by 
$\<\bu,\bv\>=\sum_{i=1}^mu_iv_i$, and the scalar product between matrices is indicated in the same way $\<\bA,\bB\> = \Tr(\bA\bB^{\sT})$.

Given a matrix $\bA\in\reals^{m\times m}$, $\diag(\bA)\in\reals^m$ is the vector
that contains its diagonal entries. Conversely,
given $\bv\in\reals^m$, $\diag(\bv)\in\reals^{m\times m}$ is the diagonal matrix with entries $\diag(\bv)_{ii} = v_i$. 

We denote by $\e$ the all-ones vector and by $\Id$ the identity matrix.

We follow the standard big-Oh notation.
%
%*****************************************************************
%
\section{Upper bound: Theorem \ref{thm:ER}}

In this section, we prove the upper bound in Theorem \ref{thm:ER}.
Denote $\PSD_1:=\{\bX:\bX \succeq 0,\bX_{ii}=1\;\forall i\}$.
Introducing dual variables $\bLambda \succeq 0$ and $\bnu \in \R^n$
and invoking strong duality, we have
\begin{align}
\SDP(\bA)&=\max_{\bX \in \PSD_1} \min_{\bnu,\bLambda\succeq 0}
\left\<\bA-\frac{d}{n}\e\e^{\sT},\bX\right\>+\< \bLambda, \bX\>
-\<\bnu, \diag(\bX)-\e\>\\
&=\min_{\bnu,\bLambda\succeq 0} \max_{\bX \in \PSD_1} \<\bnu ,\e\>
+\left\<\bA-\frac{d}{n}\e\e^{\sT}+\bLambda-\diag(\bnu)\,,\,\bX\right\>\, .
\end{align}
The minimum over $\bLambda \succeq 0$ occurs at $\bLambda=0$, hence
$\SDP(\bA)$ is equivalently given by the value of the
dual minimization problem over $\bnu \in \R^n$:
\begin{align}
\mbox{\rm minimize} & \;\;\;\<\bnu,\e\> \label{eq:dualSDP}\\
\mbox{\rm subject to} & \;\;\;\bA-\frac{d}{n}\e\e^{\sT} \preceq \diag(\bnu).\nonumber
\end{align}

We prove the upper bound in Theorem \ref{thm:ER} by constructing a dual-feasible
solution $\bnu$, parametrized by a small positive constant
$\delta \in (0,1/\sqrt{d})$.
Denote the diagonal degree matrix of $\bA$ as $\bD:=\diag(\bA\e)$ and set
\begin{equation}\label{eq:nu}
u:=\frac{1}{\sqrt{d}}-\delta,\;\;\;\;
\bnu:=\begin{cases}
\diag\left(\frac{1+\delta-u^2}{u}\Id+u\bD\right) &\mbox{ if }
\frac{1+\delta-u^2}{u}\Id+u\bD-\bA+\frac{d}{n}\e\e^{\sT} \succeq 0 \\
\diag(\bD) & \text{ otherwise.}
\end{cases}
\end{equation}
The following is the main result of this section, which ensures that the first
case in the definition of $\bnu$ in (\ref{eq:nu}) holds with high probability.
\begin{theorem}\label{thm:psd}
For fixed $d>1$, let $\bA$ be the adjacency matrix of the \ER random graph
$G \sim \sG(n,d/n)$, and let $\bD:=\diag(\bA\e)$ be the diagonal degree matrix.
Then for any $\delta \in (0,1/\sqrt{d})$ and for $u=1/\sqrt{d}-\delta$,
with probability approaching 1 as $n \to \infty$,
\[\frac{1+\delta-u^2}{u}\Id+u\bD-\bA+\frac{(1-u^2)d}{n}\e\e^{\sT} \succ 0.\]
\end{theorem}

Let us first show that this implies the desired upper bound:
\begin{proof}[Proof of Theorem \ref{thm:ER} (upper bound)]
By construction, $\bnu$ as defined in (\ref{eq:nu})
is a feasible solution for the dual problem
(\ref{eq:dualSDP}). Let $\cE$ be the event where
\[\frac{1+\delta-u^2}{u}\Id+u\bD-\bA+\frac{d}{n}\e\e^{\sT} \succeq 0.\]
Then
\[\bnu^{\sT}\e=\left(\frac{1+\delta-u^2}{u}n+u\e^{\sT}\bA\e\right)\1\{\cE\}
+(\e^{\sT}\bA\e)\1\{\cE^c\}.\]
As $\e^{\sT}\bA\e \sim 2\Binom(\binom{n}{2},d/n)$, this implies
$\E[\e^{\sT}\bA\e] \leq dn$ and $\E[(\e^{\sT}\bA\e)^2] \leq d^2(n^2+1)$. Then
\begin{align*}
\E\left[\frac{1}{n}\SDP(\bA)\right]
&\leq \frac{1}{n}\E[\bnu^{\sT}\e]
\leq \frac{1+\delta-u^2}{u}+\frac{u}{n}\E[\e^{\sT}\bA\e]
+\frac{1}{n}\E\left[(\e^{\sT}\bA\e)\1\{\cE^c\}\right]\\
&\leq \frac{1+\delta-u^2}{u}+ud
+\frac{1}{n}\sqrt{\E[(\e^{\sT}\bA\e)^2]\P[\cE^c]}
\leq \frac{1+\delta-u^2}{u}+ud+d\sqrt{\P[\cE^c]}.
\end{align*}
By Theorem \ref{thm:psd}, $\P[\cE^c] \to 0$ as $n \to \infty$. Taking $n \to
\infty$ and then $\delta \to 0$,
\[\limsup_{n \to \infty} \E\left[\frac{1}{n} \SDP(\bA)\right]
\leq 2\sqrt{d}-\frac{1}{\sqrt{d}}=2\sqrt{d}\left(1-\frac{1}{2d}\right).\]

To obtain the bound almost surely rather than in expectation, note that
if $G$ and $G'$ are two fixed graphs that differ in one edge, with adjacency
matrices $\bA$ and $\bA'$, then
\[\left|\Tr\left(\bA-\frac{d}{n}\e\e^{\sT}\right)\bX-
\Tr\left(\bA'-\frac{d}{n}\e\e^{\sT}\right)\bX\right| \leq 2\]
for any feasible point $\bX$ of (\ref{eq:SDP}),
so $|\frac{1}{n}\SDP(\bA)-\frac{1}{n}\SDP(\bA')| \leq
\frac{2}{n}$. Let $e_1,\ldots,e_{\binom{n}{2}}$ be any ordering of the edges
$\{(i,j):1 \leq i<j \leq n\}$, and denote by
$\cF_0,\cF_1,\ldots,\cF_{\binom{n}{2}}$ the filtration where $\cF_l$ is
generated by $\bA_{e_1},\ldots,\bA_{e_l}$. Then by coupling, this
implies for each $l=1,\ldots,\binom{n}{2}$
\[|d_l|:=\left|\E\left[\frac{1}{n}\SDP(\bA) \bigg| \cF_l\right]-
\E\left[\frac{1}{n}\SDP(\bA) \bigg| \cF_{l-1}\right]\right| \leq
\1\{\bA_{e_l}=0\}\frac{d}{n}\cdot \frac{2}{n}+
\1\{\bA_{e_l}=1\}\left(1-\frac{d}{n}\right)\frac{2}{n}.\]
Hence $|d_l| \leq 2/n$ for each $l$, and
\[V_n:=\sum_{l=1}^{\binom{n}{2}} |d_l|^2 \leq
\frac{\e^{\sT}\bA\e}{2}\left(\left(1-\frac{d}{n}\right)\frac{2}{n}\right)^2
+\left(\binom{n}{2}-\frac{\e^{\sT}\bA\e}{2}\right)\left(\frac{2d}{n^2}\right)^2
\leq \frac{2d^2}{n^2}+\frac{2}{n^2}(\e^{\sT}\bA\e).\]
Bernstein's inequality yields $\P[\e^{\sT}\bA\e>3dn] \leq
\exp(-C_dn)$ for a constant $C_d>0$. Then, applying the
martingale tail bound of \cite[Theorem 1.2A]{pena}, for any $\eps>0$,
\begin{align*}
&\P\left[\left|\frac{1}{n}\SDP(\bA)-\E\left[\frac{1}{n}\SDP(\bA)\right]\right|
\geq \eps\right]\\
&\hspace{1in}\leq \exp\left(-\frac{\eps^2}{2\left(\frac{2d^2}{n^2}+\frac{6d}{n}+\frac{2\eps}
{n}\right)}\right)+\P\left[V_n \geq \frac{2d^2}{n^2}+\frac{6d}{n}\right]
\leq 2\exp\left(-C_{d,\eps}n\right)
\end{align*}
for a constant $C_{d,\eps}>0$. Then the Borel-Cantelli lemma implies
$|\frac{1}{n}\SDP(\bA)-\E[\frac{1}{n}\SDP(\bA)]|<\eps$ almost surely for all
large $n$, and the result follows by taking $\eps \to 0$.
\end{proof}

In the remainder of this section, we prove Theorem \ref{thm:psd}. Heuristically,
we might expect that Theorem \ref{thm:psd} is true by the following
reasoning: The matrix $\bH(u):=(1-u^2)\Id+u^2\bD-u\bA$ is the deformed
Laplacian, or Bethe Hessian, of the graph. By a relation in graph theory known
as the Ihara-Bass formula \cite{bass,kotanisunada}, the values of $u$ for which
this matrix is singular are the inverses of the non-trivial eigenvalues of a
certain ``non-backtracking matrix'' \cite{krzakala2013spectral,saadeetal}.
Theorem 3 of \cite{bordenaveetal} shows that this non-backtracking matrix
has, with high probability, the bulk of its spectrum supported on the complex
disk of radius approximately $\sqrt{d}$, with a single outlier eigenvalue at
$d$. From this, the observation that $\bH(0) \succeq 0$, and a continuity
argument in $u$, one deduces that $\bH(u)$ has, with high probability for large
$n$, only a single negative eigenvalue when $u \in (1/d,1/\sqrt{d})$. If the
eigenvector corresponding to this eigenvalue has positive alignment with
$\e \in \R^n$, then adding a certain multiple of the rank-one matrix
$\e\e^{\sT}$ should eliminate this negative eigenvalue.

Direct analysis of the rank-one perturbation of $\bH(u)$
is hindered by the fact that the
spectrum and eigenvectors of $\bH(u)$ are difficult to characterize. Instead,
we will study a certain perturbation of the non-backtracking matrix. We
prove Theorem \ref{thm:psd} via the following two steps:
First, we prove a generalization of the Ihara-Bass relation to edge-weighted
graphs. For any graph $H=(V,E)$, let $c:E \to \R$ be a set of possibly negative edge
weights. For each $u \in \R$ such that $|u| \notin \{|c(i,j)|^{-1}:\{i,j\} \in
E\}$, define the
$n \times n$ symmetric matrix $\bA_{c,u}$ and diagonal matrix $\bD_{c,u}$ by
\[\bA_{c,u}(i,j)=\begin{cases} \frac{uc(i,j)}{1-u^2c(i,j)^2} & \{i,j\} \in E \\
0 & \text{otherwise,} \end{cases}\;\;\;\;
\bD_{c,u}(i,j)=\begin{cases} \sum_{k:\{i,k\} \in E}
\frac{u^2c(i,k)^2}{1-u^2c(i,k)^2} & i=j \\
0 & \text{otherwise.} \end{cases}\]
Let $\oE$ denote the set of directed edges $\oE:=\{(i,j),(j,i):\{i,j\}
\in E\}$, and define the $|\oE| \times |\oE|$ weighted non-backtracking matrix
$\bB_c$, with rows and columns indexed by $\oE$, as
\[\bB_c((i,j),(i',j'))=\begin{cases} c(i',j') & i'=j,j' \neq i \\
0 & \text{otherwise.}\end{cases}\]
The following result relates $\bB_c$ with a generalized
deformed Laplacian defined by $\bA_{c,u}$ and $\bD_{c,u}$:
\begin{lemma}[Generalized Ihara-Bass formula]\label{lemma:iharabass}
For any graph $H=(V,E)$, edge weights $c:E \to \R$, $u \in \R$ with $|u| \notin
\{|c(i,j)|^{-1}:\{i,j\} \in E\}$, and the matrices $\bB_c$, $\bA_{c,u}$, and
$\bD_{c,u}$ as defined above,
\begin{align}
\det(\Id-u\bB_c) = \det(\Id+\bD_{c,u}-\bA_{c,u})
\prod_{\{i,j\}\in E}(1-u^2c^2(i,j)).
\end{align}
\end{lemma}
\noindent When $c \equiv 1$, this recovers the standard Ihara-Bass identity.
The proof is a direct adaptation of the proof for the unweighted case in
\cite{kotanisunada}; for the reader's convenience we provide it in
Section \ref{subsec:iharabass}.

Second, we consider a weighted non-backtracking matrix
$\bB \in \R^{n(n-1) \times n(n-1)}$ of the above form for the complete graph
with $n$ vertices,
with rows and columns indexed by all ordered pairs $(i,j)$ of distinct indices
$i,j \in \{1,\ldots,n\}$, and defined as
\begin{equation}\label{eq:B}
\bB((i,j),(i',j'))=\begin{cases}
\bA_{i'j'}-\frac{d}{n} & i'=j,j'\neq i\\
0 & \text{otherwise}.
\end{cases}
\end{equation}
We prove in Section \ref{subsec:rhoB} that $\bB$ no longer has an outlier
eigenvalue at $d$, but instead has all of its eigenvalues contained within the
complex disk of radius approximately $\sqrt{d}$:
\begin{lemma}\label{lemma:rhoB}
Fix $d>1$, let $\bA$ be the adjacency matrix of the \ER random graph
$\sG(n,d/n)$, and define $\bB \in \R^{n(n-1) \times n(n-1)}$ by (\ref{eq:B}).
Let $\rho(\bB)$ denote the spectral radius of $\bB$. Then for any
$\eps>0$, with probability approaching 1 as $n \to \infty$,
\[\rho(\bB)\leq \sqrt{d}(1+\eps).\]
\end{lemma}

Using these results, we
prove Theorem \ref{thm:psd}:

\begin{proof}[Proof of Theorem \ref{thm:psd}]
Denote by $\bD$ the diagonal degree matrix of $\bA$.
Let $\cE$ be the event on which $\rho(\bB) \leq (1/\sqrt{d}-\delta/2)^{-1}$ and
$\|\bD\| \leq 2\log n$. Each diagonal entry of $\bD$ is distributed as
$\Binom(n-1,d/n)$, hence $\P[\|\bD\|>2\log n] \leq c_d/{n^2}$ for a constant
$c_d>0$ by Bernstein's inequality and a union bound. This and Lemma
\ref{lemma:rhoB} imply $\cE$ holds with probability approaching 1.

On $\cE$, $\det(\Id-u\bB) \neq 0$ for all $u \in (0,1/\sqrt{d}-\delta/2)$.
Applying Lemma \ref{lemma:iharabass} for the complete graph $H$ with edge
weights $c(i,j)=\bA_{i,j}-d/n$,
and noting $|u| \neq |c(i,j)|^{-1}$ for any $\{i,j\}$ and any
$u \in (0,1/\sqrt{d}-\delta/2)$
when $n$ is sufficiently large, we have $\det(\Id+\bD_u-\bA_u) \neq 0$ for
\begin{align*}
\bA_u&=\left(\frac{u(1-d/n)}{1-u^2(1-d/n)^2}
+\frac{ud/n}{1-u^2d^2/n^2}\right)\bA
-\frac{ud/n}{1-u^2d^2/n^2}\e\e^{\sT}+\frac{ud/n}{1-u^2d^2/n^2}\Id,\\
\bD_u&=\left(\frac{u^2(1-d/n)^2}{1-u^2(1-d/n)^2}
-\frac{u^2d^2/n^2}{1-u^2d^2/n^2}\right)\bD+\frac{(n-1)u^2d^2/n^2}{1-u^2d^2/n^2}
\Id.
\end{align*}
Note that at $u=0$, $\Id+\bD_u-\bA_u=\Id \succ 0$. Then by
continuity in $u$, $\Id+\bD_u-\bA_u \succ 0$ for all $u \in
(0,1/\sqrt{d}-\delta/2)$ on the event $\cE$.

Choosing $u=1/\sqrt{d}-\delta$, it is easily verified that
\[\Id+\bD_u-\bA_u=\Id+\frac{u^2}{1-u^2}\bD-\frac{u}{1-u^2}\bA
+\frac{ud}{n}11^{\sT}+\bR\]
for a remainder matrix $\bR$ satisfying
$\|\bR\| \leq \frac{C_{d,\delta}}{n}(\|\bA\|+\|\bD\|+1) \leq
\frac{C_{d,\delta}}{n}(2\|\bD\|+1)$ for a constant $C_{d,\delta}>0$.
Hence on $\cE$,
$\bR \preceq \frac{\delta}{1-u^2}\Id$ for all large $n$, and
rearranging yields the desired result.
\end{proof}

In the following two subsections, we complete the proof by proving Lemmas
\ref{lemma:iharabass} and \ref{lemma:rhoB}.

\subsection{Proof of Lemma \ref{lemma:iharabass}}\label{subsec:iharabass}
We follow the argument of \cite[Sections 4 and 5]{kotanisunada}. Assume without
loss of generality $c(i,j) \neq 0$ for all $\{i,j\} \in E$. (Otherwise, remove
$\{i,j\}$ from $E$.) Identify
$\R^{|\oE|}$ with the space of linear functionals $\omega:\oE \to \R$ and
$\R^{|V|}$ with the space of linear functionals $f:V \to \R$.
Consider the orthogonal decomposition $\R^{|\oE|}=C_- \oplus C_+$, where
\[C_{\pm}:=\{\omega:\omega(i,j)=\pm \omega(j,i)\;\;\forall (i,j) \in \oE\}.\]
Any $\omega \in \R^{|\oE|}$ has the decomposition
$\omega=\omega_-+\omega_+$ where
$\omega_\pm(i,j)=\frac{1}{2}(\omega(i,j)\pm\omega(j,i)) \in C_\pm$.
Define $\rd_\pm:\R^{|V|} \to C_\pm$ by
\[(\rd_{\pm}f)(i,j):=f(j)\pm f(i),\]
and $\delta_\pm:C_\pm\to \R^{|V|}$ by
\[(\delta_{\pm}\omega)(k):=\pm\!\!\!\sum_{(i,j) \in E^o:i=k} c(i,j)\omega(i,j).
\]
With a slight abuse of notation, denote by $c:C_\pm \to C_\pm$
the diagonal operators $(c\omega)(i,j):=c(i,j)\omega(i,j)$ (which are
well-defined as $c(i,j)=c(j,i)$).

For $\omega \in C_+$,
\begin{align*}
\frac{1}{2}\left((\bB_c\omega)(i,j)+(\bB_c\omega)(j,i)\right)
&=\frac{1}{2}\left(\sum_{(i',j') \in E^o:i'=j,j' \neq i}
c(i',j')\omega(i',j')+\sum_{(i',j') \in E^o:i'=i,j' \neq j} 
c(i',j')\omega(i',j')\right)\\
&=\frac{1}{2}\left((\delta_+\omega)(j)-c(j,i)\omega(j,i)
+(\delta_+\omega)(i)-c(i,j)\omega(i,j)\right)\\
&=\left(\tfrac{1}{2}\rd_+\delta_+\omega-c\omega\right)(i,j)
\end{align*}
Similar computations for $\frac{1}{2}((\bB_c\omega)(i,j)-(\bB_c\omega)(j,i))$
and for $\omega \in C_-$ verify that $\bB_c$ has the following block
decomposition with respect to $\R^{|E^o|}=C_-\oplus C_+$:
\[\bB_c = \left(\begin{matrix}
-\frac{1}{2} \rd_-\delta_-+c& \frac{1}{2} \rd_-\delta_+\\
-\frac{1}{2} \rd_+\delta_- & \frac{1}{2} \rd_+\delta_+-c
\end{matrix}\right).\]
Define the matrices $\bT \in \R^{|V|\times |\oE|}$, $\bS \in
\R^{ |\oE|\times |V| }$, and $\btau \in \R^{|\oE| \times |\oE|}$
(with respect to the decomposition $\R^{|\oE|}=C_-\oplus C_+$) by
\[\bT:=\Big( \; \delta_-,\;\;\; -\delta_+\; \Big),\;\;
\bS:=\left(\begin{matrix} \rd_-\\ \rd_+ \end{matrix} \right),\;\;
\btau:=\left(\begin{matrix} -c & 0 \\ 0 & c \end{matrix}\right).\]
Then $\bB_c=-\btau-\frac{1}{2}\bS\bT$, and hence
\begin{equation}
(\Id-u\bB_c) (\Id+u\btau)^{-1}=\Id
+\frac{1}{2}u\bS\bT(\Id+u\btau)^{-1}.\label{eq:BSTid}
\end{equation}
($\Id+u\btau$ is invertible by the assumption
$|u| \notin \{|c(i,j)|^{-1}:\{i,j\} \in E\}$.)
In particular, this implies that
$(\Id-u\bB_c) (\Id+u\btau)^{-1}$ preserves
$\Ker \bT(\Id+u\btau)^{-1}$ and $\Im \bS$.

For $f \in \R^{|V|}$ and $k \in V$, we compute
\begin{align*}
(u\bT(\Id+u\btau)^{-1}\bS f)(k)
&=-\!\!\!\sum_{(i,j) \in E^o:i=k} \frac{uc(i,j)}{1-uc(i,j)}(f(j)-f(i))
-\!\!\!\sum_{(i,j) \in E^o:i=k} \frac{uc(i,j)}{1+uc(i,j)}(f(j)+f(i))\\
&=-\!\!\!\sum_{j:\{k,j\} \in E} \frac{2uc(k,j)}{1-u^2c(k,j)^2}f(j)
+\!\!\!\sum_{j:\{k,j\} \in E} \frac{2u^2c(k,j)^2}{1-u^2c(k,j)^2}f(k)\\
&=(2\bD_{c,u}f-2\bA_{c,u}f)(k).
\end{align*}
Hence $u\bT(\Id+u\btau)^{-1}\bS=2\bD_{c,u}-2\bA_{c,u}$. The determinant of this
matrix is a rational function of $u$ and non-zero for large $|u|$, so
$u\bT(\Id+u\btau)^{-1}\bS$ is invertible for generic $u \in \R$. For any such
$u$, $\Ker \bT(\Id+u\btau)^{-1}$ and $\Im \bS$ are linearly independent.
Furthermore, one may verify $\delta_{\pm} c^{-1}=\rd_{\pm}^*$,
so $-\bT\btau^{-1}=\bS^*$ and $|\oE|=\dim \Im \bS+\dim \Ker \bS^*
=\dim \Im \bS+\dim \Ker \bT(\Id+u\btau)^{-1}$. Hence for generic $u$,
$\Ker \bT(\Id+u\btau)^{-1}$ and $\Im \bS$ span all of $\R^{|\oE|}$.
By (\ref{eq:BSTid}),
we may write the block decomposition of $(\Id-u\bB_c)(\Id+u\btau)^{-1}$ 
with respect to $\R^{|\oE|}=\Im\bS \oplus \Ker\bT(\Id+u\btau)^{-1}$ as
\[(\Id-u\bB_c) (\Id+u\btau)^{-1} = 
\left( \begin{matrix}
\bS\big(\Id +\frac{1}{2} u\bT(1+u\btau)^{-1}\bS\big)\bS^{-1} & 0\\
0 & \Id
\end{matrix} \right).\]
Then, computing the determinant of both sides and rearranging,
\[\det(\Id-u\bB_c)=\prod_{\{i,j\}\in E}\left(1-u^2c^2(i,j)\right)
\times\det\left(\Id
+\frac{1}{2} u\bT(1+u\btau)^{-1}\bS\right).\]
Recalling $u\bT(1+u\btau)^{-1}\bS=2\bD_{c,u}-2\bA_{c,u}$, this
establishes the result for generic $u$. The conclusion for all $|u| \notin
\{|c(i,j)|^{-1}:\{i,j\} \in E\}$ then follows by continuity.

\subsection{Proof of Lemma \ref{lemma:rhoB}}\label{subsec:rhoB}
We bound $\rho(\bB^{2m}) \leq \Tr \bB^m(\bB^{\sT})^m$ for some $m:=m(n)$,
and apply the moment method to bound the latter quantity. Throughout this
section, ``edge'' and ``graph'' mean undirected edge and undirected graph.
We begin with some combinatorial definitions:
\begin{definition}\label{def:cyclenumber}
The {\bf cycle number} of a graph $H$, denoted $\numcycle(H)$,
is the minimum number of 
edges that must be removed from $H$ so that the resulting graph has no cycles.
(If $H$ has $k$ connected components, $v$ vertices, and $e$ edges, then
$\numcycle(H)=e+k-v$.)
\end{definition}

\begin{definition}\label{def:lcoil}
For any $l \geq 1$, an $\pmb{l}${\bf-coil} is any connected graph with at most
$l$ edges and at least two cycles.
A graph $H$ is $\pmb{l}${\bf-coil-free} if $H$ contains no $l$-coils,
i.e.\ every connected subset of at most $l$ edges in $H$ contains at most one
cycle.
\end{definition}

\begin{definition}
A sequence of vertices $\gamma_0,\ldots,\gamma_m \in \{1,\ldots,n\}$ is
a {\bf non-backtracking path} of length $m$ (on the complete graph) if
$\gamma_1 \neq \gamma_0$ and $\gamma_{j+1} \notin \{\gamma_{j-1},\gamma_j\}$ for
each $j=1,\ldots,m-1$. $\gamma$ {\bf visits} the vertices 
$\gamma_0,\ldots,\gamma_m$ and the edges $\{\gamma_0,\gamma_1\},\ldots,
\{\gamma_{m-1},\gamma_m\}$. $\gamma$ is $\pmb{l}${\bf-coil-free} if the
subgraph of edges visited by $\gamma$ is $l$-coil-free. For any
$K \subseteq \{0,\ldots,m-1\}$, $\gamma$ {\bf visits on} $\pmb{K}$
the edges $\{\gamma_j,\gamma_{j+1}\}:j \in K$, and $\gamma$ is
$\pmb{l}${\bf-coil-free on} $\pmb{K}$ if the subgraph formed by these edges is
$l$-coil-free.
\end{definition}

The definition of $l$-coil-free is similar to (and more convenient for our
proof than) that of $l$-tangle-free in
\cite{bordenaveetal} and \cite{mossel2013proof}, which states that every ball
of radius $l$ in $H$ contains at most one cycle. Clearly if $H$ has an $l$-coil,
then the ball of radius $l$ around any vertex in this $l$-coil has two cycles.
Hence if $H$ is $l$-tangle-free in the sense of \cite{bordenaveetal}, then it
is also $l$-coil-free, which yields the following lemma:

\begin{lemma}\label{lemma:coilfree}
Let $d>1$ and consider the \ER graph $G \sim \sG(n,d/n)$. For some absolute
constant $C>0$ and any $l \geq 1$,
$\P[G \text{ is } l\text{-coil-free}] \geq 1-Cd^{2l}/n$.
\end{lemma}
\begin{proof}
This is proven for $G$ being $l$-tangle-free in \cite[Lemma 30]{bordenaveetal};
hence the result follows the above remark.
\end{proof}

Our moment method computation will draw on the following two technical lemmas,
whose proofs we defer to Appendix \ref{app:combinatorics}.

\begin{lemma}\label{lemma:alternatingprobbound}
Let $d>1$ and let $\bA$ be the adjacency matrix of the \ER graph
$G \sim \sG(n,d/n)$.
Let $E:=\{\{v,w\}:v,w \in \{1,\ldots,n\},v \neq w\}$ be the edges of
the complete graph on $n$ vertices, let $S,T \subseteq E$ be any disjoint edge
sets such that $|S|,|T| \leq (\log n)^2$, and let $\numcycle(S \cup
T)$ be the cycle number of the subgraph formed by the edges $S \cup T$.
Let $l$ be a positive integer with $l \leq 0.1 \log_d n$. Then for some
$C:=C(d)>0$, $N_0:=N_0(d)>0$, and all $n \geq N_0$,
\begin{align*}
&\left|\E\left[\prod_{\{v,w\} \in S} \left(\bA_{vw}-\frac{d}{n}\right)
\prod_{\{v,w\} \in T} \bA_{vw}\1\{G \text{\rm{ is }}
l\text{\rm-coil-free}\}\right]\right|\\
&\hspace{1in}\leq C(\log n)^2
\left(\frac{d}{n}\right)^{|S|+|T|}2^{|S|}n^{-0.7\left(
\frac{|S|}{l}-\#_c(S \cup T)\right)}.
\end{align*}
\end{lemma}

\begin{lemma}\label{lemma:equivclassbound}
For $m,l,v,e \geq 1$ and $K_1,K_2 \subseteq \{0,\ldots,m-1\}$,
let $W(m,l,v,e,K_1,K_2)$ denote the set of all ordered pairs
$(\gamma^{(1)},\gamma^{(2)})$ of non-backtracking paths of length $m$ (on the
complete graph with $n$ vertices), such that each
$\gamma^{(i)}$ is $l$-coil-free on $K_i$, $\gamma^{(1)}_0=\gamma^{(2)}_0$,
$\gamma^{(1)}_m=\gamma^{(2)}_m$, and $(\gamma^{(1)},\gamma^{(2)})$ visit a total
of $v$ distinct vertices and $e$ distinct edges. Call two such pairs of paths
equivalent if they are the same up to a relabeling of the
vertices, and let $\cW(m,l,v,e,K_1,K_2)$ be the set of all
equivalence classes under this relation. Then the number of distinct equivalence
classes satisfies the bound
\[|\cW(m,l,v,e,K_1,K_2)| \leq \left(l(3v^2)^{2l+2}\right)^{2m-|K_1|-|K_2|}
\left(l(3v^2)^{2e-2v+4}\right)^{\frac{2m}{l}+2}.\]
\end{lemma}

Using the above results, we prove Lemma \ref{lemma:rhoB}:
\begin{proof}[Proof of Lemma \ref{lemma:rhoB}]
Let $m:=m(n) \geq 1$ with $m =o((\log n)^2)$, to be specified later. Denote
\begin{align*}
T_m:=\Tr \bB^m(\bB^{\sT})^m&=\sum_{e_1,\ldots,e_{2m}} \prod_{j=1}^m
\bB_{e_je_{j+1}}\prod_{j=1}^m (\bB^{\sT})_{e_{m+j}e_{m+j+1}}\\
&=\sum_{e_1,\ldots,e_{2m}} \prod_{j=1}^m
\bB_{e_je_{j+1}}\prod_{j=1}^m \bB_{e_{m+j+1}e_{m+j}}
\end{align*}
where the summation runs over all possible tuples of ordered vertex pairs
$e_j \in \{(v,w):v,w \in \{1,\ldots,n\},v \neq w\}$, and where $e_{2m+1}:=e_1$.
By the definition of $\bB$, a term of the above sum corresponding to
$e_1,\ldots,e_{2m}$ is 0 unless
\[e_1,\ldots,e_{m+1}=(\tilde{\gamma}_0^{(1)},\tilde{\gamma}_1^{(1)}),\ldots,
(\tilde{\gamma}_m^{(1)},\tilde{\gamma}_{m+1}^{(1)}),\]
\[e_{2m+1},\ldots,e_{m+1}=(\tilde{\gamma}_0^{(2)},\tilde{\gamma}_1^{(2)}),
\ldots,(\tilde{\gamma}_m^{(2)},\tilde{\gamma}_{m+1}^{(2)})\]
for two non-backtracking paths $\tilde{\gamma}^{(1)}$ and
$\tilde{\gamma}^{(2)}$ of length $m+1$ on the complete graph, such that
$(\tilde{\gamma}^{(1)}_0,\tilde{\gamma}^{(1)}_1)
=(\tilde{\gamma}^{(2)}_0,\tilde{\gamma}^{(2)}_1)$ and
$(\tilde{\gamma}^{(1)}_m,\tilde{\gamma}^{(1)}_{m+1})=(\tilde{\gamma}^{(2)}_m,
\tilde{\gamma}^{(2)}_{m+1})$. Letting $\tilde{\Gamma}_m$ denote the set of all
pairs of such paths,
\[T_m=\sum_{(\tilde{\gamma}^{(1)},\tilde{\gamma}^{(2)}) \in \tilde{\Gamma}_m}
\prod_{j=1}^m \left(\bA_{\tilde{\gamma}^{(1)}_j\tilde{\gamma}^{(1)}_{j+1}}
-\frac{d}{n}\right)\prod_{j=1}^m \left(\bA_{\tilde{\gamma}^{(2)}_j
\tilde{\gamma}^{(2)}_{j+1}}-\frac{d}{n}\right).\]
(The above products do not include the $j=0$ terms corresponding to the
first edge of each path.)
Writing $(\gamma^{(i)}_0,\ldots,\gamma^{(i)}_m):=(\tilde{\gamma}^{(i)}_1,\ldots,
\tilde{\gamma}^{(i)}_{m+1})$ to remove the first vertex of each path,
and letting $\Gamma_m$ denote the set of pairs $(\gamma^{(1)},\gamma^{(2)})$
where each $\gamma^{(i)}$ is a non-backtracking path of length $m$ and such that
$\gamma^{(1)}_0=\gamma^{(2)}_0$ and $(\gamma^{(1)}_{m-1},\gamma^{(1)}_m)
=(\gamma^{(2)}_{m-1},\gamma^{(2)}_m)$, the above may be written (by summing over
$\tilde{\gamma}^{(1)}_0=\tilde{\gamma}^{(2)}_0$) as
\[T_m=\sum_{(\gamma^{(1)},\gamma^{(2)}) \in \Gamma_m}\left(n-\left|\left\{
\gamma^{(1)}_0,\gamma^{(1)}_1,\gamma^{(2)}_1\right\}\right|\right)
\prod_{j=0}^{m-1} \left(\bA_{\gamma_j^{(1)}\gamma_{j+1}^{(1)}}
-\frac{d}{n}\right)\prod_{j=0}^{m-1}
\left(\bA_{\gamma_j^{(2)}\gamma_{j+1}^{(2)}}-\frac{d}{n}\right).\]

For each edge $\{v,w\}$ in the complete graph, call $\{v,w\}$ single if it is
visited exactly once by the pair of paths $(\gamma^{(1)},\gamma^{(2)})$.
For each $i=1,2$, denote
\[J_i:=J_i(\gamma^{(1)},\gamma^{(2)})
=\{j \in \{0,\ldots,m-1\}:\{\gamma^{(i)}_j,\gamma^{(i)}_{j+1}\}
\text{ is single}\},\]
and write $J_i^c=\{0,\ldots,m-1\} \setminus J_i$.
Distributing the products over $J_1^c$ and $J_2^c$, the above may be written as
\begin{align}
T_m&=\sum_{(\gamma^{(1)},\gamma^{(2)}) \in \Gamma_m}
\sum_{K_1 \subseteq J_1^c} \sum_{K_2 \subseteq J_2^c} \left(n-\left|\left\{
\gamma^{(1)}_0,\gamma^{(1)}_1,\gamma^{(2)}_1\right\}\right|\right)
\left(-\frac{d}{n}\right)^{|J_1^c|-|K_1|+|J_2^c|-|K_2|}\nonumber\\
&\hspace{1in}\prod_{j \in J_1}
\left(\bA_{\gamma_j^{(1)}\gamma_{j+1}^{(1)}}-\frac{d}{n}\right)
\prod_{j \in K_1} \bA_{\gamma_j^{(1)}\gamma_{j+1}^{(1)}}
\prod_{j \in J_2}
\left(\bA_{\gamma_j^{(2)}\gamma_{j+1}^{(2)}}-\frac{d}{n}\right)
\prod_{j \in K_2} \bA_{\gamma_j^{(2)}\gamma_{j+1}^{(2)}}.
\label{eq:Tmexpansion}
\end{align}

Let $l:=l(n) \geq 1$ with $l \ll \log n$, to be specified later,
and let $\cE$ be the event that $G$ is $l$-coil-free.
We multiply (\ref{eq:Tmexpansion}) on both sides by $\1\{\cE\}$
and apply Lemma \ref{lemma:alternatingprobbound}.
On $\cE$, the only nonzero terms of the sum in
(\ref{eq:Tmexpansion}) are those where $\gamma^{(1)}$ is $l$-coil-free on $K_1$
and $\gamma^{(2)}$ is $l$-coil-free on $K_2$. For each such term,
denote by $S$ the set of single edges and by $T$ the set of all edges visited 
by $\gamma^{(1)}$ on $K_1$ and by $\gamma^{(2)}$ on $K_2$. Then
\begin{align*}
&\prod_{j \in J_1}
\left(\bA_{\gamma_j^{(1)}\gamma_{j+1}^{(1)}}-\frac{d}{n}\right)
\prod_{j \in K_1} \bA_{\gamma_j^{(1)}\gamma_{j+1}^{(1)}}
\prod_{j \in J_2}
\left(\bA_{\gamma_j^{(2)}\gamma_{j+1}^{(2)}}-\frac{d}{n}\right)
\prod_{j \in K_2} \bA_{\gamma_j^{(2)}\gamma_{j+1}^{(2)}}\\
&\hspace{1in}=\prod_{\{v,w\} \in S} \left(\bA_{vw}-\frac{d}{n}\right)
\prod_{\{v,w\} \in T} \bA_{vw}.
\end{align*}
Note that $|S|=|J_1|+|J_2|$. If $(\gamma^{(1)},\gamma^{(2)})$ visit
$e:=e(\gamma^{(1)},\gamma^{(2)})$ total
distinct edges, then $e-|J_1|-|J_2|$ of these are non-single and hence
potentially visited by $\gamma^{(1)}$ on $K_1$ and $\gamma^{(2)}$ on
$K_2$. For each such edge \emph{not} visited by $\gamma^{(1)}$ on $K_1$ and
$\gamma^{(2)}$ on $K_2$, there must be at least two indices corresponding to
this edge in $J_1^c \setminus K_1$ and $J_2^c \setminus K_2$. Hence
$|T| \geq e-|J_1|-|J_2|-\frac{|J_1^c|-|K_1|+|J_2^c|-|K_2|}{2}$. The cycle number
$\#_c(S \cup T)$ is at most the cycle number of the graph formed by all edges
visited by $(\gamma^{(1)},\gamma^{(2)})$, which is $e-v+1$ if
$(\gamma^{(1)},\gamma^{(2)})$ visit $v:=v(\gamma^{(1)},\gamma^{(2)})$
total distinct vertices. Combining these observations and applying Lemma
\ref{lemma:alternatingprobbound}, for a constant $C>0$ and all large $n$,
\[\E[T_m\1\{\cE\}]
\leq Cn(\log n)^2\sum_{\gamma^{(1)},\gamma^{(2)},K_1,K_2}
\left(\frac{d}{n}\right)^{e+\frac{|J_1^c|-|K_1|+|J_2^c|-|K_2|}{2}}
n^{0.7(e-v+1)}\left(\frac{2}{n^{0.7/l}}\right)^{|J_1|+|J_2|},\]
where the summation is over all $(\gamma^{(1)},\gamma^{(2)}) \in \Gamma_m$ and
$K_1 \subseteq J_1^c$ and $K_2 \subseteq J_2^c$ such that $\gamma^{(i)}$ is
$l$-coil-free on $K_i$ for $i=1,2$, and where $e,v,J_1,J_2$ all depend on the
paths $\gamma^{(1)}$ and $\gamma^{(2)}$.
As $\gamma^{(1)}$ and $\gamma^{(2)}$ are of length $m$, this implies
$2(e-|J_1|-|J_2|)+|J_1|+|J_2| \leq 2m$, so
$e \leq m+(|J_1|+|J_2|)/2$. Then, since $d>1$,
\begin{align*}
\E[T_m\1\{\cE\}] \leq Cn(\log n)^2d^m
\sum_{\gamma^{(1)},\gamma^{(2)},K_1,K_2} n^{-e}
\left(\frac{d}{n}\right)^{\frac{|J_1^c|-|K_1|+|J_2^c|-|K_2|}{2}}
n^{0.7(e-v+1)}\left(\frac{2\sqrt{d}}{n^{0.7/l}}\right)^{|J_1|+|J_2|}.
\end{align*}

Let us now drop the condition that $J_i$ corresponds to indices where
$\{\gamma_j^{(i)},\gamma_{j+1}^{(i)}\}$ is single, and instead sum over all
subsets $J_1,J_2 \subseteq \{0,\ldots,m-1\}$, all subsets $K_1 \subseteq J_1^c$
and $K_2 \subseteq J_2^c$, and all paths $(\gamma^{(1)},\gamma^{(2)}) \in
\Gamma_m$ such that $\gamma^{(i)}$ is $l$-coil-free on $K_i$ for $i=1,2$.
Letting $\cW(m,l,v,e,K_1,K_2)$ be as in Lemma \ref{lemma:equivclassbound}, and
noting that each class $\cW(m,l,v,e,K_1,K_2)$ represents at most
$n^v$ distinct pairs of paths, this yields
\begin{align*}
\E[T_m\1\{\cE\}]&\leq Cn(\log n)^2d^m \sum_{J_1,J_2,K_1,K_2} \sum_{v=2}^{2m}
\sum_{e=v-1}^{2m} n^{-e}
\left(\frac{d}{n}\right)^{\frac{|J_1^c|-|K_1|+|J_2^c|-|K_2|}{2}}\\
&\hspace{1in}n^{0.7(e-v+1)}\left(\frac{2\sqrt{d}}{n^{0.7/l}}
\right)^{|J_1|+|J_2|}n^v|\cW(m,l,v,e,K_1,K_2)|\\
&\leq Cn^2(\log n)^2d^m \sum_{v=2}^{2m}\sum_{e=v-1}^{2m}
\sum_{J_1,J_2,K_1,K_2} n^{-0.3(e-v+1)}
\left(\frac{d}{n}\right)^{\frac{|J_1^c|-|K_1|+|J_2^c|-|K_2|}{2}}\\
&\hspace{1in}\left(\frac{2\sqrt{d}}{n^{0.7/l}}\right)^{|J_1|+|J_2|}
\left(l(3v^2)^{2l+2}\right)^{2m-|K_1|-|K_2|}\left(l(3v^2)^{2e-2v+4}
\right)^{\frac{2m}{l}+2}\\
&=Cn^2(\log n)^2d^m \sum_{v=2}^{2m}\sum_{e=v-1}^{2m}
\sum_{J_1,J_2,K_1,K_2} \left(l(3v^2)^2\right)^{\frac{2m}{l}+2}
\left(\frac{(3v^2)^{\frac{4m}{l}+4}}{n^{0.3}}\right)^{e-v+1}\\
&\hspace{1in}\left(\frac{\sqrt{d}\,l(3v^2)^{2l+2}}{\sqrt{n}}
\right)^{|J_1^c|-|K_1|+|J_2^c|-|K_2|}
\left(\frac{2\sqrt{d}\,l(3v^2)^{2l+2}}{n^{0.7/l}}\right)^{|J_1|+|J_2|},
\end{align*}
where the summations are over $J_1,J_2 \subseteq \{0,\ldots,m-1\}$ and $K_i
\subseteq J_i^c$ for $i=1,2$, and in the last line above we have
written $2m=|J_1|+|J_2|+|J_1^c|+|J_2^c|$ and
collected terms with common exponents. Factoring the summations over
$J_1,J_2,K_1,K_2$, the above is equivalent to
\begin{align*}
\E[T_m\1\{\cE\}]
&\leq Cn^2(\log n)^2d^m \sum_{v=2}^{2m}\sum_{e=v-1}^{2m}
\left(l(3v^2)^2\right)^{\frac{2m}{l}+2}
\left(\frac{(3v^2)^{\frac{4m}{l}+4}}{n^{0.3}}\right)^{e-v+1}\\
&\hspace{2in}\left(1+\frac{\sqrt{d}\,l(3v^2)^{2l+2}}{\sqrt{n}}
+\frac{2\sqrt{d}\,l(3v^2)^{2l+2}}{n^{0.7/l}}\right)^{2m}.
\end{align*}

Finally, let us take $l \sim (\log\log n)^3$ and $m \sim (\log n)(\log \log n)$.
Then for $v \leq 2m$, we may verify
\[m\sqrt{d}\,l(3v^2)^{2l+2} \ll n^{0.7/l} \ll \sqrt{n},\;\;
(3v^2)^{\frac{4m}{l}+4} \ll n^{0.2}.\]
So for some $C'>0$ and all large $n$,
\[\E[T_m\1\{\cE\}] \leq C'n^2(\log n)^2d^m\sum_{v=2}^{2m}
(l(3v^2)^2)^{\frac{2m}{l}+2} \leq C'n^2(\log n)^2d^m
\cdot 2m(144lm^4)^{\frac{2m}{l}+2}.\]
We may verify, for any $\eps>0$,
\[n^2(\log n)^2 \cdot 2m(144lm^4)^{\frac{2m}{l}+2} \ll (1+\eps)^{2m}.\]
Noting that $\rho(\bB)^{2m} \leq T_m$, we obtain by Markov's inequality
\[\P\left[\rho(\bB) \geq \sqrt{d}(1+\eps), G \text{ is }
l\text{-coil-free}\right]
\leq \frac{\E[T_m\1\{\cE\}]}{d^m(1+\eps)^{2m}} \to 0.\]
Lemma \ref{lemma:coilfree} yields
$\P[G \text{ is not } l\text{-coil-free}] \to 0$, establishing
the desired result.
\end{proof}

\section{Lower bounds: Theorems \ref{thm:ER}, \ref{thm:Harmonic},
\ref{thm:SBM}, and \ref{thm:SBMlocal}}\label{sec:Lower}

In this section, we prove the lower bound in Theorem \ref{thm:ER}, and also
Theorems \ref{thm:Harmonic}, \ref{thm:SBM}, and \ref{thm:SBMlocal}.
We will begin by reviewing some definitions and facts on local weak convergence in Section
\ref{sec:LWC}, describing the sense in which the graphs $\sG(n,d/n)$ and
$\sG(n,a/n,b/n)$ converge locally to Galton-Watson trees. We will then reduce
the lower bounds to the construction of local algorithms on such trees, stated
as a sequence of lemmas in Section \ref{sec:KeyLemmas}, and finally turn to the
proofs of these lemmas.
\subsection{Local weak convergence: Definitions }
\label{sec:LWC}

To accommodate notationally both the \ER model and the stochastic-block-model,
let $\cM$ denote a general finite set of possible vertex labels.
In the \ER case, we will simply take the trivial set $\cM=\{1\}$; for the
stochastic-block-model with partially revealed labels,
we will take $\cM=\{+1,-1,u\}$ as described in Section
\ref{sec:results}. Then Definitions \ref{def:Local} and \ref{def:LocalMarked}
coincide.

Let $\cG(\cM)$ denote the space of tuples $(G,\root,\bsigma)$
where $G=(V,E)$ is a (finite or) locally-finite
graph, $\root \in V$ is a distinguished root vertex, and $\bsigma:V \to \cM$
associates a label to each vertex. Define a corresponding edge-perspective
set $\cG_e(\cM)$ as the space of 
tuples $(G,\{\root,\root'\},\bsigma)$ where $G=(V,E)$ is a (finite or)
locally-finite graph, $\bsigma:V \to \cM$ associates a label to each vertex,
and $\{\root,\root'\} \in E$ is a distinguished (undirected) root edge.
The subspaces of trees in $\cG(\cM)$, $\cG_e(\cM)$ are denoted by
$\cT(\cM)$, $\cT_e(\cM)$.

For any graph $H$, integer $\ell\geq 0$, and set of
vertices $S$ in $H$, let $\Ball_{\ell}(S;H)$ denote the subgraph induced by
all vertices at distance at most $\ell$ from $S$ in $H$
(including $S$ itself).  To make contact with previously introduced notation, for a
vertex $v$, we write $\Ball_\ell(v;H)$ for $\Ball_\ell(\{v\};H)$.

Local weak convergence  was initially introduced in \cite{benjaminischramm} and further developed in \cite{aldous2007processes,aldoussteele}.
We define it here in a somewhat restricted setting that is relevant for our proofs.
\begin{definition}
Let $\nu$ be a law over $\cT(\cM)$ and
let $\{G_n=(V_n,E_n)\}_{n\geq 1}$ be a sequence of (deterministic) graphs
with (deterministic) vertex labels $\bsigma_n:V_n \to \cM$.
We say that $(G_n,\bsigma_n)$ {\bf converges locally} to $\nu$ if, for any $\ell \geq 0$,
any $\tau \in \cT(\cM)$, and a vertex $v \in V_n$ chosen uniformly at random\footnote{Here
$\bsigma_n$ and $\bsigma$ really denote the restrictions
$\bsigma_n|_{\Ball_{\ell}(v;G_n)}$ and $\bsigma|_{\Ball_{\ell}(\root;T)}$ of the labels to the balls
of radius $\ell$. We avoid this type of cumbersome notation when the meaning is 
clear.},
\begin{align}
\lim_{n\to\infty} \P\{(\Ball_\ell(v;G_n),v,\bsigma_n) \simeq \tau\}
=\P_\nu\{(\Ball_\ell(\root;T),\root,\bsigma) \simeq \tau\},
\end{align}
where $\simeq$ denotes graph isomorphism that preserves the root vertex
and vertex labels. We write $(G_n,\bsigma_n) \toloc \nu$.
\end{definition}

A law $\nu$ over $\cT(\cM)$ is the limit of some graph sequence if and only
if $\nu$ is \emph{unimodular} \cite{benjaminischramm,elek,benjaminietal}.
Roughly speaking, this means that the law
$\nu$ does not change if the root is changed. Corresponding to any unimodular
law $\nu$ is an associated edge-perspective law $\nu_e$ over $\cT_e(\cM)$. This is obtained from $\nu$ 
as follows. First define a  law $\tnu$ over $\cT(\cM)$ whose Radon-Nykodym derivative
with respect to $\nu$ is $\frac{\de\tnu}{\de \nu}(T,\nu,\bsigma) = \deg_T(\root)/\E_{\nu}\deg_T(\root)$. Then, letting $(T,\root,\bsigma)\sim\tnu$,
define $\nu_e$ to be the law of $(T,\{\root,v\},\bsigma)$ where $v$ is a uniformly random neighbor of $\root$ in $T$.

The above definition is clarified by the following fact. Its proof is an immediate consequence of the definitions,
once we notice that, in order to sample a uniformly random edge in $G_n$, it is sufficient to sample a random vertex 
$v$ with probability proportional to $\deg(v)$, and then sample one of its neighbors uniformly at random.
\begin{lemma}\label{lemma:edgeperspective}
For $\nu$ a unimodular law over $\cT(\cM)$, denote by $\nu_e$ the corresponding edge-perspective law.
Let $\{G_n=(V_n,E_n)\}_{n \geq 1}$ be a graph sequence with vertex labels $\bsigma_n:V_n \to \cM$ such that
$(G_n,\bsigma_n) \toloc \nu$. Then for any $\ell \geq 0$ and any $\tau_e \in
\cT_e(\cM)$, if an edge $\{u,v\} \in E_n$ is chosen uniformly at random, we have
\begin{align}
\lim_{n \to \infty} \P\big\{(\Ball_{\ell}(\{u,v\};G_n),\{u,v\},\sigma_n) \simeq \tau_e\big\}
=\P_{\nue}\big\{(\Ball_{\ell}(\{\root,\root'\};T),\{\root,\root'\},\sigma) \simeq \tau_e\big\}\, ,
\end{align}
where $\simeq$ denotes graph isomorphism that preserves the root edge and
vertex labels.
\end{lemma}

Both the \ER random graph and the planted partition random graph with partially
observed labels (revealed independently at random) satisfy the
above definitions, where the laws $\nu$ and $\nue$ are the laws of 
Galton-Watson trees.
\begin{definition}
A {\bf Galton-Watson tree} with offspring distribution $\mu$ is a random tree
rooted at a vertex $\root$, such that each vertex $v$
has $N_v \sim \mu$ children independently of the other
vertices. 

A {\bf two-type Galton-Watson tree} with offspring distributions
$\mu^=$ and $\mu^{\neq}$ is a random tree with binary vertex labels $\{+1,-1\}$
rooted at $\root$, such that $\root$ has label $\pm 1$ with
equal probability, and each vertex has $N_v^= \sim \mu^=$ children with same
label as itself and $N_v^{\neq} \sim \mu^{\neq}$ children with opposite label
from itself, independently of each other and of the other vertices.

The labels of a two-type Galton-Watson tree are {\bf partially revealed with
probability $\pmb{\delta}$} if the label set is augmented to $\{+1,-1,u\}$ and,
conditional on the tree, the label of each vertex is replaced by $u$
independently with probability $1-\delta$.
\end{definition}
\begin{example}\label{example:ER}
Fix $d>0$, let $G_n=(V_n,E_n) \sim \sG(n,d/n)$ be an \ER random graph,
let $\cM=\{1\}$, and let $\bsigma_n \equiv 1$ the trivial labeling.
Then almost surely (over the realization of $G_n$),
$(G_n,\bsigma_n) \toloc \nu$ where $\nu$ is the law of a Galton-Watson tree
rooted at $\root$ with offspring distribution $\Pois(d)$
(and labels $\bsigma \equiv 1$). 

The associated edge-perspective law $\nue$ is
the law of two independent such trees rooted at $\root$ and $\root'$ and
connected by the single edge $\{\root,\root'\}$.
\end{example}
\begin{example}\label{example:SBM}
Fix $a,b>0$ and $\delta \in (0,1]$, let $G_n=(V_n,E_n) \sim
\sG(n,a/n,b/n)$ be the planted partition random graph, and let $\bsigma_n:V_n
\to \{+1,-1,u\}$ be such that, independently for each vertex $i$,
with probability $1-\delta$ we have $\sigma_n(i)=u$, and with probability
$\delta$ we have that $\sigma_n(i)$ equals the vertex
label ($+1$ or $-1$) of the hidden partition to which $i$ belongs.
Then almost surely (over the realization of $G_n$ and $\bsigma_n$),
$(G_n,\bsigma_n) \toloc \nu$ where $\nu$ is the law of a two-type
Galton-Watson tree rooted at $\root$, with
offspring distributions $\Pois(a/2)$ and $\Pois(b/2)$ and with vertices
partially revealed with probability $\delta$.

The associated edge-perspective law $\nue$ is
the law of two such trees rooted at $\root$ and $\root'$ and connected by the
single edge $\{\root,\root'\}$, where $\root$ and $\root'$ belong to the same
side of the partition with probability $a/(a+b)$ and to opposite sides of the
partition with probability $b/(a+b)$, and the trees are
independent conditional on the partition memberships of $\root$ and $\root'$.
\end{example}

As in Section \ref{sec:results}, to define local algorithms that solve the SDP
(\ref{eq:SDP}), we extend our definitions to include additional random
real-valued marks.
Namely, we denote by $\cG^*(\cM)$ the space of tuples
$(G,\root,\bsigma,\bz)$ where $(G,\root,\bsigma) \in \cG(\cM)$ and
$\bz:V(G) \to \R$ associates a real-valued mark to each vertex of $G$.
The spaces $\cG_e^*(\cM)$, $\cT^*(\cM)$, $\cT_e^*(\cM)$ are defined analogously.

For a unimodular law $\nu$ over $\cT(\cM)$, we let $\nu^*$ be the law over
$\cT^*(\cM)$ such that $(T,\root,\bsigma,\bz) \sim \nu^*$ if
$(T,\root,\bsigma) \sim \nu$ and, conditional on $(T,\root,\bsigma)$,
$z(i) \overset{iid}{\sim} \Normal(0,1)$ for all vertices $i \in V(T)$.
$\nue^*$ is defined analogously.

\begin{remark}
Since we are interested in graph sequences that converge locally to trees, it will turn out to 
be sufficient to define local algorithms $F:\cG^*(\cM)\to\reals$ on trees  and, for instance, extend it arbitrarily to other 
graphs. With a slight abuse of notation, we will therefore write $F:\cT^*(\cM)\to\reals$.
\end{remark}

Finally, given a local algorithm $F\in\cF_*^\cM(\ell)$ and a unimodular
probability measure $\nu$ on $\cT(\cM)$, we define the value of $F$
with respect to $\nu$ as
\begin{align}\label{eq:EFnu}
\cE(F,\nu):=d\, \E_{\nue^*}\big\{F(\Ball_{\ell}(\root;T),\root,\bsigma,\bz)
F(\Ball_{\ell}(\root';T),\root',\bsigma,\bz)\big\}\, ,
\end{align}
where $d=\E_{\nu}\{\deg(\root)\}$ is the expected degree of the root under $\nu$.
(This is a slight abuse of notation, given the definitions of $\cE(F;G)$ and
$\cE(F;G,\bsigma)$ in (\ref{eq:ValueDef}) and (\ref{eq:ValueDefMarked}).)

\subsection{Key lemmas}
\label{sec:KeyLemmas}
Using the above framework, the desired lower bounds
are now consequences of the following results.
\begin{lemma}\label{lemma:localalglowerbound}
Let $\{G_n=(V_n,E_n)\}_{n \geq 1}$ be a deterministic sequence of graphs with
deterministic vertex marks
$\bsigma_n:V_n \to \cM$, such that $|V_n|=n$, $|E_n|/n \to d/2$
for a constant $d>0$, and $(G_n,\bsigma_n) \toloc \nu$ for a law $\nu$ on
$\cT(\cM)$.

For fixed $\ell \geq 0$, let $F \in \cF_*^\cM(\ell)$ be any radius-$\ell$  local
algorithm such that the following two conditions 
hold \footnote{The second of these is the same as
condition 2 of Definition \ref{def:LocalMarked} and condition 2 of
Definition \ref{def:Local} in the case of trivial markings $\cM=\{1\}$;
we restate it here for convenience.}:
\begin{equation}\label{eq:Fconditions}
\E_{\nu^*}\big\{F(T,\root,\bsigma,\bz)\big\}=0,\;\;\;\;
\E_{\nu^*}\big\{F(T,\root,\bsigma,\bz)^2 \mid T,\root,\bsigma\big\} \equiv 1\, .
\end{equation}
Then we have
\begin{align}
\lim_{n \to \infty} \cE(F;G_n,\bsigma_n) \geq \cE(F,\nu)\, .
\end{align}
\end{lemma}

\begin{lemma}\label{lemma:ERlowerbound}
Fix $d>1$, $\cM=\{1\}$, and let $\nu$ be the law of the Galton-Watson tree
with offspring distribution $\Pois(d)$ (and trivial marking $\bsigma \equiv 1$).
Then there exist local algorithms $F_\ell \in
\cF_*^\cM(\ell)$ for $\ell \geq 1$ satisfying (\ref{eq:Fconditions}) and
such that
\begin{align}
\lim_{\ell \to \infty} \cE(F_\ell,\nu)
\geq 2\sqrt{d}\left(1-\frac{1}{d+1}\right)\, .
\end{align}
\end{lemma}

\begin{lemma}\label{lemma:ERlowerboundharmonic}
In the same setup as Lemma \ref{lemma:ERlowerbound}, there
exist local algorithms $F_{\ell,L} \in \cF_*^\cM(L)$ for $L \geq \ell \geq 1$
satisfying (\ref{eq:Fconditions}) and such that
\begin{align}
\lim_{\ell \to \infty} \lim_{L \to \infty} \cE(F_{\ell,L},\nu)
\geq d\,\E\Psi(\cond_1,\cond_2)\,,
\end{align}
where $\Psi(\cond_1,\cond_2)$ is as in (\ref{eq:Psi}).
\end{lemma}

\begin{lemma}\label{lemma:SBMlowerbound}
Fix $a,b>0$ such that $d:=(a+b)/2 \geq 2$
and $\lambda:=(a-b)/\sqrt{2(a+b)}>1$. Fix
$\delta \in (0,1]$, let $\cM=\{+1,-1,u\}$, and let $\nu$ be the law of the
two-type Galton-Watson tree with
offspring distributions $\Pois(a/2)$ and $\Pois(b/2)$ and with vertex labels
partially revealed with probability $\delta$. Then for a universal constant
$C>0$, there exist local algorithms $F_\ell \in
\cF_*^{\cM}(\ell)$ for $\ell \geq 1$ satisfying (\ref{eq:Fconditions}) and
such that
\begin{align}
\lim_{\ell \to \infty} \cE(F_\ell,\nu) \geq \sqrt{d}\left(2+\frac{(\lambda-1)^2}{\lambda\sqrt{d}}-\frac{C}{d}\right)\,.
\end{align}
\end{lemma}
Proofs of the above four lemmas are contained in the next four subsections. Let
us first show that these lemmas imply the desired lower bounds.

\begin{proof}[Proof of Theorem \ref{thm:ER} (lower bound) and Theorems
\ref{thm:Harmonic}, \ref{thm:SBM}, and \ref{thm:SBMlocal}]
Consider models $\sG(n,d/n)$ and $\sG(n,a/n,b/n)$ with $d:=(a+b)/2$. Then (\ref{eq:locallowermain})
and (\ref{eq:locallowerharmonic}) follow from Example \ref{example:ER},
Lemma \ref{lemma:localalglowerbound}, Lemma \ref{lemma:ERlowerbound}, and
Lemma \ref{lemma:ERlowerboundharmonic}, while
(\ref{eq:locallowerSBM}) follows from Example \ref{example:SBM},
Lemma \ref{lemma:localalglowerbound}, and Lemma \ref{lemma:SBMlowerbound}.
The bounds (\ref{eq:UpperLowerMain}), (\ref{eq:lowerharmonic}),
and the second bound of (\ref{eq:lowerSBM}) in the case $d \geq 2$
follow in turn from
(\ref{eq:locallowermain}), (\ref{eq:locallowerharmonic}), and
(\ref{eq:locallowerSBM}), as any local algorithm defines a feasible solution
$\bX$ for the SDP (\ref{eq:SDP}) which achieves the SDP value
$n\,\cE(F;G_n,\bsigma_n)$, as discussed in Section \ref{sec:results}.
For the second bound of (\ref{eq:lowerSBM}) in the case $d \in (1,2)$,
we may take $C>4$ so that
$\lambda>2+(\lambda-1)^2/(\lambda\sqrt{d})-C/d$ always when $d \in (1,2)$,
and hence the first bound dominates in (\ref{eq:lowerSBM}).
For the first bound of (\ref{eq:lowerSBM}) and any $d>1$,
let us simply consider the feasible
point $\bX=\bsigma_n\bsigma_n^{\sT}$ for (\ref{eq:SDP}), where $\bsigma_n \in
\{+1,-1\}^n$ is the indicator vector of the hidden partition. Then
\[\frac{1}{n}\SDP(\bA) \geq \frac{1}{n}\langle \bA-\frac{d}{n}\e\e^{\sT},\bX
\rangle=\frac{1}{n}\sum_{i,j \in V_n} \bA_{ij}\sigma_n(i)\sigma_n(j).\]
From the definition of $\sG(n,a/n,b/n)$, we obtain almost surely
\[\liminf_{n \to \infty} \frac{1}{n}\SDP(\bA)
\geq \frac{1}{n}\left(\frac{n^2}{2} \cdot \frac{a}{n}-\frac{n^2}{2} \cdot
\frac{b}{n}\right)=\frac{a-b}{2}=\lambda\sqrt{d}.\]
Finally, the large $d$ expansion (\ref{eq:ExpansionCond}) in Theorem
\ref{thm:Harmonic} is proven in Appendix \ref{app:ExpansionCond}.
\end{proof}

In the remainder of this section, we establish Lemmas
\ref{lemma:localalglowerbound}, \ref{lemma:ERlowerbound},
\ref{lemma:ERlowerboundharmonic}, and \ref{lemma:SBMlowerbound}.

\subsection{Proof of Lemma \ref{lemma:localalglowerbound}}

We first recall some well-known properties of locally convergent
graphs. (Short proofs are provided for the reader's convenience.)
\begin{lemma}\label{lemma:maxballsize}
Let $(G_n,\bsigma_n)\toloc \nu$ for any law $\nu$, where $G_n=(V_n,E_n)$ and
$|V_n|=n$. Denote by $|\Ball_\ell(v;G_n)|$ the number of vertices in $\Ball_\ell(v;G_n)$.
Then for any fixed $\ell \geq 0$,
\[\lim_{n\to\infty} \frac{1}{n}\max_{v\in V_n}|\Ball_\ell(v;G_n)|=0.\]
\end{lemma}
\begin{proof}
Suppose by contradiction that the claim is false. Then there exist $\eps>0$,
a sequence of graph sizes $\{n_i\}$, and vertices $v_i \in
G_{n_i}$ for which $|\Ball_\ell(v_i;G_{n_i})|\geq n\eps$. In particular the maximum
degree of any vertex in $\Ball_\ell(v_i;G_{n_i})$ is at least $n^\delta$ for some
$\delta>0$. Hence, for any $w \in \Ball_\ell(v_i;G_{n_i})$, the maximum degree
of any vertex in $\Ball_{2\ell}(w;G_{n_i})$ is at least $n^{\delta}$. Since there are
at least $n\eps$ such vertices $w$, we have, for $w$ a vertex of $G_n$ chosen
uniformly at random,
\begin{align}
\limsup_{n\to\infty} \P\Big(\max\big\{\deg(v):v\in \Ball_{2\ell}(w;G_n)\big\}
\geq n^{\delta} \Big) \geq \eps.
\end{align}
This contradicts the hypothesis that $(\Ball_{2\ell}(w;G_{n}),w,\bsigma_n)$
converges in law to $(\Ball_{2\ell}(\root;T),\root,\bsigma)$ where $(T,\root,\bsigma)
\sim \nu$.
\end{proof}

\begin{lemma}\label{lemma:vertexconvergence}
Suppose $(G_n,\bsigma_n) \toloc \nu$, where $G_n=(V_n,E_n)$. For any fixed $\ell
\geq 0$, let $f(H_{\ell},\root,\bsigma)$ be any bounded function of a graph $H_{\ell}$
with root vertex $\root$ and vertex marks $\bsigma$,
such that each vertex of $H_{\ell}$ is at distance at most $\ell$ from $\root$.
Then as $n \to \infty$,
\begin{align}
\frac{1}{|V_n|} \sum_{v \in V_n} f(\Ball_{\ell}(v;G_n),v,\bsigma_n) \to
\E_{\nu}[f(\Ball_\ell(\root;T),\root,\bsigma)].
\end{align}
\end{lemma}
\begin{proof}
Let $v \in V_n$ be a vertex chosen uniformly at random. Then by the assumption
of local weak convergence, $f(\Ball_{\ell}(v;G_n),v,\sigma_n)$ is a random
variable that converges in law to $f(\Ball_{\ell}(\root;T),\root,\sigma)$ where
$(T,\root,\sigma) \sim \nu$, and the
conclusion follows from the bounded convergence theorem.
\end{proof}

\begin{lemma}\label{lemma:edgeconvergence}
Suppose $(G_n,\bsigma_n) \toloc \nu$, where $G_n=(V_n,E_n)$. Let $\nue$ be the
edge-perspective law associated to $\nu$. For any fixed $\ell
\geq 0$, let $f(H_{\ell},\{\root,\root'\},\sigma)$ be any bounded function of a graph
$H_l$ with root edge $\{\root,\root'\}$ and vertex marks $\sigma$, such that
each vertex of $H_{\ell}$ is at distance at most $\ell$ from $\root$ or $\root'$.
Then as $n \to \infty$,
\begin{align}
\frac{1}{|E_n|} \sum_{\{v,w\} \in E_n} f(\Ball_{\ell}(\{v,w\};G_n),\{v,w\},\sigma_n)
\to
  \E_{\nue}[f(\Ball_{\ell}(\{\root,\root'\};T),\{\root,\root'\},\sigma)].
\end{align}
\end{lemma}
\begin{proof}
The proof is the same as Lemma \ref{lemma:vertexconvergence}; we let
$\{v,w\} \in E_n$ be an edge chosen uniformly at random, and apply Lemma
\ref{lemma:edgeperspective} and the bounded convergence theorem.
\end{proof}

\begin{proof}[Proof of Lemma \ref{lemma:localalglowerbound}]
For notational convenience, let us denote $\E_{\bz}$ simply by $\E$ (so
expectations are understood to be with respect to $\bz$ only).
Given a local algorithm $F:\cT_*^\cM(\ell) \to \reals$ defined on trees,
augment it to $F:\cG_*^\cM(\ell) \to \reals$ defined on all graphs by setting
$F(\Ball_\ell(\root;G),\root,\bsigma,\bz)=z(\root)$ if $\Ball_\ell(\root,G)$
is not a tree. Note that this satisfies the conditions of Definitions
\ref{def:Local} and \ref{def:LocalMarked}.

Define $\xi(i)=F(G_n,i,\bsigma_n,\bz_n)$.
To bound the value $\cE(F;G_n,\bsigma_n)$, let us write
\[\E\left[\sum_{i,j \in V_n} (\bA_{G_n})_{ij}\xi(i)\xi(j)\right]
=2\sum_{\{i,j\} \in E_n} \E[\xi(i)\xi(j)].\]
For any vertices $i,j \in V_n$, $|\E[\xi(i)\xi(j)]| \leq 1$ by condition
(\ref{eq:Fconditions}). Furthermore,
for each edge $\{i,j\} \in E_n$, $\E[\xi(i)\xi(j)]$ is a function only of
$\{i,j\}$,
the local neighborhood $\Ball_\ell(\{i,j\};G_n)$, and the marks $z(v)$
of vertices $v$ in this
neighborhood. Then Lemma \ref{lemma:edgeconvergence} and the assumption
$(G_n,\bsigma_n) \toloc \nu$ implies
\[\frac{1}{2|E_n|}\E\left[\sum_{i,j \in V_n}
(\bA_{G_n})_{ij}\xi(i)\xi(j)\right] \to \cE(F,\nu)\]
as $n \to \infty$.

Next, note that if $i \notin \Ball_{2\ell}(j;G_n)$, then $\xi(i)$ and $\xi(j)$
are independent by construction. Hence
\begin{align*}
\E\left[\sum_{i,j \in V_n} \xi(i)\xi(j)\right]
&=\sum_{i,j \in V_n} \E[\xi(i)]\E[\xi(j)]+\sum_{i \in V_n}
\sum_{j \in \Ball_{2l}(i;G_n)} (\E[\xi(i)\xi(j)]-\E[\xi(i)]\E[\xi(j)])\\
&\leq \left(\sum_{i \in V_n} \E[\xi(i)]\right)^2
+2n\max_{i \in V_n} |\Ball_{2\ell}(i;G_n)|.
\end{align*}
For each vertex $i \in V_n$, $|\E[\xi(i)]| \leq 1$ and $\E[\xi(i)]$ is a
function
only of $i$, the ball $\Ball_\ell(i;G_n)$, and the marks of vertices in this ball. Then
Lemma \ref{lemma:vertexconvergence} and the first condition in
(\ref{eq:Fconditions}) implies $\frac{1}{n}\sum_{i \in V_n} \E[\xi(i)] \to 0$.
Together with Lemma \ref{lemma:maxballsize}, this implies
\[\frac{d}{n^2}\E\left[\sum_{i,j \in V_n}\xi(i)\xi(j)\right] \to 0.\]
Combining the above and applying $|E_n|/n \to d/2$ yields the desired result.
\end{proof}

\subsection{Proof of Lemma \ref{lemma:ERlowerbound}}
For any rooted tree $T$ and each vertex $v$ of $T$, let $k(v):=\dist(v,\root)$
denote the
distance from $v$ to the root $\root$. For each $\ell \geq 0$, denote
\begin{equation}\label{eq:Xl}
N_\ell:=|\{v:k(v)=\ell\}|,\;\;\;\;X_\ell:=d^{-\ell}N_\ell.
\end{equation}

In the case where $T$ is a random Galton-Watson tree with offspring distribution
$\Pois(d)$, let $\cF_\ell$ be the $\sigma$-field generated by $N_0,\ldots,N_\ell$.
Then for each $\ell \geq 1$, conditional on $\cF_{\ell-1}$,
$X_\ell \sim d^{-\ell}\Pois(dN_{\ell-1})$, so
$\E[X_\ell \mid \cF_{\ell-1}]=X_{\ell-1}$. Hence $\{X_\ell\}_{\ell \geq 0}$ is a nonnegative
martingale with respect to the filtration $\{\cF_\ell\}_{\ell \geq 0}$, and
\begin{equation}\label{eq:X}
X:=\lim_{\ell \to \infty} X_\ell
\end{equation}
exists almost surely by the martingale convergence theorem, with $X \geq 0$.

We first establish the following lemma.
\begin{lemma}\label{lemma:ERlowerboundhelper}
Let $X,X'$ be independent random variables with law defined by (\ref{eq:X}),
where $X_\ell$ are defined by (\ref{eq:Xl}) for the Galton-Watson tree with
offspring distribution $\Pois(d)$.
In the setup of Lemma \ref{lemma:ERlowerbound}, for
each $\ell \geq 1$, there exists $F_\ell \in \cF_*^\cM(\ell)$
satisfying (\ref{eq:Fconditions}) such that
\[\lim_{\ell \to \infty} \cE(F_\ell,\nu)=\E[\us(X,X')],\]
where
\begin{equation}\label{eq:usER}
\us(X,X')=\begin{cases} d & X=X'=0 \\
\frac{\sqrt{d}(X+X')}{\sqrt{X+\frac{1}{d}X'}\sqrt{X'+\frac{1}{d}X}} & 
\text{otherwise}.\end{cases}
\end{equation}
\end{lemma}
\begin{proof}
As the vertex marking $\bsigma \equiv 1$ is trivial, for notational clarity we
omit it from all expressions below. Define the local algorithm
\begin{equation}\label{eq:FER}
F_\ell(T,\root,\bz):=
\begin{cases}
\frac{\sum_{v \in \Ball_\ell(\root;T)} d^{-k(v)/2}z(v)}
{\sqrt{\sum_{v \in \Ball_\ell(\root;T)} d^{-k(v)}}} & X_\ell>0, \\
\sign\left(\sum_{v \in \Ball_\ell(\root;T)} z(v)\right) & X_\ell=0,
\end{cases}
\end{equation}
where $k(v)$ and $X_\ell$ are defined for the tree $T$.
When $z(v) \overset{iid}{\sim} \Normal(0,1)$ conditional on $(T,\bsigma)$,
the conditions of (\ref{eq:Fconditions}) hold
by construction. It remains to compute $\cE(F_{\ell,L},\nu)$.

For $(T,\{\root,\root'\}) \in \cT_e$,
denote by $T_\root$ the subtree rooted at $\root$ of vertices
connected to $\root$ by a path not including $\root'$, and by $T_{\root'}$ the
subtree of remaining vertices rooted at $\root'$ (i.e.\ connected to $\root'$
by a path not including $\root$). Recall from Example \ref{example:ER} that if
$(T,\{\root,\root'\}) \sim \nue^*$, then $T_\root$ and $T_{\root'}$ are
independent Galton-Watson trees with offspring distribution $\Pois(d)$.
For each vertex $v \in T_\root$, denote by $k(v)$ its distance to $\root$, and
for each vertex $v' \in T_{\root'}$, denote by $k'(v')$ its distance to
$\root'$. Let $X_{\ell}$ be as defined in (\ref{eq:Xl}) for the subtree
$T_\root$ and $X_\ell'$ be as defined in (\ref{eq:Xl}) for the subtree
$T_{\root'}$.

For any $k \geq 0$, write as shorthand $\Ball_k:=\Ball_k(\root;T_\root)$
and $\Ball_k'=\Ball_k(\root';T_{\root'})$. Note that $\Ball_\ell(\root;T)$
consists of $\Ball_\ell$ and $\Ball_{\ell-1}'$ connected by the edge
$\{\root,\root'\}$, and similarly $\Ball_\ell(\root';T)$ consists of
$\Ball_\ell'$ and $\Ball_{\ell-1}$ connected by this edge.
We consider three cases: (I) If $X_{\ell-1}=0$ and $X_{\ell-1}'=0$,
then $\Ball_{\ell}(\root;T)=\Ball_{\ell}(\root';T)$ and the second case of
(\ref{eq:FER}) holds for both balls. In this case
\[F_{\ell}(T,\root,\bz)F_{\ell}(T,\root',\bz)=1.\]
(II) If $X_\ell=0$ and $X_{\ell-1}>0$ and $X_{\ell-1}'=0$, or
if $X_\ell'=0$ and $X_{\ell-1}'>0$ and $X_{\ell-1}=0$,
then the first case of (\ref{eq:FER}) holds for one of the balls
$\Ball_\ell(\root;T)$ or $\Ball_\ell(\root';T)$ and the second case holds for
the other ball. In this case we simply bound, using Cauchy-Schwarz,
\[|\E_{\bz}[F_{\ell}(T,\root,\bz)F_{\ell}(T,\root',\bz)]| \leq 1.\]
(III) Otherwise, the first case of (\ref{eq:FER}) holds for both balls
$\Ball_\ell(\root;T)$ and $\Ball_\ell(\root';T)$. Then we have
\begin{align*}
&\E_\bz\left[F_{\ell}(T,\root,\bz)F_{\ell}(T,\root',\bz)\right]\\
&=\frac{\sum_{v \in \Ball_{\ell-1}} d^{-k(v)/2}d^{-(k(v)+1)/2}
+\sum_{v' \in \Ball_{\ell-1}'} d^{-k'(v')/2}d^{-(k'(v')+1)/2}}
{\sqrt{\sum_{v \in \Ball_\ell} d^{-k(v)}+\sum_{v' \in \Ball_{\ell-1}'} d^{-(k'(v')+1)}}
\sqrt{\sum_{v' \in \Ball_\ell'} d^{-k'(v')}+\sum_{v \in \Ball_{\ell-1}} d^{-(k(v)+1)}}}.
\end{align*}
Letting $S_\ell=\sum_{j=0}^\ell X_\ell$ and
$S_\ell'=\sum_{j=0}^\ell X_\ell'$, the above may be written as
\[\E_\bz\left[F_\ell(T,\root,\bz)F_\ell(T,\root',\bz)\right]
=\frac{1}{\sqrt{d}} \frac{S_{\ell-1}+S_{\ell-1}'}{\sqrt{S_\ell+\frac{1}{d}S_{\ell-1}'}
\sqrt{S_\ell'+\frac{1}{d}S_{\ell-1}}}.\]

Combining the above three cases, taking the full expectation with respect to
$\nue^*$, and recalling the definition (\ref{eq:EFnu}),
\[\cE(F_\ell,\nu)=d\,\E\left[\1\{\mathrm{I}\}+\1\{\mathrm{II}\}
\E_{\bz}[F_{\ell}(T,\root,\bz)F_{\ell}(T,\root',\bz)]+\1\{\mathrm{III}\}
\frac{1}{\sqrt{d}} \frac{S_{\ell-1}+S_{\ell-1}'}{\sqrt{S_\ell+\frac{1}{d}S_{\ell-1}'}
\sqrt{S_\ell'+\frac{1}{d}S_{\ell-1}}}\right],\]
where $\1\{\mathrm{I}\}$, $\1\{\mathrm{II}\}$, and $\1\{\mathrm{III}\}$
indicate which of the above three cases occur. By convergence of C\'esaro sums,
\begin{align*}
&\lim_{\ell \to \infty} \tfrac{1}{\ell}S_{\ell-1}=
\lim_{\ell \to \infty} \tfrac{1}{\ell}S_{\ell}=\lim_{\ell \to \infty}
X_{\ell}=X,\\
&\lim_{\ell \to \infty} \tfrac{1}{\ell}S_{\ell-1}'=
\lim_{\ell \to \infty} \tfrac{1}{\ell}S_\ell'=\lim_{\ell \to \infty}
X_{\ell}'=X'
\end{align*}
almost surely, where $X$ and $X'$ are independent random variables with the law 
defined by (\ref{eq:X}) for the Galton-Watson tree. The events where
$T_{\root}$ has maximal depth exactly $\ell-1$ for $\ell=1,2,3,\ldots$ are
disjoint, and similarly for $T_{\root'}$, so $\1\{\mathrm{II}\} \to 0$ a.s.
Recall that for the
Galton-Watson tree, the extinction event $\lim_{\ell \to \infty} \1\{X_\ell=0\}$
equals $\1\{X=0\}$ a.s. Then $\1\{\mathrm{I}\} \to \1\{X=0,X'=0\}$ a.s., and
hence $\1\{\mathrm{III}\} \to \1\{X>0 \text{ or } X'>0\}$ a.s. Taking $\ell
\to \infty$ and applying the
bounded convergence theorem yields the desired result.
\end{proof}
\begin{proof}[Proof of Lemma \ref{lemma:ERlowerbound}]
Let $X,X'$ and $\us(X,X')$ be as in Lemma \ref{lemma:ERlowerboundhelper}, and
write $W=\sqrt{X+\frac{1}{d}X'}$ and $W'=\sqrt{X'+\frac{1}{d}X}$. Then
\[\us(X,X')=\begin{cases} d & W=W'=0,\\
\frac{d^{3/2}}{d+1}\frac{W^2+{W'}^2}{WW'} & W>0 \text{ and } W'>0 \end{cases}.\]
Conditional on the event $\cE:=\{W>0 \text{ and } W'>0\}$,
the bivariate law of $(W,W')$ is
exchangeable in $W$ and $W'$ by symmetry. Then applying Jensen's inequality,
\[\log \E[\tfrac{W}{W'} \mid \cE] \geq \E[\log \tfrac{W}{W'} \mid \cE]
=\E[\log W-\log W' \mid \cE]=0,\]
so $\E[\frac{W}{W'} \mid \cE] \geq 1$. Similarly
$\E[\frac{W'}{W} \mid \cE] \geq 1$, so
\[\E[\us(X,X')] \geq \frac{2d^{3/2}}{d+1}\P[\cE]+d\,\P[\cE^c]
\geq \frac{2d^{3/2}}{d+1}=2\sqrt{d}\left(1-\frac{1}{d+1}\right).\]
The result then follows from Lemma \ref{lemma:ERlowerboundhelper}.
\end{proof}

\subsection{Proof of Lemma \ref{lemma:ERlowerboundharmonic}}
The proof is similar to that of Lemma \ref{lemma:ERlowerboundhelper}.
For a rooted tree $(T,\root)$ and any vertex $v$ of $T$, let $k(v):=\dist(v,\root)$
denote the distance to the root. For any $k>0$, let us call vertices $v$ for
which $k(v)=k$ the `leaf vertices' of the ball $\Ball_k(\root;T)$.

For $L \geq \ell \geq 1$, if $\Ball_L(\root;T)$ has at least one leaf vertex
(i.e.\ $\{v:k(v)=L\}$ is non-empty),
then let us define a depth-$L$ approximation $h^{(\root,L)}(v;T)$ to the
harmonic measure introduced in
Section \ref{sec:harmonic}, as follows: For any vertex $v \in
\Ball_L(\root;T)$, let $k=k(v)$ and consider a simple random walk on $T$
starting at $\root$ and ending when it visits the first leaf vertex of
$\Ball_L(\root;T)$. Then let
$h^{(\root,L)}(v;T)$ be the probability that $v$ is the last
vertex at distance $k$ from $\root$ that is visited by this walk.
Clearly, for each $k=0,\ldots,L$,
\begin{equation}\label{eq:hsum}
\sum_{v \in T:k(v)=k} h^{(\root,L)}(v;T)=1.
\end{equation}
We may then define a local algorithm $F_{\ell,L} \in \cF_*^\cM(L)$ by
\begin{equation}\label{eq:Fharmonic}
F_{\ell,L}(T,\root,\bz):=
\begin{cases}
\frac{1}{\sqrt{\ell+1}}\left(\sum_{v \in \Ball_\ell(\root;T)}
\sqrt{h^{(\root,L)}(v;T)}z(v)\right) & \text{ if }
\{v:k(v)=L\} \text{ is non-empty},\\
\sign\left(\sum_{v \in \Ball_\ell(\root;T)} z(v)\right) & \text{ otherwise}.
\end{cases}
\end{equation}
When $z(v) \overset{iid}{\sim} \Normal(0,1)$ conditional on $(T,\bsigma)$, the
conditions of (\ref{eq:Fconditions}) hold by (\ref{eq:hsum}). So it remains to
compute $\cE(F_{\ell,L},\nu)$.

For $(T,\{\root,\root'\}) \in \cT_e$, define $T_\root$, $T_{\root'}$, $k(v)$,
$k'(v')$, $\Ball_k$, and $\Ball_k'$ as in the proof of Lemma
\ref{lemma:ERlowerboundhelper}, and recall that $\Ball_L(\root;T)$ consists
of $\Ball_L$ and $\Ball_{L-1}'$ connected by the edge $\{\root,\root'\}$
and that $\Ball_L(\root';T)$ consists of $\Ball_L'$ and $\Ball_{L-1}$
connected by this same edge. We consider the same three cases as in the proof of
Lemma \ref{lemma:ERlowerboundhelper}: (I) If $k(v) \leq L-2$ for all
$v \in T_\root$ and $k'(v') \leq L-2$ for all $v' \in T_{\root'}$, then the
second case of (\ref{eq:Fharmonic}) holds for both $\Ball_L(\root;T)$ and
$\Ball_L(\root';T)$, and
\begin{equation}\label{eq:EFFharmoniccase1}
F_{\ell,L}(T,\root,\bz)F_{\ell,L}(T,\root'\bz)=1.
\end{equation}
(II) If $\max_{v \in T_\root} k(v)=L-1$ and $k'(v') \leq L-2$ for all $v' \in
T_{\root'}$, or vice versa, then we simply bound by Cauchy-Schwarz
\begin{equation}\label{eq:EFFharmoniccase1}
|\E_{\bz}[F_{\ell,L}(T,\root,\bz)F_{\ell,L}(T,\root'\bz)]| \leq 1.
\end{equation}
(III) Otherwise, the first case of (\ref{eq:Fharmonic}) holds for both
$\Ball_L(\root;T)$ and $\Ball_L(\root';T)$, and we have
\begin{align*}
&\E_\bz\left[F_{\ell,L}(T,\root,\bz)F_{\ell,L}(T,\root',\bz)\right]\\
&\;\;\;\;=\frac{1}{\ell+1}\left(\sum_{v \in \Ball_{\ell-1}}
\sqrt{h^{(\root,L)}(v;T)h^{(\root',L)}(v;T)}+\sum_{v' \in \Ball_{\ell-1}'}
\sqrt{h^{(\root,L)}(v';T)h^{(\root',L)}(v';T)}\right).
\end{align*}

For $v \in \Ball_{\ell-1}$ and $v' \in \Ball_{\ell-1}'$, let us write as
shorthand $h^{(L)}(v):=h^{(\root,L)}(v;T_\root)$ and
${h^{(L)}}'(v'):=h^{(\root',L)}(v';T_{\root'})$ for the depth-$L$ harmonic
measures in the subtrees $T_\root$ and $T_{\root'}$. To relate these quantities
to the harmonic measure in the full tree $T$, consider a simple random walk
on $T$ starting at $\root$ and ending when it hits the first leaf vertex of
$\Ball_L(\root;T)$, i.e.\ when it hits the first vertex $v \in T_\root$ for
which $k(v)=L$ or the first vertex $v' \in T_{\root'}$ for which
$k'(v')=L-1$. Let $\cA$ be the event that the last vertex in
$\{\root,\root'\}$ visited by this walk is $\root$. Then the
Markov property of the walk implies that
for any $v \in \Ball_L$ with $k(v)=k$, $v$ can be the last vertex at distance
$k$ from $\root$ that is visited by this walk only if $\cA$ holds, and the
probability of this occurring conditional on $\cA$ is $h^{(L)}(v)$. Similarly,
for any $v' \in \Ball_{L-1}'$ with $k'(v')=k-1$, $v'$ can be the last
vertex at distance $k$ from $\root$ that is visited by this walk only if
$\cA^c$ holds, and the probability of this occurring conditional on $\cA^c$ is
${h^{(L-1)}}'(v')$. Hence for any $v \in \Ball_{\ell-1}$ and
$v' \in \Ball_{\ell-1}'$, and letting $\P_{\root}$ denote the probability distribution of the simple random walk started at $\root$,
we have
\[
h^{(\root,L)}(v;T)=\P_{\root}[\cA]h^{(L)}(v),\;\;h^{(\root,L)}(v';T)
=\P_{root}[\cA^c]{h^{(L-1)}}'(v').
\]
Considering analogously a walk on $T$ starting at $\root'$ and ending at the
first visited leaf vertex of $\Ball_L(\root';T)$, we have
\[h^{(\root',L)}(v;T)=\P_{\root'}[\cA]h^{(L-1)}(v),\;\;h^{(\root',L)}(v';T)
=\P_{\root'}[\cA^c]{h^{(L)}}'(v').\]

Denote by $\cond^{(L)}$ the conductance between $\root$ and the leaves of
$\Ball_L$ in the subtree $T_\root$, with $\cond^{(L)}=0$ if $T_\root$ has
no vertices $v$ with $k(v)=L$. Similarly, denote by ${\cond^{(L)}}'$ the
conductance between $\root'$ and the leaves of $\Ball_L'$ in the subtree
$T_{\root'}$. Recall \cite{lyons1997unsolved} that
if $T_\root$ is augmented with the vertex $\root'$
connected by an edge to $\root$, then
$\cond^{(L)}/(1+\cond^{(L)})$ gives the probability that a simple
random walk on $T_\root$ started at $\root$ hits a leaf vertex of $\Ball_L$
before hitting $\root'$, and the analogous statement holds for 
${\cond^{(L)}}'$.
Then letting $i$ count the number of visits of the random walk to $\root'$,
\[\P_{\root}[\cA]=\sum_{i=0}^\infty \frac{\cond^{(L)}}{1+\cond^{(L)}}
\left(\frac{1}{1+\cond^{(L)}}\frac{1}{1+{\cond^{(L-1)}}'}\right)^i
=\frac{\cond^{(L)}(1+{\cond^{(L-1)}}')}
{\cond^{(L)}+{\cond^{(L-1)}}'+\cond^{(L)}{\cond^{(L-1)}}'}.\]
The analogous formula holds for $\P_{\root'}[\cA^c]$. Recalling \cite{lyons1997unsolved} 
that as $L \to \infty$, $\cond^{(L)} \to \cond$ and ${\cond^{(L)}}' \to \cond'$
where $\cond$ and $\cond'$ are the conductances of the infinite trees $T$ and
$T'$, and $h^{(L)}(v) \to h(v)$ and ${h^{(L)}}'(v') \to h'(v')$ for any fixed
$v \in T$ and $v' \in T'$ where $h$ and $h'$ are the the harmonic
measures of the infinite trees $T$ and $T'$, as defined in Section
\ref{sec:harmonic}, this implies that in case (III),
\begin{align*}
\lim_{L \to \infty} \E_{\bz}[F_{\ell,L}(T,\root,\bz)F_{\ell,L}(T,\root',\bz)]
&=\frac{1}{\ell+1}\left(\sum_{v \in \Ball_{\ell-1}}
h(v)\frac{\cond\sqrt{1+\cond'}}{\cond+\cond'+\cond\cond'}
+\sum_{v' \in \Ball_{\ell-1}'} h'(v')
\frac{\cond'\sqrt{1+\cond}}{\cond+\cond'+\cond\cond'}\right)\\
&=\frac{\ell}{\ell+1}\frac{\cond\sqrt{1+\cond'}+\cond'\sqrt{1+\cond}}
{\cond+\cond'+\cond\cond'},
\end{align*}
where the second equality follows from $\sum_{v:k(v)=k} h(v)=1$ and
$\sum_{v':k'(v')=k} h'(v')=1$ for each $k=0,\ldots,\ell-1$.

Let $\1\{\mathrm{I}\}$, $\1\{\mathrm{II}\}$, and $\1\{\mathrm{III}\}$ indicate
which of the above three cases occur. The event that $T$ goes extinct equals
the event $\{\cond=0\}$ a.s., and similarly for $T'$ and $\cond'$, so by the
same argument as in the proof of Lemma \ref{lemma:ERlowerboundhelper},
as $L \to \infty$, $\1\{\mathrm{I}\} \to \1\{\cond>0 \text{ or } \cond'>0\}$,
$\1\{\mathrm{II}\} \to 0$, and $\1\{\mathrm{III}\} \to
\1\{\cond=0,\cond'=0\}$. Combining the three cases above, taking the full
expectation with respect to $\nue^*$, letting $L \to
\infty$, and applying the bounded convergence theorem,
\[\lim_{L \to \infty} \cE(F_{\ell,L},\nu)=d\,\E\Psi^{(\ell)}(\cond,\cond')\]
where
\[\Psi^{(\ell)} \equiv \begin{cases} \frac{\ell}{\ell+1}
\frac{\cond\sqrt{1+\cond'}+\cond'
\sqrt{1+\cond}}{\cond+\cond'+\cond\cond'} & \text{ if } \cond>0 \text{ or }
\cond'>0 \\ 1 & \text{ otherwise}. \end{cases}\]
Then taking $\ell \to \infty$ and applying again the bounded convergence
theorem yields the desired result.

\subsection{Proof of Lemma \ref{lemma:SBMlowerbound}}
We use throughout the fixed values $\delta \in (0,1]$,
$d=(a+b)/2$, $\mu=(a-b)/2$, and
$\lambda=\mu/\sqrt{d}$. Recall that by assumption, $d \geq 2$ and $\lambda>1$.

For any rooted tree $(T,\root,\bsigma)$ with vertex labels $\bsigma:V(T) \to
\{+1,-1,u\}$, define by $k(v):=\dist(v,\root)$ the distance from $v$
to $\root$. For $\ell \geq 0$, define $N_\ell$ and $X_\ell$ as in
(\ref{eq:Xl}), and define
\[N_\ell^+:=|\{v:k(v)=\ell,\sigma(v)=+1\}|,\;\;
N_\ell^-:=|\{v:k(v)=\ell,\sigma(v)=-1\}|,\]
\begin{equation}\label{eq:Dl}
D_\ell:=\delta^{-1}\mu^{-\ell}(N^+_\ell-N^-_\ell).
\end{equation}
Note that $(X_\ell,D_\ell)$ is computable from the observed labels in
$\Ball_\ell(\root;T)$.

For $(T,\root,\bsigma)$ a random two-type Galton-Watson tree with offspring
distributions
$\Pois(a/2)$ and $\Pois(b/2)$ and vertex labels partially revealed with
probability $\delta$, we denote by $\sigma_\true(\root) \in \{+1,-1\}$ the
vertex label of the hidden partition that contains $\root$. (So
$\sigma_\true(\root)=\sigma(\root)$ if the label of $\root$ is revealed.)
Define
\begin{equation}\label{eq:XY}
(X,Y):=\lim_{\ell \to \infty} (X_\ell,\sigma_\true(\root)D_\ell),
\end{equation}
where the limit exists by the following lemma.
\begin{lemma}\label{lemma:XYmomentstails}
Let $\delta \in (0,1]$ and let $a,b>0$ be such that $d>1$ and $\lambda>1$.
Let $(X_\ell,D_\ell)$ be defined by (\ref{eq:Xl}) and (\ref{eq:Dl})
for the two-type Galton-Watson
tree with offspring distributions $\Pois(a/2)$ and $\Pois(b/2)$ and vertex
labels partially revealed with probability $\delta$. Then
the limit $(X,Y)$ in (\ref{eq:XY}) exists almost surely, and $X$ and $Y$ satisfy
\begin{equation}\label{eq:XYmoments}
\E[X]=1,\;\;\Var[X]=\frac{1}{d-1},\;\;
\E[Y]=1,\;\;\Var[Y]=\frac{d}{\mu^2-d}.
\end{equation}
Furthermore, if $d \geq 2$, then for some universal constants $C,c>0$
and any $\gamma>0$,
\begin{equation}\label{eq:XYtails}
\P\left[|X-1| \geq \frac{\gamma}{\sqrt{d-1}} \right] \leq
C\exp(-c\gamma),\;\;\;\;
\P\left[|Y-1| \geq \gamma\sqrt{\frac{d}{\mu^2-d}} \right] \leq C\exp(-c\gamma).
\end{equation}
\end{lemma}
The proof of this lemma is deferred to Appendix \ref{appendix:XY}.

Analogous to our proof of Lemma \ref{lemma:ERlowerbound} in the \ER case,
to establish Lemma \ref{lemma:SBMlowerbound},
we first prove the following intermediary result.
\begin{lemma}\label{lemma:SBMlowerboundhelper}
Let $(X,Y)$, $(X',Y')$ be independent pairs of random variables with law 
defined by (\ref{eq:XY}), where $(X_\ell,D_\ell)$ are defined by (\ref{eq:Xl})
and (\ref{eq:Dl}) for
the two-type Galton-Watson tree with offspring distributions $\Pois(a/2)$ and
$\Pois(b/2)$ and vertex labels partially revealed with probability $\delta$.
Then in the setup of Lemma \ref{lemma:SBMlowerbound},
for any fixed $\alpha>0$ and for each $\ell \geq 1$, there
exists a local algorithm $F_{\ell,\alpha} \in \cF_*^{\cM}(\ell)$
satisfying (\ref{eq:Fconditions}) such that
\[\lim_{\ell \to \infty} \cE(F_{\ell,\alpha},\nu)=
\E[\us(X,Y,X',Y';\alpha)],\]
where
\begin{align}
\us(X,Y,X',Y';\alpha)&:=
\frac{a}{2}\frac{\frac{1}{\sqrt{d}}(X+X')+\alpha(Y+\frac{Y'}{\mu})
(Y'+\frac{Y}{\mu})}{\sqrt{X+\frac{X'}{d}+\alpha(Y+\frac{Y'}{\mu})^2}
\sqrt{X'+\frac{X}{d}+\alpha(Y'+\frac{Y}{\mu})^2}}\\
&\hspace{0.5in}
+\frac{b}{2}\frac{\frac{1}{\sqrt{d}}(X+X')-\alpha(Y-\frac{Y'}{\mu})
(Y'-\frac{Y}{\mu})}{\sqrt{X+\frac{X'}{d}+\alpha(Y-\frac{Y'}{\mu})^2}
\sqrt{X'+\frac{X}{d}+\alpha(Y'-\frac{Y}{\mu})^2}}\nonumber
\end{align}
if $X>0$ or $X'>0$, and $\us(X,Y,X',Y';\alpha):=d$ if $X=X'=0$.
\end{lemma}
\begin{proof}
The proof is similar to that of Lemma \ref{lemma:ERlowerboundhelper}, and we
explain only the key differences. Define the local algorithm
\[F_{\ell,\alpha}(T,\root,\bsigma,\bz):=\begin{cases}
\frac{\sum_{v \in \Ball_\ell(\root;T)} d^{-k(v)/2}z(v)+D_\ell
\sqrt{\alpha \ell}}
{\sqrt{\sum_{v \in \Ball_\ell(\root;T)} d^{-k(v)}+D_\ell^2\alpha \ell}}
& X_\ell>0, \\
\sign\left(\sum_{v \in \Ball_\ell(\root;T)} z(v)\right) & X_\ell=0,
\end{cases}\]
where $k(v)$ and $(X_\ell,D_\ell)$ are defined as above for the labeled
tree $(T,\root,\bsigma)$.
When $z(v) \overset{iid}{\sim} \Normal(0,1)$ conditional on
$(T,\root,\bsigma)$, the conditions of (\ref{eq:Fconditions}) hold by
construction, where the first
condition follows from noting that $\E[z(v) \mid T,\root,\bsigma]=0$ and that
$\E[D_\ell]=0$. It remains to compute $\cE(F_{\ell,\alpha},\nu)$.

For $(T,\{\root,\root'\},\bsigma) \in \cT_e(\cM)$,
define $T_\root$ and $T_{\root'}$ as in the proof of Lemma
\ref{lemma:ERlowerboundhelper}. Recall from Example \ref{example:SBM} that
$T_\root$ and $T_{\root'}$ (with the marks $\bsigma|_T$ and $\bsigma|_{T'}$)
each have the law of a two-type Galton-Watson tree
with offspring distributions $\Pois(a/2)$ and $\Pois(b/2)$ and labels partially
revealed with probability $\delta$, and they are conditionally independent given
$\sigma_\true(\root)$ and $\sigma_\true(\root')$.
Define $(X_\ell,D_\ell)$ by (\ref{eq:Xl}) and (\ref{eq:Dl}) for the subtree
$T_\root$, and $(X_\ell',D_\ell')$ by (\ref{eq:Xl}) and (\ref{eq:Dl}) for the
subtree $T_{\root'}$. Define also $S_\ell=\sum_{j=0}^\ell X_j$,
$S_\ell'=\sum_{j=0}^\ell X_\ell'$,
$Y_\ell=\sigma_\true(\root)D_\ell$, and
$Y_\ell'=\sigma_\true(\root')D_\ell'$. Considering the same three cases (I),
(II), and (III) as in the proof of Lemma \ref{lemma:ERlowerboundhelper}, the
same argument shows for case (I)
\[F_{\ell,\alpha}(\Ball_\ell(\root;T),\root,\bsigma,\bz)
F_{\ell,\alpha}(\Ball_\ell(\root';T),\root',\bsigma,\bz)=1,\]
for case (II)
\[|\E_{\bz}[F_{\ell,\alpha}(\Ball_\ell(\root;T),\root,\bsigma,\bz)
F_{\ell,\alpha}(\Ball_\ell(\root';T),\root',\bsigma,\bz)]| \leq 1,\]
and for case (III)
\begin{align*}
&\E_\bz[F_{\ell,\alpha}(\Ball_\ell(\root;T),\root,\bsigma,\bz)
F_{\ell,\alpha}(\Ball_\ell(\root';T),\root',\bsigma,\bz)]\\
&=\begin{cases}
\frac{\frac{1}{\sqrt{d}}(S_{\ell-1}+S_{\ell-1}')+\alpha \ell
(Y_\ell+\frac{1}{\mu}Y_{\ell-1}')
(Y_\ell'+\frac{1}{\mu}Y_{\ell-1})}
{\sqrt{S_\ell+\frac{1}{d}S_{\ell-1}'+\alpha
\ell(Y_\ell+\frac{1}{\mu}Y_{\ell-1}')^2}
\sqrt{S_\ell'+\frac{1}{d}S_{\ell-1}+\alpha
\ell(Y_\ell'+\frac{1}{\mu}Y_{\ell-1})^2}}
& \sigma_\true(\root)=\sigma_\true(\root') \\
\frac{\frac{1}{\sqrt{d}}(S_{\ell-1}+S_{\ell-1}')-\alpha \ell
(Y_\ell-\frac{1}{\mu}Y_{\ell-1}')
(Y_\ell'-\frac{1}{\mu}Y_{\ell-1})}
{\sqrt{S_\ell+\frac{1}{d}S_{\ell-1}'+\alpha
\ell(Y_\ell-\frac{1}{\mu}Y_{\ell-1}')^2}
\sqrt{S_\ell'+\frac{1}{d}S_{\ell-1}+\alpha
\ell(Y_\ell'-\frac{1}{\mu}Y_{\ell-1})^2}}
& \sigma_\true(\root)=-\sigma_\true(\root').
\end{cases}
\end{align*}
Note that by the symmetry of $+1$ and $-1$ labels in the definition of the
two-type Galton-Watson tree,
$\{(X_\ell,Y_\ell)\}_{\ell=1}^\infty$
is independent of $\sigma_\true(\root)$, and
similarly $\{(X_\ell',Y_\ell')\}_{\ell=1}^\infty$
is independent of $\sigma_\true(\root')$.
Hence by the characterization of $\nue$ in
Example \ref{example:SBM}, $\{(X_\ell,Y_\ell)\}_{\ell=1}^\infty$
is independent of $\{(X_\ell',Y_\ell')\}_{\ell=1}^\infty$ under 
$\nue$. Then convergence of C\'esaro sums implies, almost surely,
\begin{align*}
&\lim_{\ell \to \infty} \tfrac{1}{\ell}S_\ell
=\lim_{\ell \to \infty} \tfrac{1}{\ell}S_{\ell-1}=\lim_{\ell \to \infty} X_{\ell}=X,\\
&\lim_{\ell \to \infty} \tfrac{1}{\ell}S_{\ell}'
=\lim_{\ell \to \infty} \tfrac{1}{\ell}S_{\ell-1}'=\lim_{\ell \to \infty}
X_\ell'=X',\\
&\lim_{\ell \to \infty} Y_\ell=\lim_{\ell \to \infty} Y_{\ell-1}=Y,\\
&\lim_{\ell \to \infty} Y_\ell'=\lim_{\ell \to \infty} Y_{\ell-1}'=Y',
\end{align*}
where $(X,Y)$ and $(X',Y')$ are independent pairs of random variables with law
defined by (\ref{eq:XY}). As in the proof of Lemma
\ref{lemma:ERlowerboundhelper}, as $\ell \to \infty$, $\1\{\mathrm{I}\} \to
\1\{X=0,X'=0\}$, $\1\{\mathrm{II}\} \to 0$, and $\1\{\mathrm{III}\} \to
\1\{X>0 \text{ or } X'>0\}$.
Then combining these three cases, taking the full expectation with respect to
$\nue^*$, recalling from Example \ref{example:SBM} that under $\nue^*$ we have
$\sigma_\true(\root)=\sigma_\true(\root')$ with probability
$a/(a+b)=a/(2d)$ and $\sigma_\true(\root)=-\sigma_\true(\root')$ with
probability $b/(a+b)=b/(2d)$,
letting $\ell \to \infty$, and applying the bounded convergence theorem,
we obtain the desired result.
\end{proof}

\begin{proof}[Proof of Lemma \ref{lemma:SBMlowerbound}]
We compute a lower bound for the quantity $\E[\us(X,Y,X',Y';\alpha)]$ 
in Lemma \ref{lemma:SBMlowerboundhelper}, for the choice
\[\alpha=\frac{\mu^2-d}{\mu^2}\frac{1}{\sqrt{d}}.\]

For any positive function $f(\mu,d)$ and any random
variable $Z:=Z(\mu,d)$ whose law depends on $\mu$ and $d$, we write
$Z=\OP(f(\mu,d))$ if $\E[|Z|^k] \leq C_kf(\mu,d)^k$ for some constants
$\{C_k\}_{k \geq 0}$ independent of $\mu$ and $d$ and for all
$k \geq 1$, $d \geq 2$, and $\mu>\sqrt{d}$. By the Minkowski, Cauchy-Schwarz,
and Jensen inequalities, if $Z=\OP(f(\mu,d))$ and $Z'=\OP(g(\mu,d))$, then
$Z+Z'=\OP(f(\mu,d)+g(\mu,d))$, $ZZ'=\OP(f(\mu,d)g(\mu,d))$, and
$\sqrt{|Z|}=\OP(\sqrt{f(\mu,d)})$. Note that if
$\P[|Z| \geq \gamma f(\mu,d)] \leq C\exp(-c\gamma)$
for some constants $C,c>0$ and all $\gamma>0$, then for each $k \geq 1$
\[\E\left[\frac{|Z|^k}{f(\mu,d)^k}\right]
=\int_0^\infty k\gamma^{k-1}
\P\left[\frac{|Z|}{f(\mu,d)}>\gamma\right]d\gamma \leq C_k\]
for a constant $C_k>0$, and hence $Z=\OP(f(\mu,d))$. Then 
Lemma \ref{lemma:XYmomentstails} implies $X-1=\OP(1/\sqrt{d})$ and
$Y-1=\OP(\sqrt{d/(\mu^2-d)})$. This then implies $X=\OP(1)$
and $Y^2 \leq 2(Y-1)^2+2=\OP(\mu^2/(\mu^2-d))$, so
$\alpha Y^2,\alpha {Y'}^2,\alpha YY'=\OP(1/\sqrt{d})$.

Let us write $W_{\pm}=X+\frac{X'}{d}+\alpha(Y\pm \frac{Y'}{\mu})^2$
and $W_{\pm}'=X'+\frac{X}{d}+\alpha(Y'\pm \frac{Y}{\mu})^2$.
Define the event $\cE=\{X>1/2 \text{ and } X'>1/2\}$. On $\cE$, we have
$W_{\pm},W_{\pm}'>1/2$. Then applying the bound
$|(1+x)^{-1/2}-1+x/2| \leq x^2$ for all $x>-1/2$ and noting
$1/\mu<1/\sqrt{d}$,
\begin{align*}
\frac{1}{\sqrt{W_{\pm}W_{\pm}'}}\1\{\cE\}
&=\left(1-\frac{X-1}{2}-\frac{\alpha Y^2}{2}+\OP(1/d)\right)
\left(1-\frac{X'-1}{2}-\frac{\alpha {Y'}^2}{2}+\OP(1/d)\right)\1\{\cE\}\\
&=\left(1-\frac{X-1}{2}-\frac{X'-1}{2}-\frac{\alpha Y^2}{2}
-\frac{\alpha {Y'}^2}{2}+\OP(1/d)\right)\1\{\cE\}.
\end{align*}
Then, recalling $d=(a+b)/2$ and $\mu=(a-b)/2$ and noting $a,b \leq d$,
\begin{align*}
\us(X,Y,X',Y';\alpha)\1\{\cE\}
&=\left(d\left(\frac{1}{\sqrt{d}}(X+X')
+\frac{\alpha}{\mu}(Y^2+{Y'}^2)\right)+\alpha\mu\left(1+\frac{1}{\mu^2}\right)
YY'\right)\\
&\hspace{0.5in}\left(1-\frac{X-1}{2}-\frac{X'-1}{2}-\frac{\alpha Y^2}{2}
-\frac{\alpha {Y'}^2}{2}+\OP(1/d)\right)\1\{\cE\}\\
&=\left[2\sqrt{d}+\sqrt{d}(X-1+X'-1)+\tfrac{\alpha d}{\mu}(Y^2+{Y'}^2)
+\alpha\mu YY'\right.\\
&\left.\hspace{1in}-\sqrt{d}(X-1+X'-1)-\alpha\sqrt{d}(Y^2+{Y'}^2)
+\OP(1/\sqrt{d})\right]\1\{\cE\}\\
&=\left[2\sqrt{d}+\alpha\left(\frac{d}{\mu}-\sqrt{d}\right)
(Y^2+{Y'}^2)+\alpha \mu YY'+\OP(1/\sqrt{d})\right]\1\{\cE\}.
\end{align*}

Writing
\begin{equation}\label{eq:Rdef}
R:=\us(X,Y,X',Y';\alpha)-\left(2\sqrt{d}+\alpha\left(\frac{d}{\mu}-\sqrt{d}
\right)(Y^2+{Y'}^2)+\alpha \mu YY'\right),
\end{equation}
the above implies $\E[|R|\1\{\cE\}] \leq C/\sqrt{d}$ for an absolute constant
$C>0$ and all $d \geq 2$ and $\mu>\sqrt{d}$. On the other hand,
$\E[|R|\1\{\cE^c\}] \leq \E[R^2]\P[\cE^c] \leq \E[R^2](Ce^{-cd})$ for some
constants $C,c>0$ by (\ref{eq:XYtails}). Noting that
$\us(X,Y,X',Y';\alpha)$ satisfies the deterministic bound
\[\us(X,Y,X',Y';\alpha)
\leq \sqrt{d}\frac{(X+X')}{\sqrt{X+\frac{X'}{d}}\sqrt{X'+\frac{X}{d}}}+
\frac{a}{2}+\frac{b}{2}\leq 2d\]
(which holds also if $X=X'=0$, by definition of $\us$),
we have $R=\OP(d)$, so $\E[R^2] \leq Cd^2$. As $d^2e^{-cd} \leq C'/\sqrt{d}$ for
all $d \geq 2$ and some constant $C'>0$, this yields
$\E[|R|] \leq C/\sqrt{d}$. Finally, applying (\ref{eq:Rdef}) and
(\ref{eq:XYmoments}), for any $d \geq 2$ and $\mu>\sqrt{d}$,
\[\E[\us(X,Y,X',Y';\alpha)] \geq 2\sqrt{d}
+2\alpha\left(\frac{d}{\mu}-\sqrt{d}\right)\frac{\mu^2}{\mu^2-d}
+\alpha \mu-\frac{C}{\sqrt{d}}
=2\sqrt{d}+\frac{(\mu-\sqrt{d})^2}{\mu\sqrt{d}}-\frac{C}{\sqrt{d}}.\]
Combining with Lemma \ref{lemma:SBMlowerboundhelper} yields the desired result.
\end{proof}

\section*{Acknowledgements}
Z.F.\ was partially supported by a Hertz Foundation Fellowship and an
NDSEG Fellowship (DoD AFOSR 32 CFR 168a).
A.M.\ was partially supported by the NSF grants CCF-1319979.

\appendix
\section{Combinatorial lemmas}\label{app:combinatorics}
In this appendix, we prove Lemmas \ref{lemma:alternatingprobbound} and 
\ref{lemma:equivclassbound}. For any graph $H$, $l \geq 1$, and set of
vertices $S$ in $H$, let $\Ball_l(S;H)$ denote the subgraph consisting of
all vertices at distance at most $l$ from $S$ in $H$
(including $S$ itself) and all edges between pairs of such vertices.
Let $|\Ball_l(S;H)|$ denote the number of such
vertices. For a single vertex $v$, we write $\Ball_l(v;H):=\Ball_l(\{v\};H)$.

\begin{lemma}\label{lemma:ballbound}
Fix $d>1$ and consider the \ER graph $G \sim \sG(n,d/n)$. Then
there exist $C,c>0$ such that for any $s \geq 0$, $l \geq 1$, and
$v \in \{1,\ldots,n\}$,
\[\P[|\Ball_l(v;G)|>sd^l] \leq Ce^{-cs}.\]
\end{lemma}
\begin{proof}
See \cite[Lemma 29]{bordenaveetal}.
\end{proof}

Recall Definition \ref{def:cyclenumber} of the cycle number $\numcycle(H)$,
and also Definition \ref{def:lcoil} of an $l$-coil. Let us say an
$l$-coil is {\bf irreducible} if no proper subset of its edges forms an
$l$-coil.

\begin{lemma}\label{lemma:cyclecoilnumber}
For any $l \geq 1$ and any graph $H$, the number of distinct edges in $H$ that
belong to any irreducible $l$-coil of $H$ is at most $l \cdot \numcycle(H)$.
\end{lemma}
\begin{proof}
Let $H'$ denote the subgraph of $H$ formed by the union of all irreducible
$l$-coils in $H$. It suffices to show $l \cdot \numcycle(H') \geq e(H')$ where
$e(H')$ is the number of edges in $H'$. As the cycle number and number of edges
are additive over connected components, it suffices to show this separately for
each connected component of $H'$; hence assume without loss of generality
$H'$ is connected (and non-empty, otherwise the result is trivial).

Let us construct $H'$ by starting with a single irreducible $l$-coil $H_1'$
and, for each $t \geq 2$, letting $H_t'$ be the union of
$H_{t-1}'$ and an irreducible $l$-coil sharing at least one vertex with
$H_{t-1}'$ but not entirely contained in $H_{t-1}'$. (Such an $l$-coil exists
until $H_t'=H'$.) Denote by $e(H_t')$ and $v(H_t')$ the number
of distinct edges and vertices in $H_t'$. Clearly $e(H_t')-e(H_{t-1}') \leq l$
for each $t \geq 1$, hence $e(H_t') \leq tl$. Note $\numcycle(H_1') \geq 1$.
For each $t \geq 1$, $e(H_t')-e(H_{t-1}') \geq v(H_t')-v(H_{t-1}')$, since
adding each new vertex requires adding at least one new edge as $H_t'$ remains
connected. Furthermore, equality can only hold if the newly added vertices and
edges form a forest, where each tree in the forest intersects $H_{t-1}'$ only
at its root node. But this would imply that there exists a vertex in $H_t'$ 
with degree one, contradicting that this vertex is part of any irreducible
$l$-coil. Hence in fact $e(H_t')-e(H_{t-1}') \geq v(H_t')-v(H_{t-1}')+1$,
implying $\numcycle(H_t') \geq \numcycle(H_{t-1}')+1$ for all $t$. Then
$\numcycle(H_t') \geq t$, and so $l \cdot \numcycle(H') \geq lt \geq e(H')$ for
the value of $t$ such that $H_t'=H'$.
\end{proof}

\begin{lemma}\label{lemma:cyclebound}
Fix $d>1$.
Let $S$ be a subset of edges in the complete graph on $n$ vertices such that
$|S| \leq 2(\log n)^2$, and let $\#_c(S)$ denote the cycle number of the
subgraph formed by the edges in $S$. Let $l$ be a positive integer with $l \leq
0.1\log_d n$. Let $G^o \cup S$ denote the random subgraph of the complete
graph in which each edge outside of $S$ is present
independently with probability $d/n$ and each edge in $S$ is present with
probability 1. Let $V \subseteq \{1,\ldots,n\}$
denote the set of vertices incident to at least one edge in $S$.
Then for some $C:=C(d)>0$, $N_0:=N_0(d)>0$, all $n \geq N_0$, and
all $0<t \leq (\log n)^2$,
\[\P[\#_c(\Ball_l(V;G^o \cup S)) \geq \#_c(S)+t] \leq C(\log n)^2n^{-0.7t}.\]
\end{lemma}
\begin{proof}
Let $G^o$ denote the graph $G^o \cup S$ with all edges in $S$ removed.
Construct a growing breadth-first-search forest in $G^o$, ``rooted'' at $V$,
in the following manner: Initialize $F_0$ as the
graph with the vertices $V$ and no edges, and mark each vertex in $V$ as
unexplored. Iteratively for each $t \geq 1$, consider the set of unexplored
vertices in $F_{t-1}$ having minimal distance from $V$, and let $v_t$ be the
one with smallest index. Mark $v_t$ as explored, let $N_{v_t}$ be the set of
neighbors of $v_t$ in $G^o$ which are not in $F_{t-1}$, and let $F_t$ be 
$F_{t-1}$ with all vertices $v \in N_{v_t}$ and edges 
$\{v_t,v\}:v \in N_{v_t}$ added. (Hence each $F_t$ is a forest of $|V|$
disjoint trees, with one tree rooted at each vertex $v \in V$ and with all of
its edges in $G^o$.) Let $\tau$ be the first time for which all vertices
in $\Ball_l(V;G^o \cup S)$ are in $F_\tau$. Note that for any $l \geq 1$,
$\Ball_l(V;G^o \cup S)=S \cup \Ball_l(V;G^o)$ (i.e.\ $\Ball_l(V;G^o)$
with the edges in $S$ added), as $S$ is contained in $\Ball_1(V;G^o \cup
S)$ and also any vertex at distance at most $l$ from $V$ in $G^o \cup S$ is at
distance at most $l$ from $V$ in $G^o$, by definition of $V$. Then the cycle
number $\numcycle(\Ball_l(V;G^o \cup S))$ is at most $\numcycle(S)$ plus the
number of edges in $\Ball_l(V;G^o)$ that are not in $F_\tau$ (as removing these
edges and $\numcycle(S)$ edges from $S$ yields a graph with no cycles).

Each edge in $\Ball_l(V;G^o)$ that is not in $F_\tau$ must either be
between two vertices $v_1$ and $v_2$ at the same distance $r \in [0,l]$ from
$V$, or between a vertex $v_1$ at some distance $r \in [0,l-1]$ from $V$ and a
vertex $v_2$ at distance $r+1$ from $V$, where $v_2$ is a child (in $F_\tau$)
of a different vertex $v'$ at distance $r$ from $V$ and having smaller index
than $v_1$. Given $F_\tau$, let $\mathcal{S}_{F_\tau}$ denote the set of all
such pairs of vertices $\{v_1,v_2\}$. Then the event that $F_\tau$ is the
breadth-first-search forest as constructed above is exactly the event that
the vertices of $\Ball_l(S;G^o)$ are those of $F_\tau$ and the edges of
$\Ball_l(S;G^o)$ are those of $F_\tau$ together with some subset of the edges
corresponding to the vertex pairs in $\mathcal{S}_{F_\tau}$. Hence, conditional
on $F_\tau$, each edge $\{v_1,v_2\} \in \mathcal{S}_{F_\tau}$ is present in
$G^o$ independently with probability $d/n$. Then the number of such edges
has conditional law $\Binom(|\mathcal{S}_{F_\tau}|,d/n)$,
which is stochastically dominated by $\Binom(|F_\tau|^2,d/n)$ where $|F_\tau|$
is the number of vertices in $|F_\tau|$. Then for any $t>0$,
letting $c$ be the constant in Lemma \ref{lemma:ballbound},
\begin{align*}
&\P[\numcycle(\Ball_l(V;G^o \cup S)) \geq \numcycle(S)+t]\\
&\leq \P\left[|F_\tau|>c^{-1}|V|t(\log n)d^l\right]+\P\left[\Binom\left(\lfloor
c^{-1}|V|t(\log n)d^l\rfloor^2,d/n\right) \geq t\right].
\end{align*}

To bound the first term, note $|F_\tau| \leq \sum_{v \in V} |\Ball_l(v;G^o)|
\leq \sum_{v \in V} |\Ball_l(v;G)|$, where $G \sim \sG(n,d/n)$ denotes the full
\ER graph on $n$ vertices. Then by Lemma \ref{lemma:ballbound},
\[\P\left[|F_\tau|>c^{-1}|V|t(\log n)d^l\right]
\leq \sum_{v \in V} \P\left[|\Ball_l(v;G)|>c^{-1}t(\log n)d^l\right]
\leq C|V|n^{-t} \leq 4C(\log n)^2 n^{-t},\]
where the last bound uses $|V| \leq 2|S| \leq 4(\log n)^2$. For the second
term, let $N:=\lfloor c^{-1}|V|t(\log n)d^l\rfloor^2$ and assume $N>0$,
otherwise the probability is 0. Then for any $\lambda>0$, by the Chernoff bound,
\[\P\left[\Binom(N,d/n) \geq t\right]
\leq e^{-\lambda t}\left(1-\frac{d}{n}+\frac{d}{n}e^\lambda\right)^N\\
\leq \exp\left(-\lambda t+\left(-\frac{d}{n}+\frac{d}{n}e^\lambda
\right)N\right).\]
Using $t \leq (\log n)^2$, $l \leq 0.1\log_d n$ and 
$|V| \leq 2|S| \leq 4(\log n)^2$, and
setting $\lambda=-\log(dN/nt)$ (which is positive for large $n$),
the above is at most
\[\exp\left(t\log\left(\frac{d\lfloor c^{-1}|V|t(\log n)d^l
\rfloor^2}{nt}\right)+t-\frac{d\lfloor c^{-1}|V|t(\log
n)d^l\rfloor^2}{n}\right)
\leq \exp\left(t \log (n^{-0.7})\right)\]
for all large $n$. Combining these bounds yields the desired result.
\end{proof}
\begin{proof}[Proof of Lemma \ref{lemma:alternatingprobbound}]
For any edge set $S$ and $c \in \{0,1\}$, write $\bA_S=c$ as shorthand for
the condition $\forall \{v,w\} \in S:\bA_{vw}=c$. Then
\begin{align*}
&\left|\E\left[\prod_{\{v,w\} \in S} \left(\bA_{vw}-\frac{d}{n}\right)
\prod_{\{v,w\} \in T} \bA_{vw}\1\{G \text{ is } l\text{-coil-free}\}
\right]\right|\\
&=\left|\sum_{J \subseteq S} \left(1-\frac{d}{n}\right)^{|J|}
\left(-\frac{d}{n}\right)^{|S|-|J|}
\P\left[\bA_{S \setminus J}=0,\;\bA_{J \cup T}=1,\;G \text{ is }
l\text{-coil-free}\right]\right|\\
&=\left|\sum_{J \subseteq S} \left(1-\frac{d}{n}\right)^{|S|}
\left(-\frac{d}{n}\right)^{|S|-|J|}\left(\frac{d}{n}\right)^{|J|+|T|}
\P\left[G \text{ is } l\text{-coil-free} \mid
\bA_{S \setminus J}=0,\;\bA_{J \cup T}=1\right]\right|\\
&\leq \left(\frac{d}{n}\right)^{|S|+|T|}\left|\sum_{J \subseteq S} (-1)^{|J|}
\P\left[G \text{ is } l\text{-coil-free} \mid \bA_{S \setminus J}=0,
\;\bA_{J \cup T}=1\right]\right|.
\end{align*}

Let $G^o$ denote the random graph on $n$ vertices in which each edge outside of
$S \cup T$ is present independently with probability $d/n$, and having no
edges in $S \cup T$. Then for any $J \subseteq S$, the distribution of the 
graph $G^o \cup J \cup T$ (i.e.\ $G^o$ with the edges in $J \cup T$ added) is
equal to the conditional distribution of
$G \mid \bA_{S \setminus J}=0,\;\bA_{J \cup T}=1$. Thus
\[\sum_{J \subseteq S} (-1)^{|J|} \P\left[G \text{ is } l\text{-coil-free}
\mid \bA_{S \setminus J}=0,\;\bA_{J \cup T}=1\right]=\E[f(G^o)]\]
for the function
\begin{align*}
f(g^o)&=\sum_{J \subseteq S} (-1)^{|J|} \1\{g^o \cup J \cup T \text{ is }
l\text{-coil-free}\},
\end{align*}
where $g^o$ denotes any fixed realization of $G^o$.
If $|S|/l \leq \#_c(S \cup T)$, then the desired result follows from the
trivial bound $|\E[f(G^o)]| \leq \max_{g^o} |f(g^o)| \leq 2^{|S|}$.

For $|S|/l>\#_c(S \cup T)$, note that if $g^o$ is such
that there exists $e \in S$ for which $g^o \cup S \cup T$ does not contain any
irreducible $l$-coil containing $e$, then for each $J \subseteq S \setminus
\{e\}$, $\1\{g^o \cup J \cup T \text{ is }
l\text{-coil-free}\}=\1\{g^o \cup J \cup \{e\} \cup T
\text{ is } l\text{-coil-free}\}$, so
$f(g^o)=0$. Hence if $f(g^o) \neq 0$, then each edge in $S$ must be contained in
an irreducible $l$-coil of $g^o \cup S \cup T$, which must be an irreducible
$l$-coil in $\Ball_l(V;g^o \cup S \cup T)$ where $V$ is the set of vertices incident
to at least one edge in $S \cup T$. This implies that
$\Ball_l(V;g^o \cup S \cup T)$ has at least $|S|$ distinct edges contained in
irreducible $l$-coils, so by Lemma \ref{lemma:cyclecoilnumber},
$\#_c(\Ball_l(V;g^o \cup S \cup T)) \geq |S|/l$. Applying the bound
$|f(g^o)| \leq 2^{|S|}$, this yields
\[|\E[f(G^o)]| \leq 2^{|S|}\P[f(G^o) \neq 0] \leq 2^{|S|}
\P\Big[\#_c(\Ball_l(V;g^o \cup S \cup T)) \geq |S|/l\Big].\]
Note $|S|+|T| \leq 2(\log n)^2$ and $|S|/l-\#_c(S \cup T) \leq
|S| \leq (\log n)^2$, so applying Lemma
\ref{lemma:cyclebound}, for some $C:=C(d)>0$ and all $n \geq N_0:=N_0(d)>0$,
\[\P\Big[\#_c(\Ball_l(V;g^o \cup S \cup T)) \geq |S|/l\Big]
\leq C(\log n)^2n^{-0.7\left(\frac{|S|}{l}-\#_c(S \cup T)\right)}.\]
Combining the above yields the desired result.
\end{proof}

\begin{proof}[Proof of Lemma \ref{lemma:equivclassbound}]
The proof idea is similar to that of \cite[Lemma 17]{bordenaveetal}.
We order the steps taken by $\gamma^{(1)}$ and $\gamma^{(2)}$ as
$(\gamma^{(1)}_0,\gamma^{(1)}_1),\ldots,(\gamma^{(1)}_{m-1},\gamma^{(1)}_m),
(\gamma^{(2)}_0,\gamma^{(2)}_1),\ldots,(\gamma^{(2)}_{m-1},\gamma^{(2)}_m)$.
(We will use ``step'' to refer to an ordered pair of consecutive vertices
in one of the paths $\gamma^{(1)},\gamma^{(2)}$ and ``edge'' to refer to an
undirected edge $\{v,w\}$ in the complete graph.)
Corresponding to each equivalence class in $\cW(m,l,v,e,K_1,K_2)$ is
a unique canonical element in which $\gamma^{(1)}_0=\gamma^{(2)}_0=1$, and the
successive new vertices visited in the above ordering are $2,3,\ldots,v$.
We bound $|\cW(m,l,v,e,K_1,K_2)|$ by constructing an injective encoding of these
canonical elements and bounding the number of possible codes.

Call a step $(\gamma^{(i)}_j,\gamma^{(i)}_{j+1})$ an ``innovation'' if
$\gamma^{(i)}_{j+1}$ is a vertex not previously visited (in the above ordering).
Edges corresponding to innovations form a tree $T$; call an edge $\{v,w\}$ in
the complete graph a ``tree edge'' if it belongs to $T$. Note that as
$\gamma^{(i)}$ is non-backtracking, if $(\gamma^{(i)}_j,\gamma^{(i)}_{j+1})$ is 
an innovation and $(\gamma^{(i)}_t,\gamma^{(i)}_{t+1})$ is the first step with
$t \geq j$ that is not an innovation, then
$\{\gamma^{(i)}_t,\gamma^{(i)}_{t+1}\}$ cannot be a tree-edge.

The set $K_1$ uniquely partitions into maximal intervals of consecutive
indices. For instance, if
$K_1=\{1,2,3,5,7,9,10\}$ then these intervals are $\{1,2,3\}$, $\{5\}$, $\{7\}$,
$\{9,10\}$. If any such interval is of size $L \geq l$,
let us remove every $l$th element of the interval and replace the interval
by the resulting sub-intervals (of which there are at most
$\lfloor L/l \rfloor+1$, each of size at most $l-1$). Call
the final collection of intervals $\cI_1$. In the same manner, we may obtain
a collection of intervals $\cI_2$ for $K_2$. Then 
\[\left|\left\{j \in \{0,\ldots,m-1\}: j \notin \bigcup_{I \in \cI_1}
I\right\}\right| \leq m-|K_1|+\frac{m}{l}.\]
As each pair of consecutive intervals in $\cI_1$ is separated by at least one
index $j \in \{0,\ldots,m-1\}$, this also implies that
the total number of intervals in $\cI_1$ is bounded as
\[|\cI_1| \leq m-|K_1|+\frac{m}{l}+1.\]
Analogous bounds hold for $\cI_2$. For each $i \in \{1,2\}$,
the collection of intervals $\cI_i$ corresponds to a collection of sub-paths of
$\gamma^{(i)}$, where the interval $I=\{j,j+1,\ldots,j'\} \in \cI_i$
corresponds to $\gamma^{(i)}_I:=(\gamma^{(i)}_j,\gamma^{(i)}_{j+1},
\ldots,\gamma^{(i)}_{j'},\gamma^{(i)}_{j'+1})$. Each sub-path
$\gamma^{(i)}_I:I \in \cI_i$ is a non-backtracking path of length at most $l$.
As $\gamma^{(i)}$ is $l$-coil-free on $K_i$, this implies that
the graph $G(\gamma^{(i)}_I)$ of distinct edges visited by each such sub-path
$\gamma^{(i)}_I$ contains at most one cycle.

For each sub-path $\gamma^{(i)}_I$ corresponding to $I=\{j,j+1,\ldots,j'\}$ and
any innovation $(\gamma^{(i)}_k,\gamma^{(i)}_{k+1})$ in the sub-path, call it a
``leading innovation'' if $k=j$ or if its preceding step
$(\gamma^{(i)}_{k-1},\gamma^{(i)}_k)$ is not an innovation.
Also, call each step $(\gamma^{(i)}_k,\gamma^{(i)}_{k+1})$
that does not coincide with a tree edge a ``non-tree-edge step'' (where ``tree
edge'' is as previously defined by the tree $T$ traversed by all innovations
in the two paths $\gamma^{(1)},\gamma^{(2)}$). Note that the non-tree-edge steps
are disjoint from the innovations, as innovations (by definition) coincide
with edges of the tree $T$. Call each non-tree-edge step either a
``short-cycling step'', a ``long-cycling step'', or a ``superfluous'' step, as
follows: If $G(\gamma^{(i)}_I)$ does not contain any cycles, then all
non-tree-edge
steps are long-cycling steps. If $G(\gamma^{(i)}_I)$ contains a cycle $C$, then 
for each non-tree-edge in $C$, the first step
$(\gamma^{(i)}_k,\gamma^{(i)}_{k+1})$ that traverses that edge
is a short-cycling step. All non-tree-edge steps preceding the first
short-cycling step are long-cycling steps. Letting $\tau$ be such that
$(\gamma^{(i)}_\tau,\gamma^{(i)}_{\tau+1})$ is the first step after the
last short-cycling step that does not belong to $C$ (if such a step exists),
all non-tree-edge steps $(\gamma^{(i)}_k,\gamma^{(i)}_{k+1})$ with $k \geq
\tau$ are also long-cycling steps. The remaining non-tree-edge steps (which
must belong to $C$) are superfluous steps.

Our injective encoding of canonical elements consists of:
\begin{enumerate}[(1)]
\item For each $i \in \{1,2\}$ and each $j \notin \bigcup_{I \in \cI_i} I$:
The vertex indices $\gamma^{(i)}_j$ and $\gamma^{(i)}_{j+1}$.
\item For each sub-path $\gamma^{(i)}_I$: (a) The sequence
of vertex index pairs $(\gamma^{(i)}_{k_1},\gamma^{(i)}_{k_1+1}),\ldots,
(\gamma^{(i)}_{k_P},\gamma^{(i)}_{k_P+1})$ corresponding to leading innovations,
long-cycling steps, and short-cycling steps, (b) for each of these $P$ steps,
whether it is a leading innovation, short-cycling, or long-cycling, and
(c) the total number of superfluous non-tree-edge steps.
\end{enumerate}

To see that this encoding is injective, note that item (1) above specifies the
start and end vertex of each sub-path $\gamma^{(i)}_I$. Between the start
of each sub-path and the first leading innovation, non-tree-edge step, or
end of the sub-path (whichever comes first),
$\gamma^{(i)}_I$ is a non-backtracking walk on the sub-tree
of $T$ corresponding to already-visited vertices and hence is uniquely
determined by the start and end vertices of this walk. Similarly, between each
non-tree-edge step and the next leading innovation, non-tree-edge step, or end
of the sub-path (whichever comes first), $\gamma^{(i)}_I$ is also a
non-backtracking walk uniquely determined by its start and end vertices.
Following a leading innovation, all further steps must be (non-leading)
innovations until the next non-tree-edge step, and hence are uniquely
determined for the canonical element of the equivalence class. Hence, given the
sub-tree of $T$ already visited before $\gamma^{(i)}_I$, as well as 
the vertex index pair for each leading innovation and non-tree-edge step in
$\gamma^{(i)}_I$ and the first and last vertices of $\gamma^{(i)}_I$,
we may uniquely reconstruct $\gamma^{(i)}_I$.

The above encoding does not explicitly specify the vertex index pairs of
superfluous non-tree-edge steps, but note that if $G(\gamma^{(i)}_I)$ contains
a cycle $C$, then $\gamma^{(i)}_I$ cannot leave and return to $C$ because
$G(\gamma^{(i)}_I)$ contains only one cycle. I.e., the structure
of $\gamma^{(i)}_I$ must be such that it enters the cycle $C$ at 
some step, traverses all of its short-cycling steps in the first loop around
$C$, then traverses all of its superfluous steps in additional loops around
$C$, and then exits the cycle $C$ and
does not return. So in fact, given the total number of superfluous steps, the
vertex index pairs of these superfluous steps are uniquely determined by
repeating those of the short-cycling steps in order. Hence the
above encoding specifies the vertex index pairs of all non-tree-edge steps, so
the encoding is injective.

Finally, we bound the number of different codes under this encoding. For item
(1), there are at most $v^2$ vertex pairs for each $j \notin \bigcup_{I \in
\cI_i} I$, so there are at most
\[(v^2)^{m-|K_1|+\frac{m}{l}+m-|K_2|+\frac{m}{l}}\]
possible codes for item (1). For each sub-path $\gamma^{(i)}_I$
of item (2), note that the edges corresponding to long-cycling steps (which do
not belong to the cycle $C$) are
distinct from those corresponding to short-cycling steps (which belong to
$C$), and each edge corresponding to a long-cycling or short-cycling step is
visited exactly once since $\gamma^{(i)}_I$ contains at most one cycle.
As there are at
most $e-v+1$ distinct non-tree edges and $\gamma^{(i)}_I$ has length at most
$l$, $\gamma^{(i)}_I$ has at most $\min(e-v+1,l)$ total long-cycling steps and
short-cycling steps. Recall that the first non-innovation step after each
leading innovation must be a non-tree-edge step, and note that there cannot be
an innovation between a superfluous step and the next long-cycling step.
Hence between each pair of
successive long-cycling or short-cycling steps, and before the first such step
and after the last such step, there is at most one leading innovation. So
the total number $P$ of leading innovations, short-cycling steps, and
long-cycling steps in $\gamma^{(i)}_I$ is at most $2\min(e-v+1,l)+1$.
The number of superfluous non-tree-edge steps at most $l$. Hence for each
sub-path $\gamma^{(i)}_I$ of item (2), there are at most
$l(3v^2)^{2\min(e-v+1,l)+1}$ possible codes, yielding at most
\[\left(l(3v^2)^{2\min(e-v+1,l)+1}\right)
^{m-|K_1|+\frac{m}{l}+1+m-|K_2|+\frac{m}{l}+1}
\leq \left(l(3v^2)^{2l+1}\right)^{2m-|K_1|-|K_2|}
\left(l(3v^2)^{2e-2v+3}\right)^{\frac{2m}{l}+2}\]
possible codes for item (2) by the above bounds on $|\cI_1|$ and
$|\cI_2|$. Combining these bounds gives the desired result.
\end{proof}

\section{Galton-Watson martingales}\label{appendix:XY}
In this appendix, we prove Lemma \ref{lemma:XYmomentstails}. Let us first
establish the lemma in the case $\delta=1$, i.e.\ all labels are revealed (so
$\sigma_\true(\root)=\sigma(\root)$).

Let $Y_\ell=\sigma(\root)D_\ell$, and let
$\{\cF_\ell\}_{\ell=0}^\infty$ be the filtration such that $\cF_\ell$ is the
sigma-field generated by $B_\ell(\root;T)$ and the labels in this ball.
For each vertex $v \neq \root$ in $T$,
denote by $p(v)$ its parent vertex in $T$ and by $k(v)$ its distance
to $\root$. Fix $\ell \geq 1$, let $J_1,J_2$ be the
numbers of vertices $v$ at generation $\ell-1$ with $\sigma(v)=\sigma(\root)$
and $\sigma(v)=-\sigma(\root)$, respectively, and let $K_1,K_2,K_3,K_4$
be the numbers of vertices $v$ at generation $\ell$ with $\sigma(v)=\sigma(p(v))=\sigma(\root)$, $-\sigma(v)=\sigma(p(v))=\sigma(\root)$,
$\sigma(v)=-\sigma(p(v))=\sigma(\root)$, and
$-\sigma(v)=-\sigma(p(v))=\sigma(\root)$, respectively.
Then $X_{\ell-1}=d^{-(\ell-1)}(J_1+J_2)$,
$X_\ell=d^{-\ell}(K_1+K_2+K_3+K_4)$, $Y_{\ell-1}=\mu^{-(\ell-1)}(J_1-J_2)$, and
$Y_\ell=\mu^{-\ell}(K_1-K_2+K_3-K_4)$. Conditional on $\cF_{\ell-1}$,
$K_1,K_2,K_3,K_4$ are independent and distributed as 
\[K_1\sim \Pois(\tfrac{a}{2}J_1),\;\;
K_2\sim \Pois(\tfrac{b}{2}J_1),\;\;
K_3\sim \Pois(\tfrac{b}{2}J_2),\;\;
K_4\sim \Pois(\tfrac{a}{2}J_2).\]
Hence for any $t,s \in \R$,
\begin{align}
&\E[\exp(tX_\ell+sY_\ell) \mid \cF_{\ell-1}]\nonumber\\
&=\E\left[\exp\left(
\left(\frac{t}{d^\ell}+\frac{s}{\mu^\ell}\right)K_1+
\left(\frac{t}{d^\ell}-\frac{s}{\mu^\ell}\right)K_2+
\left(\frac{t}{d^\ell}+\frac{s}{\mu^\ell}\right)K_3+
\left(\frac{t}{d^\ell}-\frac{s}{\mu^\ell}\right)K_4\right) \Bigg|
\cF_{\ell-1}\right]
\nonumber\\
&=\exp\left(
\frac{a}{2}J_1(e^{\frac{t}{d^\ell}+\frac{s}{\mu^\ell}}-1)+
\frac{b}{2}J_1(e^{\frac{t}{d^\ell}-\frac{s}{\mu^\ell}}-1)+
\frac{b}{2}J_2(e^{\frac{t}{d^\ell}+\frac{s}{\mu^\ell}}-1)+
\frac{a}{2}J_2(e^{\frac{t}{d^\ell}-\frac{s}{\mu^\ell}}-1)\right)\nonumber\\
&=\exp\left(\frac{d}{2}\left(e^{\frac{t}{d^\ell}+\frac{s}{\mu^\ell}}
+e^{\frac{t}{d^\ell}-\frac{s}{\mu^\ell}}-2\right)(J_1+J_2)
+\frac{\mu}{2}\left(e^{\frac{t}{d^\ell}+\frac{s}{\mu^\ell}}
-e^{\frac{t}{d^\ell}-\frac{s}{\mu^\ell}}\right)(J_1-J_2)\right)
\nonumber\\
&=\exp\left(d^\ell\left(\exp\left(\frac{t}{d^\ell}\right)\cosh
\left(\frac{s}{\mu^\ell}
\right)-1\right)X_{\ell-1}+\mu^\ell\exp\left(\frac{t}{d^\ell}\right)
\sinh\left(\frac{s}{\mu^\ell}\right)Y_{\ell-1}\right).
\label{eq:conditionaljointmgf}
\end{align}
Denote the joint moment generating function of $(X_\ell,Y_\ell)$
by $M_\ell(t,s)=\E[\exp(tX_\ell+sY_\ell)]$.
Taking the full expectation of (\ref{eq:jointmgf}), for each $\ell \geq 1$,
\begin{equation}\label{eq:jointmgf}
M_\ell(t,s)=M_{\ell-1}\left(d^\ell\left(\exp\left(\frac{t}{d^\ell}\right)
\cosh\left(\frac{s}{\mu^\ell}\right)-1\right),\;
\mu^\ell\exp\left(\frac{t}{d^\ell}\right)\sinh\left(\frac{s}{\mu^\ell}\right)
\right).
\end{equation}
In particular, induction on $\ell$ shows $M_\ell(t,s)$ is finite for all
$t,s \in \R$ and $\ell \geq 0$.

Differentiating (\ref{eq:conditionaljointmgf}) at $(0,0)$,
\begin{align*}
\E[X_\ell \mid \cF_{\ell-1}]&=\tfrac{\partial}{\partial t}
\E[\exp(tX_\ell+sY_\ell) \mid \cF_{\ell-1}]_{t=0,s=0}=X_{\ell-1},\\
\E[Y_\ell \mid \cF_{\ell-1}]&=\tfrac{\partial}{\partial s}
\E[\exp(tX_\ell+sY_\ell) \mid \cF_{\ell-1}]_{t=0,s=0}=Y_{\ell-1},\\
\E[X_\ell^2 \mid \cF_{\ell-1}]&=\tfrac{\partial^2}{\partial t^2}
\E[\exp(tX_\ell+sY_\ell) \mid \cF_{\ell-1}]_{t=0,s=0}=X_{\ell-1}^2+d^{-\ell}
X_{\ell-1},\\
\E[Y_\ell^2 \mid \cF_{\ell-1}]&=\tfrac{\partial^2}{\partial s^2}
\E[\exp(tX_\ell+sY_\ell) \mid
\cF_{\ell-1}]_{t=0,s=0}=Y_{\ell-1}^2+d^\ell\mu^{-2\ell}X_{\ell-1}.
\end{align*}
Then $\{X_\ell\}$ and $\{Y_\ell\}$ are martingales with respect to $\{\cF_\ell\}$,
with $\E[X_\ell]=X_0=1$, $\E[Y_\ell]=Y_0=1$,
\[\E[X_\ell^2]=\E[X_{\ell-1}^2]+d^{-\ell}=\ldots=\sum_{k=0}^\ell d^{-k},\;\;\;\;
\E[Y_\ell^2]=\E[Y_{\ell-1}^2]+d^\ell\mu^{-2\ell}=\ldots=\sum_{k=0}^\ell
d^k\mu^{-2k}.\]
Hence $\{X_\ell\},\{Y_\ell\}$ are bounded in $L_2$, so they converge a.s.\ and
in $L_2$ to some $(X,Y) \in \cF_\infty$ by the martingale convergence theorem.
As $\E[X_\ell^2] \to d/(d-1)$ and $\E[Y_\ell^2] \to \mu^2/(\mu^2-d)$,
(\ref{eq:XYmoments}) follows.

For the tail bounds (\ref{eq:XYtails}), set $\alpha=1/6$ and define
\[T_0=\alpha\sqrt{d-1},\;\;\;\;S_0=\alpha\sqrt{\frac{\mu^2-d}{d}},\]
\[T_\ell=T_0-T_0^2\sum_{k=1}^\ell d^{-k}-S_0^2\sum_{k=1}^\ell d^k\mu^{-2k},\;\;\;\;
S_\ell=S_0-2S_0T_0\sum_{k=1}^\ell d^{-k}-S_0^3\sum_{k=1}^\ell \mu^{-2k}.\]
Note that for $d \geq 2$, $\alpha \leq T_0 \leq \alpha(d-1)$ and
$S_0^2 \leq \alpha^2(\mu^2-1)/2$. Then as $\ell \to \infty$,
\[T_\ell \downarrow T_\infty:=T_0-\frac{T_0^2}{d-1}-\frac{S_0^2d}{\mu^2-d}
=T_0-2\alpha^2 \geq (1-2\alpha)T_0,\]
\[S_\ell \downarrow S_\infty:=S_0-\frac{2S_0T_0}{d-1}-\frac{S_0^3}{\mu^2-1}
\geq (1-2\alpha-\alpha^2/2)S_0.\]

We claim by induction on $\ell$ that
for all $\ell \geq 0$ and $t,s \in \R$ with $|t| \leq T_\ell$ and $|s| \leq
S_\ell$,
\begin{equation}\label{eq:MGFinductivebound}
M_\ell(t,s) \leq \exp\left(t+s+|t|\sqrt{d-1}
\sum_{k=1}^\ell d^{-k}+|s|\sqrt{\frac{\mu^2-d}{d}}
\sum_{k=1}^\ell d^k\mu^{-2k}\right).
\end{equation}
As $X_0=Y_0=1$, (\ref{eq:MGFinductivebound}) holds with equality for $\ell=0$. Let
$\ell \geq 1$ and assume inductively that (\ref{eq:MGFinductivebound}) holds for
$\ell-1$. Let $t_\ell,s_\ell \in \R$ with $|t_\ell| \leq T_\ell$ and $|s_\ell|
\leq S_\ell$, and
denote
\[t_{\ell-1}=d^\ell\left(\exp\left(\frac{t_\ell}{d^\ell}\right)
\cosh\left(\frac{s_\ell}{\mu^\ell}\right)-1\right),\;\;\;\;
s_{\ell-1}=\mu^\ell\exp\left(\frac{t_\ell}{d^\ell}\right)
\sinh\left(\frac{s_\ell}{\mu^\ell}\right).\]
As $|t_\ell| \leq T_0 \leq \alpha d^\ell$ and $|s_\ell| \leq S_0 \leq \alpha\mu
\leq \alpha \mu^\ell$, and $|\exp(x)\cosh(y)-1-x| \leq x^2+y^2$ and
$|\exp(x)\sinh(y)-y| \leq 2|xy|+|y|^3$
for all $|x|,|y| \leq \alpha$, we obtain
\[\left|t_{\ell-1}-t_\ell\right| \leq
\frac{t_\ell^2}{d^\ell}+\frac{s_\ell^2d^\ell}{\mu^{2\ell}}
\leq \frac{T_0^2}{d^\ell}+\frac{S_0^2d^\ell}{\mu^{2\ell}},\;\;\;\;
\left|s_{\ell-1}-s_\ell\right| \leq
\frac{2|s_\ell t_\ell|}{d^\ell}+\frac{|s_\ell|^3}{\mu^{2\ell}}
\leq \frac{2S_0T_0}{d^\ell}+\frac{S_0^3}{\mu^{2\ell}}.\]
This implies $|t_{\ell-1}| \leq T_{\ell-1}$ and $|s_{\ell-1}| \leq S_{\ell-1}$.
Then (\ref{eq:jointmgf}) and the induction hypothesis imply
\[M_\ell(t_\ell,s_\ell)\leq \exp\left(t_{\ell-1}+s_{\ell-1}+|t_{\ell-1}|
\sqrt{d-1}\sum_{k=1}^{\ell-1} d^{-k}
+|s_{\ell-1}|\sqrt{\frac{\mu^2-d}{d}}\sum_{k=1}^{\ell-1}
d^k\mu^{-2k}\right).\]
To complete the induction, it suffices to show
\begin{align}
&|t_{\ell-1}-t_\ell|\left(1+\sqrt{d-1}\sum_{k=1}^{\ell-1} d^{-k}\right)
+|s_{\ell-1}-s_\ell|\left(1+\sqrt{\frac{\mu^2-d}{d}}
\sum_{k=1}^{\ell-1} d^k\mu^{-2k}\right)\nonumber\\
&\hspace{1in}\leq \sqrt{d-1}\frac{|t_\ell|}{d^\ell}
+\sqrt{\frac{\mu^2-d}{d}}\frac{|s_\ell|d^\ell}{\mu^{2\ell}}.
\label{eq:completeinduction}
\end{align}
Recalling $d \geq 2$, we may bound
\begin{align*}
|t_{\ell-1}-t_\ell|&\leq
\frac{t_\ell^2}{d^\ell}+\frac{s_\ell^2d^\ell}{\mu^{2\ell}}
\leq \alpha\sqrt{d-1}\frac{|t_\ell|}{d^\ell}+\alpha\sqrt{\frac{\mu^2-d}{d}}
\frac{|s_\ell|d^\ell}{\mu^{2\ell}},\\
1+\sqrt{d-1}\sum_{k=1}^{\ell-1} d^{-k} &\leq 1+\frac{1}{\sqrt{d-1}}
\leq 2,\\
|s_{\ell-1}-s_\ell|&\leq
\frac{2|s_\ell t_\ell|}{d^\ell}+\frac{|s_\ell|^3}{\mu^{2\ell}}
\leq 2\alpha\sqrt{\frac{\mu^2-d}{d}}\frac{|t_\ell|}{d^\ell}+
\alpha^2\frac{\mu^2-d}{d}\frac{|s_\ell|}{\mu^{2\ell}},\\
1+\sqrt{\frac{\mu^2-d}{d}}\sum_{k=1}^{\ell-1} d^k\mu^{-2k}
&\leq 1+\sqrt{\frac{d}{\mu^2-d}}.
\end{align*}
Combining the above, applying the bounds $1 \leq \sqrt{d-1} \leq d^\ell$ and
$\sqrt{(\mu^2-d)/d} \leq \sqrt{d-1} \leq d^\ell$, and recalling $\alpha=1/6$,
we obtain (\ref{eq:completeinduction}). This completes the induction
and our proof of (\ref{eq:MGFinductivebound}).

Finally, (\ref{eq:MGFinductivebound}) implies, in particular,
\begin{align*}
\E[\exp(tX_\ell)]&=M_\ell(t,0) \leq \exp\left(t+\frac{|t|}{\sqrt{d-1}}\right)\;\;
\forall |t| \leq \frac{\sqrt{d-1}}{9} \leq T_\infty,\\
\E[\exp(sY_\ell)]&=M_\ell(0,s) \leq \exp\left(s+|s|\sqrt{\frac{d}{\mu^2-d}}
\right)\;\;
\forall |s| \leq \frac{1}{10}\sqrt{\frac{\mu^2-d}{d}} \leq S_\infty.
\end{align*}
By Fatou's lemma, the same bounds hold for $\E[\exp(tX)]$ and
$\E[\exp(sY)]$. Choosing $t=\sqrt{d-1}/9$,
\[\P\left[|X-1| \geq \tfrac{\gamma}{\sqrt{d-1}}\right]
\leq e^{-t\frac{\gamma}{\sqrt{d-1}}}\E[e^{t|X-1|}]
\leq e^{-\gamma/9}(\E[e^{t(X-1)}]+\E[e^{t(1-X)}])
\leq 2e^{1/9}e^{-\gamma/9},\]
yielding the bound for $X$ in (\ref{eq:XYtails}). The same argument yields the
bound for $Y$, and this completes the proof of Lemma \ref{lemma:XYmomentstails}
in the case of $\delta=1$.

For $\delta \in (0,1)$, first note that the definition of $X_\ell$ does not
depend on the revealed labels, and hence $X_\ell \to X$ a.s.\ for the same limit
$X$ as in the case $\delta=1$. To show $Y_\ell \to Y$ a.s., denote
\[N_{\ell,\true}^+=|\{v:k(v)=\ell,\sigma_\true(v)=+1\}|,\;\;
N_{\ell,\true}^-=|\{v:k(v)=\ell,\sigma_\true(v)=-1\}|,\]
\[Y_{\ell,\true}=\sigma_\true(\root)\mu^{-\ell}(N_{\ell,\true}^+
-N_{\ell,\true}^-)\]
where $\sigma_\true(v) \in \{+1,-1\}$ denotes the `true label' of each vertex
$v$ (i.e.\ the vertex set of the hidden
partition containing $v$). Then the $\delta=1$ case implies
$Y_{\ell,\true} \to Y$ a.s. On the event that the tree $T$ goes extinct, we have
$Y_\ell=Y_{\ell,\true}=0$ for all sufficiently large $\ell$, so
$Y_\ell \to Y=0$. On the event that $T$ does not go
extinct, we have $X>0$ a.s., so in particular
$N_{\ell,\true}^++N_{\ell,\true}^- \to \infty$.
As $d>\mu$ and as $Y_{\ell,\true} \to Y<\infty$, this also
implies $N_{\ell,\true}^+/N_{\ell,\true}^- \to 1$, so in fact
$d^{-\ell}N_{\ell,\true}^+ \to X/2$ and $d^{-\ell}N_{\ell,\true}^- \to X/2$.
Conditional on $T$ and the true labels,
$N_\ell^+ \sim \Binom(N_{\ell,\true}^+,\delta)$, so Hoeffding's inequality
implies $|\delta^{-1}N_\ell^+-N_{\ell,\true}^+| \leq (\log N_{\ell,\true}^+)
\sqrt{N_{\ell,\true}^+}$ almost surely for all large $\ell$.
A similar bound holds for $N_\ell^-$, implying that
\[|Y_\ell-Y_{\ell,\true}| \leq \mu^{-\ell}\left((\log N_{\ell,\true}^+)
\sqrt{N_{\ell,\true}^+}+(\log N_{\ell,\true}^-)
\sqrt{N_{\ell,\true}^-}\right)\]
almost surely for all large $\ell$. As $\mu>\sqrt{d}$ and $d^{-\ell}
N_{\ell,\true}^+ \to X/2$ and $d^{-\ell}N_{\ell,\true}^- \to X/2$,
this implies $|Y_\ell-Y_{\ell,\true}| \to 0$.
Hence $Y_\ell \to Y$ a.s.\ also on the event that $T$ does not go extinct.

\section{Proof of Theorem \ref{thm:Harmonic}, Eq. (\ref{eq:ExpansionCond})}
\label{app:ExpansionCond}

Recall that, given an infinite rooted tree $(T,\root)$, we denote by $\cond(T,\root)$ its conductance.
We start by recalling some basic notions that can be found in \cite{lyons1997unsolved}.
Let $\cond^{(\ell)}(T,\root)$ be the conductance of the first $\ell$ generations of $(T,\root)$, 
i.e.\ the intensity of current flowing through the tree when
the root has potential equal to one, and the vertices at generation $\ell$
have potential equal to $0$. By definition,
$\cond^{(\ell)} (T,\root)$ is monotone non-increasing, and $\cond(T,\root) = \cond^{(\infty)}(T,\root)\equiv \lim_{\ell\to\infty}\cond^{(\ell)}(T,\root)$.
We omit the argument and write $\cond^{(\ell)}$, $\cond$ when $(T,\root)$ is a
Galton-Watson tree with $\Pois(d)$ 
offspring distribution. 

By the standard rules for series/parallel combinations of resistances, we get the distributional recursion
\begin{align}
\cond^{(\ell+1)} \ed \sum_{i=1}^L\frac{\cond^{(\ell)}_i}{1+\cond^{(\ell)}_i}\, ,\label{eq:RecursionCond}
\end{align}
where $L\sim\Pois(d)$ is independent of the i.i.d.\ collection $(\cond^{(\ell)}_i)_{i\ge 1}$, $\cond^{(\ell)}_i\ed \cond^{(\ell)}$.
This is to be complemented with the boundary condition $\cond^{(0)}=\infty$ (with the convention that $\infty/(1+\infty)=1$).
The limit conductance $\cond$ is a fixed point of the above recursion.

We start with a simple concentration estimate.
\begin{lemma}\label{lemma:ConcentrCond}
Let $h(x) \equiv (1+x)\log(1+x)-x$. Then for any $\ell\ge 1$, we have
\begin{align}
\P\Big\{\big|\cond^{(\ell)}-\E\cond^{(\ell)}\big|\ge t\Big\}\le 2\, e^{-d\, h(t/d)}\, .\label{eq:CCone}
\end{align}
In particular, for any $M>0$ there exists $d_0(M)$ such that, for all $d\ge d_0(M)$,
\begin{align}
\P\Big\{\big|\cond^{(\ell)}-\E\cond^{(\ell)}\big|\ge \sqrt{4\, M d\log
d}\Big\}\le \frac{2}{d^M}\, . \label{eq:CCtwo}
\end{align}
\end{lemma}
\begin{proof}
Consider a modified random variable $\tcond^{(\ell)}$ defined by
\begin{align}
\tcond^{(\ell)} \ed \sum_{i=1}^m\frac{\cond^{(\ell-1)}_i}{1+\cond^{(\ell-1)}_i} \, B_i\, ,
\end{align}
where $m$ is an integer and $B_i\sim_{iid}$ Bernoulli$(d/m)$. By
Bennet's inequality \cite[Theorem 2.9]{boucheron2013concentration},
we have
\begin{align}
\P\Big\{\big|\tcond^{(\ell)}-\E\tcond^{(\ell)}\big|\ge t\Big\}\le
2 \, e^{-\nu\, h(t/\nu)}\, ,
\end{align}
where $\nu = \sum_{i=1}^m\E(X_i^2)$, $X_i\equiv B_i\cond^{(\ell-1)}_i/(1+\cond^{(\ell-1)}_i)$.
Claim (\ref{eq:CCone}) simply follows because $\tcond^{(\ell)}$ converges to
$\cond^{(\ell)}$ in distribution and in expectation
(for instance by coupling $\Pois(d)$ and $\Binom(m,d/m)$), and  using $\nu\le d$ together with the fact that $\nu\mapsto \nu\, h(t/\nu)$
is monotone non-increasing for all $t\ge 0$. 

Claim (\ref{eq:CCtwo}) follows by using $h(x) \ge x^2/4$ for $x\le 1$. 
\end{proof}

We next estimate the mean and variance of $\cond$.
\begin{lemma}\label{lemma:ExpVar}
With the above definitions, we have  $\E\cond\le d$, $\Var(\cond)\le d$ and, for large $d$,
\begin{align}
\E\cond = d-1+O_d\Big(\sqrt{\frac{\log d}{d}}\Big)\, ,\label{eq:ExpectC}\\
\Var(\cond) = d-O_d(1)\, . \label{eq:VarC}
\end{align}
\end{lemma}
\begin{proof}
First note that, by the recursion (\ref{eq:RecursionCond}), we have
\begin{align}
\E\cond^{(\ell+1)} = d\, \E\Big\{\frac{\cond^{(\ell)}}{\cond^{(\ell)}+1}\Big\}\, ,\;\;\;\;\;
\Var(\cond^{(\ell+1)}) = d\, \E\Big\{\Big(\frac{\cond^{(\ell)}}{\cond^{(\ell)}+1}\Big)^2\Big\}\, . \label{eq:RecursionE}
\end{align}
whence we get immediately $\E\cond\le d$ and $\Var(\cond)\le d$ (note that $\lim_{\ell\to\infty}\E(\cond^{(\ell)})= \E(\cond)$,
$\lim_{\ell\to\infty}\Var(\cond^{(\ell)})= \Var(\cond)$ by dominated
convergence, since $\cond^{(\ell)}$ is dominated by a $\Pois(d)$ random variable).  
 Also, by Jensen's inequality and using the
notation $\bc^{(\ell)}\equiv \E\cond^{(\ell)}$, we get
\begin{align}
\bc^{(\ell+1)} \le d\,\frac{\bc^{(\ell)}}{\bc^{(\ell)}+1}\, .
\end{align}
Iterating this bound from $\bc^{(0)} = \infty$, we obtain $\bc \le d-1$.

In order to prove a lower bound,  define $B_{\ell} = [\bc^{(\ell)}-\sqrt{4Md\log d},\bc^{(\ell)}+\sqrt{4Md\log d}]\equiv [\bc_1^{(\ell)},\bc_2^{(\ell)}]$, 
with $M$ to be fixed below.  We have
\begin{align}
\E\Big\{\frac{\cond^{(\ell)}}{\cond^{(\ell)}+1}\Big\}&\ge
\E\Big\{\frac{\cond^{(\ell)}}{\cond^{(\ell)}+1}\, \1\{\cond^{(\ell)}\in
B_{\ell}\}\Big\}\\
& \ge \frac{\bc^{(\ell)}_1}{\bc^{(\ell)}_1+1}\, \P \big(\cond^{(\ell)}\in B_{\ell}\big)\\
& \ge  \frac{\bc^{(\ell)}-\Delta}{\bc^{(\ell)}-\Delta+1}\, \Big(1-\frac{2}{d^M}\Big)
\end{align}
where $\Delta = \sqrt{4Md\log d}$ and we used (\ref{eq:CCtwo}). Defining
$\tilde{d} = d(1-2d^{-M})$, and 
the function $F(x) = \tilde{d}(x-\Delta)/(x-\Delta+1)$,
we obtain the  lower bounds
\begin{align}
\bc^{(\ell)}\ge F(\bc^{(\ell-1)})\, .
\end{align}
By calculus, we obtain that the fixed point equation $x=F(x)$ has two positive solutions $0<x_0(d)<x_1(d)<\infty$ for all $d$ large enough, with
\begin{align}
x_1(d) &= \frac{\td+\Delta-1}{2}+\sqrt{\frac{(\td+\Delta-1)^2}{4}-\td\Delta}\\
& = d-1+O\Big(\sqrt{\frac{\log d}{d}}\Big)\, ,
\end{align}
where, for the second equality, we used $M\ge 3/2$. Furthermore $x\mapsto F(x)$ is non-decreasing for $x\ge x_1(d)$. Using the initial condition
$\bc^{(0)} = \infty$, this implies $\bc^{(\ell)}\ge x_1(d)$ for all $\ell$, thus finishing the proof of (\ref{eq:ExpectC}).

In order to prove (\ref{eq:VarC}), recall that we already proved
$\Var(\cond)\le d$. Using Jensen's inequality in (\ref{eq:RecursionE}),
we get 
\begin{align}
\Var(\cond)\ge d\,  \E\Big\{\frac{\cond}{\cond+1}\Big\}^2  = d\Big(\frac{\bc}{d}\Big)^2\, ,
\end{align}  
and the claim follows from our estimate of $\bc$.
\end{proof}

We are now in position to prove the Taylor expansion in (\ref{eq:ExpansionCond}).
\begin{proof}[Proof of Theorem \ref{thm:Harmonic}, Eq.\ (\ref{eq:ExpansionCond})]
First we claim that $\Psi(x,y)\in [0,1]$ for all $x,y\ge 0$.
Indeed, it is clear that $\Psi(x,y)\ge 0$. Furthermore, $\lim_{x\to 0}\Psi(x,y) = 1$ for any $y>0$ and
\begin{align}
\frac{\partial \Psi}{\partial x}(x,y) &= -\frac{y(1-\sqrt{K})^2}{2(x+y+xy)^2\sqrt{1+x}}\le 0 \, ,
\end{align}
where $K\equiv (1+x)(1+y)$. This proves the claim $\Psi(x,y) \in [0,1]$.

Next, let $\bc =\E\cond$ and  $B = [\bc-\sqrt{4Md\log d},\bc+\sqrt{4Md\log d}]\equiv [\bc_1,\bc_2]$, with $M$ to be fixed below.
By the above calculation and Lemma \ref{lemma:ConcentrCond}, we get
\begin{align}
d\,\E\Psi(\cond_1,\cond_2) &= 2d\,
\E\Big\{\frac{\cond_1\sqrt{\cond_2+1}}{\cond_1+\cond_2+\cond_1\cond_2}\,
\1\{\cond_1\in B\}
\, \1\{\cond_2\in B\}\Big\} + O(d^{-M+1})\\
&  = 2d\,\E\Big\{\frac{\cond_1\sqrt{\cond_2+1}}{(\cond_1+1)(\cond_2+1)}\,
\1\{\cond_1\in B\}
\, \1\{\cond_2\in B\}\Big\} + O(d^{-3/2})\\
& = 2d\,\E\Big\{\frac{\cond}{\cond+1}\, \1\{\cond\in
B\}\Big\}\E\Big\{\frac{1}{\sqrt{\cond+1}}\, \1\{\cond\in B\}\Big\}+ O(d^{-3/2})\, .
\label{eq:ExpressionForAsymp}
\end{align}
In the second equality we took $M\ge 5/2$ and used the fact that there exists a constant $C=C(M)$ such that
\begin{align}
\left|\frac{\cond_1\sqrt{\cond_2+1}}{\cond_1+\cond_2+\cond_1\cond_2}-\frac{\cond_1\sqrt{\cond_2+1}}{(\cond_1+1)(\cond_2+1)}\right| = 
\left|\frac{\cond_1\sqrt{\cond_2+1}}{(\cond_1+\cond_2+\cond_1\cond_2)(\cond_1+1)(\cond_2+1)}\right| \le C\, d^{-5/2}\, ,
\end{align}
for all $\cond_1,\cond_2\in B$.

We are left with the task of evaluating the two expectations in (\ref{eq:ExpressionForAsymp}).
Consider the first one. We have
\begin{align}
\E\Big\{\frac{\cond}{\cond+1}\, \1\{\cond\in B\}\Big\}& =  \E\Big\{\frac{\cond}{\cond+1}\Big\} +O(d^{-M})\\
& = \frac{\bc}{d} +O(d^{-M})\\
& = 1-\frac{1}{d}+O_d\Big(\frac{(\log d)^{1/2}}{d^{3/2}}\Big)\, . \label{eq:FirstExpextation}
\end{align}
where the first equality follows from Lemma \ref{lemma:ConcentrCond}, the second from (\ref{eq:RecursionE}), and 
the last from Lemma \ref{lemma:ExpVar}.

Next let $f(x) = (1+x)^{-1/2}$. Note that $\sup_{x\in B}|f'''(x)| = O(d^{-7/2})$. Hence by the intermediate value theorem
(for $\xi$ a point in $B$), we have
\begin{align}
\E\big\{f(\cond)\,\1\{\cond\in B\}\big\} & 
= \E\Big\{\Big[f(\bc)
+f'(\bc)\,(\cond-\bc)+\frac{1}{2}f''(\bc)\,(\cond-\bc)^2+\frac{1}{6}f'''(\xi)(\cond-\bc)^3\Big]\,\1\{\cond\in
B\}\Big\}\\
& \stackrel{(a)}{=}\E\Big\{f(\bc) +f'(\bc)\,(\cond-\bc)+\frac{1}{2}f''(\bc)\,(\cond-\bc)^2\Big\} +O(d^{-M+1})+O\left(\frac{(\log d)^{3/2}}{d^2}\right)\\
& \stackrel{(b)}{=}\frac{1}{(1+\bc)^{1/2}} +\frac{3}{8}\frac{1}{(1+\bc)^{5/2}}\Var(\cond) + O\left(\frac{(\log d)^{3/2}}{d^2}\right)\\
& \stackrel{(c)}{=} \frac{1}{d^{1/2}}+\frac{3}{8d^{3/2}} + O\left(\frac{(\log d)^{3/2}}{d^2}\right)\, . \label{eq:SecondExpextation}
\end{align}
Here $(a)$ follows from Lemma \ref{lemma:ConcentrCond} and the above upper bound on $|f'''(x)|$, $(b)$ by taking  $M\ge 3$,
and $(c)$ from Lemma \ref{lemma:ExpVar}.

The proof is completed by substituting the estimates (\ref{eq:FirstExpextation}) and (\ref{eq:SecondExpextation}) in (\ref{eq:ExpressionForAsymp}).
\end{proof}

\bibliographystyle{alpha}
\newcommand{\etalchar}[1]{$^{#1}$}

\end{document}